%% file: ms.tex
\documentclass[USenglish,oneside,twocolumn]{article}
\pdfoutput=1
\usepackage{todonotes}
\usepackage{tikz}
\usepackage{amsthm}
\usepackage{amssymb}
\usepackage{centernot}
\usepackage{seqsplit}
\usepackage{acronym}
\usepackage{placeins}
\usepackage{pifont}
\usepackage{dblfloatfix} 

\usepackage{xfrac}

\usepackage{caption}
\usepackage{subcaption}

\usepackage{footnote}
\makesavenoteenv{tabular}
\makesavenoteenv{table}

\pgfdeclarelayer{bg}    
\pgfsetlayers{bg,main} 
\usepackage[normalem]{ulem}
\usetikzlibrary{positioning}
\usetikzlibrary{arrows}

\usepackage[vlined,ruled]{algorithm2e}

\SetKwProg{uponMes}{Upon query}{}{}
\SetKwProg{initializeState}{initializeState}{}{}
\SetKwProg{Validate}{Validate}{}{}
\SetKwProg{RunProtocol}{Run Protocol}{}{}

\SetKwProg{TemplateAlpha}{$X$}{}{}
\SetKwProg{calcNewState}{calcNewState}{}{}
\SetKwProg{checkFor}{checkFor}{}{}
\SetKwProg{calculateProperties}{calculateProperties}{}{}
\SetKwProg{change}{preprocess}{}{}

\newenvironment{proofsketch}[1][Proof sketch]{%
  \proof[#1]%
}{\endproof}

\definecolor{red}{rgb}{0.83,0,0}
\definecolor{green}{rgb}{0.39,0.83,0}
\definecolor{blue}{rgb}{0,0,0.83}
\definecolor{yellow}{rgb}{0.83,0.67,0}
\definecolor{cyan}{rgb}{0,0.83,0.8}

\usepackage[utf8]{inputenc}
\usepackage[big]{dgruyter_NEW}

\newif\iflong

\makeatletter
\@ifundefined{longVersion}{%
  \longfalse%
}{\longtrue}
\makeatother

\longtrue

\DOI{foobar}

\cclogo{\includegraphics{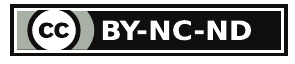}}

\def\inlineheading#1{\textbf{{#1}. }}
\def\inlineheadingTwo#1{\textsc{{#1}. }}
\def\example#1{\textsc{Example: }\emph{#1}}
\def\exampleCont#1{\textsc{Example (cont.): }\emph{#1}}
\def\exampleSpec#1{\textsc{Example: }\emph{#1}}

\input{helper/definitions.tex}

\newtheorem{definition}{Definition}
\newtheorem{lemma}{Lemma}
\newtheorem{theorem}{Theorem}
\newtheorem{translation}{Attack Construction}
\newtheorem{construction}{Protocol Construction}
\newtheorem{claim}{Claim}

 \acrodef{ACN}[ACN]{anonymous communication network}

\begin{document}

  \author*[1]{Christiane Kuhn}

  \author[2]{Martin Beck}

  \author[3]{Stefan Schiffner}

  \author[4]{Eduard Jorswieck}
  
  \author[5]{Thorsten Strufe}

  \affil[1]{TU Dresden, E-mail: christiane.kuhn@tu-dresden.de}

  \affil[2]{TU Dresden, E-mail: martin.beck1@tu-dresden.de}

  \affil[3]{Université du Luxembourg, E-mail: stefan.schiffner@uni.lu}

  \affil[4]{TU Dresden, E-mail: eduard.jorswieck@tu-dresden.de}
  
  \affil[5]{TU Dresden, E-mail: thorsten.strufe@tu-dresden.de}

  \title{\huge On Privacy Notions in Anonymous Communication}
  \runningtitle{On Privacy Notions in Anonymous Communication}

  \input{sections/abstract.tex}

  \journalname{Extended Version}

  \journalyear{..}
  \journalvolume{..}
  \journalissue{..}

\maketitle
\iflong
\emph{
This is the extended version to ``On Privacy Notions in Anonymous Communication'' published at PoPETs 2019.}
\fi

\input{sections/introduction.tex}

\vspace{-0.3cm}
\input{sections/background.tex}

\input{sections/model.tex}

\vspace{-0.5cm}
\input{sections/properties.tex}

\input{sections/notions.tex}

\input{sections/choiceOfNotions.tex}
\input{sections/application.tex}

\iflong
\input{sections/options.tex}

\input{sections/adversaryModel.tex}

\fi

\input{sections/otherFrameworks.tex}
\input{sections/hierarchy.tex}
\iflong
\input{sections/howToUse.tex}

\fi

\input{sections/discussion.tex}

\input{sections/conclusion.tex}

\bibliographystyle{abbrvnat}
\bibliography{framework}

\appendix
\input{sections/appendix/challenger.tex}
\iflong
\input{sections/appendix/pseudocode.tex}

\fi

\iflong
\else
\input{sections/appendix/epsilonAchieving.tex}
\input{sections/appendix/moreExamples.tex}
\fi

\input{sections/appendix/namingScheme.tex}

\input{sections/appendix/symbolList.tex}

\iflong
\else
\vspace{-0.5cm}
\section{Proof Sketches}
Here we include the proof sketches mentioned before.
\subsection{For Implication Completeness}
\vspace{-1em}
\label{completeSketch}
\input{sections/proofs/completeSketch.tex}

\vspace{-1em}
\subsection{For Notions of other Frameworks}
\vspace{-1em}
\label{OtherFrameworksProof}

\begin{figure*}[h!] 
  \center
  \includegraphics[width=0.95\textwidth]{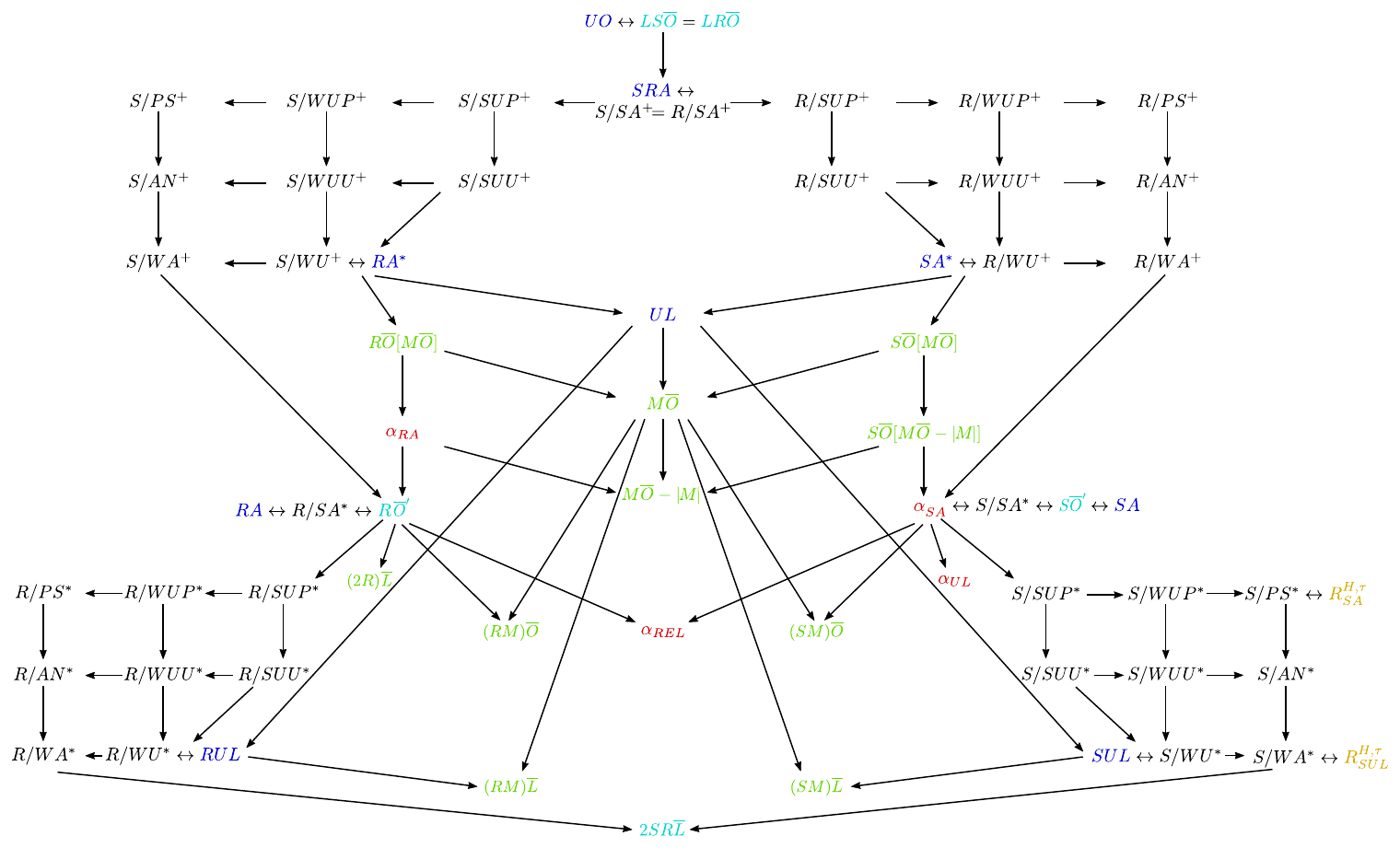}

\caption{\mbox{Our hierarchy with the mapping of the other works (Bohli's, \textcolor{red}{AnoA}, \textcolor{blue}{Hevia's}, \textcolor{yellow}{Gelernter's framework}, \textcolor{cyan}{Loopix's \ac{ACN}} and  \textcolor{green}{new notions})}}

 \label{fig:hierarchyold}
\end{figure*}

\input{sections/proofs/otherFrameworksSketch.tex}

\vspace{-1.5em}
\subsection{For Loopix's Notions 1}
\vspace{-1em}
\label{Loopix1Sketch}
\input{sections/proofs/Loopix1Sketch.tex}
\vspace{-1.5em}
\subsection{For Loopix's Notions 2}
\vspace{-1em}
\label{Loopix2Sketch}

\input{sections/proofs/Loopix2Sketch.tex}
\fi

\iflong
\input{sections/appendix/HierarchyAndTables.tex}

\fi

\input{sections/appendix/graphicOtherFrameworks.tex}

\end{document}

%% file: helper/definitions.tex
\def\nothing {$\aleph$}
\def\something {$%
  \mathrel{\ooalign{\hss$\Diamond$\hss\cr%
  \kern-1pt\raise0ex\hbox{\scalebox{1}{$\not$}}}}$}
\def\challengeRows#1{{#1}_{CR_{\#cr}}}
\def\static#1{{#1}_{c^-}}
\def\corrStandard#1{{#1}_{c^e}}
\def\noCorr#1{{#1}_{c^0}}
\def\corrNoComm#1{{#1}_{c^\text{sr}}}
\def\corrOnlyPartnerSender#1{{#1}_{c^s}}
\def\corrOnlyPartnerReceiver#1{{#1}_{c^r}}

\def\manySess#1{{#1}_{s}}
\def\rgame#1{R\overline{O}\{{#1}\}}
\def\sgame#1{S\overline{O}\{{#1}\}}

\def\EveryButSender {E_{S}}
\def\EveryButRec {E_{R}}
\def\EveryButSenderMsg {E_{SM}}
\def\EveryButMsg {E_{M}}
\def\EveryButSenderRec {E_{SR}}
\def\EveryButReceiverMsg {E_{RM}}
\def\EqualOrNothing {E_{\Diamond}}

\def \bohliPrefixS#1{\sgame{#1}}

\def \bohliSA{\overline{O}}
\def \bohliRSUP{S\overline{O}\{R\overline{O}-|U'|\}}
\def \bohliRWUP {S\overline{O}\{R\overline{O}-H'\}}
\def \bohliRPS{S\overline{O}\{R\overline{O}-P'\}}
\def \bohliRSUU {S\overline{O}\{RF\overline{L}\}}
\def \bohliRWUU {S\overline{O}\{RF\overline{L}-H'\}}
\def \bohliRAN {S\overline{O}\{RF\overline{L}-P'\}}
\def \bohliRWU{S\overline{O}\{RM\overline{L}\}}
\def \bohliRWA {S\overline{O}\{RM\overline{L}-P'\}}
\def \bohliSSUP{R\overline{O}\{S\overline{O}-|U|\}}
\def \bohliSWUP{R\overline{O}\{S\overline{O}-H\}}
\def \bohliSPS{R\overline{O}\{S\overline{O}-P\}}
\def \bohliSSUU{R\overline{O}\{SF\overline{L}\}}
\def \bohliSWUU{R\overline{O}\{SF\overline{L}-H\}}
\def \bohliSAN{R\overline{O}\{SF\overline{L}-P\}}
\def \bohliSWU {R\overline{O}\{SM\overline{L}\}}
\def \bohliSWA{R\overline{O}\{SM\overline{L}-P\}}
\def \bohliRSAc{R\overline{O}}
\def \bohliRSUPc{R\overline{O}-|U'|}
\def \bohliRWUPc{R\overline{O}-H'}
\def \bohliRPSc{R\overline{O}-P'}
\def \bohliRSUUc{RF\overline{L}}
\def \bohliRWUUc{RF\overline{L}-H'}
\def \bohliRANc {RF\overline{L}-P'}
\def \bohliRWUc{RM\overline{L}}
\def \bohliRWAc{RM\overline{L}-P'}
\def \bohliSSAc{S\overline{O}}
\def \bohliSSUPc {S\overline{O}-|U|}
\def \bohliSWUPc{S\overline{O}-H}
\def \bohliSPSc {S\overline{O}-P}
\def \bohliSSUUc {SF\overline{L}}
\def \bohliSWUUc{SF\overline{L}-H}
\def \bohliSANc {SF\overline{L}-P}
\def \bohliSWUc{SM\overline{L}}
\def \bohliSWAc{SM\overline{L}-P}
\def\heviaUO{C\overline{O}}
\def\heviaUL{M\overline{O}[M\overline{L}]}

\def\anoaRA {R\overline{O}[M\overline{O}-|M|]}
\def\anoaRAWOL {R\overline{O}[M\overline{O}]}
\def\anoaREL{(SR)\overline{O}}
\def\anoaUL{(2S)\overline{L}}
\def\gelernterSA{{{R^{H, \tau}_{SA}}} }
\def\gelernterSUL{{{R^{H, \tau}_{S\overline{L}}}}}
\def\newCONF {M\overline{O}-|M|}
\def\newCONFWOL {M\overline{O}}
\def\newSAUO {(SM)\overline{O}}
\def\newRAUO {(RM)\overline{O}}
\def\newSA {(SM)\overline{L}}
\def\newRA {(RM)\overline{L}}
\def \loopixSUO{LS\overline{O}}
\def \loopixRUO{LR\overline{O}}
\def \loopixSUONe{S\overline{O}'}
\def \loopixRUONe{R\overline{O}'}
\def \loopixSRTPU{(SR)\overline{L}}
\def\newAnoaSA {S\overline{O}[M\overline{O}-|M|]}
\def\newAnoaSAWOL {S\overline{O}[M\overline{O}]}
\def\newAnoaUL{(2R)\overline{L}}

\def \bohliSALong{Unobservability}
\def \bohliRSUPLong{Sender/Message Unobservability with Receiver Unobservability leaking User Number}
\def \bohliRWUPLong{Sender/Message Unobservability with Receiver Unobservability leaking Histogram}
\def \bohliRPSLong{Sender/Message Unobservability with Receiver Unobservability leaking Pseudonym}
\def \bohliRSUULong{Sender/Message Unobservability with Receiver-Frequency Unlinkability}
\def \bohliRWUULong{Sender/Message Unobservability with Receiver-Frequency Unlinkability leaking Histogram}
\def \bohliRANLong{Sender/Message Unobservability with Receiver-Frequency Unlinkability leaking Pseudonym}
\def \bohliRWULong{Sender/Message Unobservability with Receiver-Message Unlinkability}
\def \bohliRWALong{Sender/Message Unobservability with Receiver-Message Unlinkability leaking Pseudonym}

\def \bohliSSAcLong{Sender Unobservability}
\def \bohliSSUPcLong{Sender Unobservability leaking User Number}
\def \bohliSWUPcLong{Sender Unobservability leaking Histogram}
\def \bohliSPScLong{Sender Unobservability leaking Pseudonym}
\def \bohliSSUUcLong{Sender-Frequency Unlinkability}
\def \bohliSWUUcLong{Sender-Frequency Unlinkability leaking Histogram}
\def \bohliSANcLong{Sender-Frequency Unlinkability leaking Pseudonym}
\def \bohliSWUcLong{Sender-Message Unlinkability}
\def \bohliSWAcLong{Sender-Message Unlinkability leaking Pseudonym}
\def\heviaUOLong{Communication Unobservability}
\def\heviaULLong{Message Unobservability with  Message Unlinkability}

\def\anoaRELLong{Sender-Receiver Unobservability}
\def\anoaULLong{Twice Sender Unlinkability}

\def\newCONFLong {Message Unobservability leaking Message Length}
\def\newCONFWOLLong {Message Unobservability}
\def\newSAUOLong {Sender-Message Pair Unobservability}

\def\newSALong {Sender-Message Pair Unlinkability}

\def \loopixSUONeLong{Restricted Sender Unobservability}
\def \loopixSRTPULong{Sender-Receiver Pair Unlinkability}
\def\newAnoaSALong {Sender Unobservability with Message Unobservability leaking Message Length}
\def\newAnoaSAWOLLong {Sender Unobservability with Message Unobservability}

%% file: sections/abstract.tex
  \begin{abstract}
{
\iflong
Many anonymous communication networks (ACNs) with different privacy goals have been developed.  However, there are no accepted formal definitions of privacy and  ACNs often define their goals and adversary models ad hoc.
However, for the understanding and comparison of different flavors of privacy, a common foundation is needed.
  In this paper, we introduce an analysis framework for ACNs that captures the notions and assumptions known from different analysis frameworks.
Therefore, we formalize privacy goals as notions and  identify their building blocks. For any pair of notions we prove whether one is strictly stronger, and, if so, which. Hence, we are able to present a complete hierarchy.
Further, we show how to add practical assumptions, e.g. regarding the protocol model or  user corruption as options to our notions. This way, we capture the notions and assumptions of, to the best of our knowledge, all existing analytical frameworks for ACNs and are able to  revise inconsistencies between them. Thus, our new framework builds a common ground and allows for sharper analysis, since new combinations of assumptions are possible and the relations between the notions are known.
\else
Many anonymous communication networks (ACNs) with different privacy goals have been developed.  Still, there are no accepted formal definitions of privacy goals,  and  ACNs often define their goals ad hoc.
However, the formal definition of privacy goals benefits the understanding and comparison of different flavors of privacy and, as a result, the improvement of ACNs.
  In this paper, we work towards defining and comparing privacy goals by formalizing them as privacy notions and identifying their building blocks.
  For any pair of notions we prove whether one is strictly stronger, and, if so, which. Hence, we are able to present a complete hierarchy.
  Using this rigorous comparison between notions, we revise inconsistencies between the existing works and improve the understanding of privacy goals. \iflong Further, we extend the unified definition of privacy notions with options that capture the different assumptions of, to the best of our knowledge all existing analytical frameworks based on indistinguishability games. This new framework allows for sharper analysis, since new combinations of assumptions are possible and the relations among notions are known.\fi  \fi }
\end{abstract}

%% file: sections/introduction.tex


\section{Introduction}
\label{intro}

 
With our frequent internet usage of, e.g., social networks, instant messaging, and web browsing, we constantly reveal personal data. Content encryption can reduce the footprint, but metadata (e.g. correspondents' identities) still leaks. 
To protect metadata from state and industrial surveillance, a broad variety of \acp{ACN} has emerged; one of the most deployed is Tor \cite{dingledine04tor}, but also others, e.g. I2P \cite{zantout11i2p} or Freenet \cite{clarke01freenet}, are readily available. Additionally, many conceptual systems, like Mix-Nets~\cite{chaum81untraceable}, DC-Nets~\cite{chaum88dining}, Loopix \cite{piotrowska17loopix} and Crowds \cite{reiter98crowds} have been published. 


The published \acp{ACN} address a variety of privacy goals. However, many definitions of privacy goals are ad hoc and created for a particular use case. We believe that a solid foundation for future analysis is still missing. This hinders the understanding and comparison of different privacy goals and, as a result, comparison and improvement of \acp{ACN}.
In general, comparing privacy goals is difficult since their formalization is often incompatible and their naming confusing.
This has contributed to a situation where existing informal comparisons disagree: e.g., Sender Unlinkablity of Hevia and Micciancio's framework~\cite{hevia08indistinguishability} and Sender Anonymity of AnoA~\cite{backes17anoa} are both claimed to be equivalent to Sender Anonymity of Pfitzmann and Hansen's terminology~\cite{pfitzmann10terminology}, but significantly differ in the protection they actually provide. These naming issues further complicate understanding of privacy goals and hence  analysis of \acp{ACN}.

To allow rigorous analysis, i.e. provable privacy, of \acp{ACN}, their goals need to be unambiguously defined. Similar to the notions of semantic security (like CPA, CCA1, CCA2 \cite{goos_relations_1998}) for confidentiality, privacy goals can be formally defined as indistinguishability games. 
We call such formally defined privacy goals \emph{privacy notions}. Further, notions need to be compared according to their strength:
achieving the stronger notion implies the weaker one. Comparison of notions, and of the \acp{ACN} achieving them, is otherwise impossible. 
To understand the  ramifications of privacy goals, we aim at setting all notions into mutual relationships.  This means for every pair of notions it must be clear if one is stronger or weaker than the other, or if they have no direct relationship. Such a comparison has already been made for the notions of semantic security~\cite{goos_relations_1998}. 
\iflong
Further, all the assumptions of different existing analysis frameworks, e.g. regarding corruption or specific protocol parts like sessions, have to be unified in one framework to find a common basis for the comparison. 
\fi

\iflong
In this work, we introduce such a unified framework.
\else
In this work, we tackle the formal definition and comparison of privacy goals.
\fi
 To achieve this, we build on the foundations of existing analytical frameworks~\cite{backes17anoa,hevia08indistinguishability,gelernter13limits,bohli11relations}. 
With their preparatory work, we are able to present basic building blocks of privacy notions: observable properties of a communication, that (depending on the notion) must either be protected, i.e. kept private, by the protocol, or are permitted to be learned by the adversary. Defining our notions based on the idea of properties simplifies comparison.
Further, we map practitioners' intuitions to their underlying formal model, justify our choice of notions with exemple use cases for each, and make a sanity check to see that the privacy goals of a current \ac{ACN} (Loopix~\cite{piotrowska17loopix}) are covered. 
\iflong
As a next step, we include assumptions of existing analysis frameworks by defining them similarly as  building blocks that can be combined to any notion. Finally, we argue how the notions and assumptions of existing works map to ours.  
\else
Additionally,  for all formalized goals of existing analysis frameworks~\cite{backes17anoa,hevia08indistinguishability,gelernter13limits,bohli11relations} we reason to which notions they correspond if they are broken down to  the general observable properties and interpreted for \acp{ACN}. 
This means that we focus on general privacy goals and do not present aspects regarding the adversary model, infrequently-used observable information, or the quantification of privacy goals. However, those aspects are compatible with our formalization and have not been ignored; they are presented in the long version of this paper~\cite{longVersion}. 
\fi

We compare all identified privacy notions and present a complete proven hierarchy. As a consequence of our comparison, we are able to rectify mapping inconsistencies of previous work and show how privacy notions and data confidentiality interact.
Furthermore, the proofs for building the hierarchy include templates in order to compare and add new privacy notions to the established hierarchy, if necessary.
\iflong
As we added the assumptions, our resulting framework  captures all the assumptions and notions of the AnoA \cite{backes17anoa}, Hevia and Miccianchio's \cite{hevia08indistinguishability}, Gelernter and Herzberg's \cite{gelernter13limits} frameworks, captures most  and  adapts some of  Bohli and Pashalidis's framework \cite{bohli11relations} and adds missing ones. 
We capture the assumptions and notions of the other frameworks by demonstrating equivalences between their and our corresponding notion.
This removes the constraints of co-existing frameworks and allows to use all options when analyzing an \ac{ACN}. To make our work more accessible, we included a how-to-use section and intuitions, recommendations and limits of this work in the discussion.
\fi
 
\noindent
In summary, our main contributions are:
 
 \begin{itemize}
\iflong
 \item a holistic framework for analyzing \acp{ACN}, capturing more notions and assumptions than each existing framework,
 \fi
 \item the mapping of practitioners' intuitions to game-based proofs, 
 \item the definition of building blocks for privacy notions,
\item the selection and unified definition of notions, 
 \item  a complete hierarchy of privacy notions, which simplifies comparison of \acp{ACN}, \iflong  \else and \fi 
 \item the resolution of inconsistencies and revision of mistakes in previous (frame)works%
 \iflong  
  \item the definition of building blocks for assumptions  compatible to our notions and 
  \item a guide to use the framework and  an example of mapping the goals of an \ac{ACN} into our hierarchy.
  \else .
  \fi
 \end{itemize}

\iflong
\else
 \fi

\inlineheading{Outline}
Section \ref{background} contains an introductory example and gives an overview of our paper.
In Section \ref{sec:GeneralGame}, we introduce the underlying model and indistinguishability games.
In Section \ref{aspects}, we introduce the basic building blocks of privacy notions: properties.
In Section \ref{notions}, we define the privacy notions.
In Section \ref{sec:choiceNotions}, we argue our choice of notions.
\iflong
In Section \ref{sec:options}, we introduce further assumptions, that can be combined with our notions as options.
In Section \ref{sec:adversary}, we explain how results regarding restricted adversaries carry over to our work.
In Section \ref{sec:mappingpapers}, we state the relation of our notions to the other existing analytical frameworks.
\fi
 In Section \ref{sec:hierarchy}, we present the relations between the notions. 
 \iflong
 In Section \ref{howToUse}, we explain how to use the framework for analysis.
 \fi
 In Section \ref{discussion}, we discuss our results.
 In Section \ref{conclusion}, we conclude our paper and give an outlook.

%% file: sections/background.tex
\section{Overview}
\label{background}
\vspace{-0.5cm}
We start with an example of a use case and the corresponding implicit privacy goal, to then introduce the idea of the related indistinguishability game.
We show how such a game works and what it means for a protocol to be secure according to this goal.
Furthermore, by adopting the game we sketch how privacy goals can be formalized as notions and provide an intuition for the relations between different goals.

\example{Alice is a citizen of a repressive regime and  engaged with a resistance group. Despite the regime's sanctions on distributing critical content, Alice wants to publish her latest critical findings.}
A vanilla encryption scheme would reduce Alice's potential audience and thus does not solve her problem. Hence, she needs to hide the link between the message and herself as the sender. We call this goal sender-message unlinkability.\footnote{Usually this is  called sender anonymity. However, since the term sender anonymity is overloaded and sometimes also used with a slightly different meaning, we refer to it as sender-message unlinkability, as the message should not be linkable to the sender.} 

\inlineheading{First attempt}
We start by presenting an easy game, that at first glance looks like the correct formalization for the goal of the example, but turns out to model an even stronger goal.

For Alice's safety, the regime should not suspect her of being the sender of a compromising message, otherwise she risks persecution. Thus, we need to show for the applied protection measure, that compared to any other sender of this message, it is not more probable that Alice is the sender. We analyze the worst case: in a group of users, let Charlie be a user for whom the probability of being the sender differs most from Alice's probability. If even these two are too close to distinguish, Alice is safe, since all other probabilities are closer. Hence, the regime cannot even exclude a single user from its suspects.

We abstract this idea into a game\footnote{Similar to indistinguishability games in cryptology  \cite{goldwasser84probabilistic}.}, where the adversary aims to distinguish two ``worlds'' or scenarios.
These may only differ in the properties the protocol is required to protect, but within these restrictions the adversary can choose freely, especially the worst case that is easiest for her to distinguish (e.g. in one scenario Alice sends the message, in the other Charlie).
Fig. \ref{fig:game1} shows such a game.

\begin{figure}[htbp]
  \centering
  \includegraphics[width=0.4\textwidth]{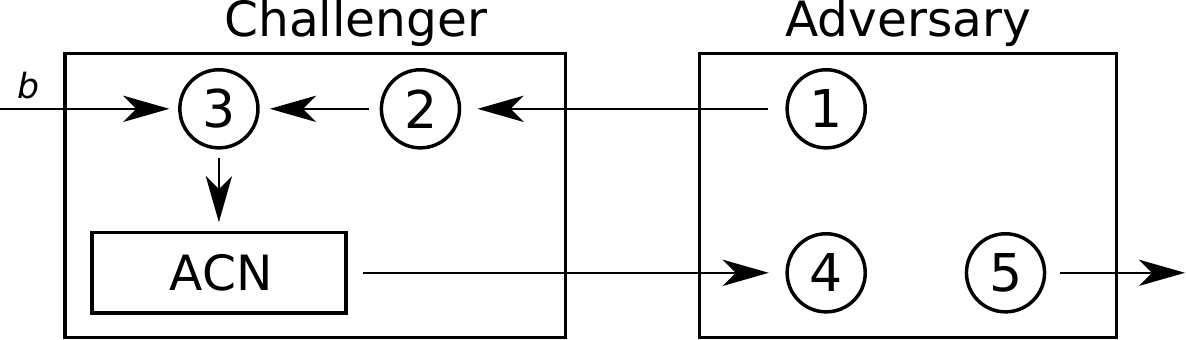}
  \caption{Steps of the sample game: \textbf{1)}~adversary picks two scenarios; \textbf{2)}~challenger checks if scenarios only differ in senders; \textbf{3)}~based on random bit $b$ the challenger inputs a scenario into the \ac{ACN}; \textbf{4)}~adversary observes execution; \textbf{5)}~adversary outputs `guess' as to which scenario was executed}
  \label{fig:game1}
\end{figure}

What the adversary can observe in step 4 depends on her capabilities and area of control. A weak adversary may only receive a message from somewhere, or discover it on a bulletin board. However, a stronger adversary could e.g. also observe the activity on the Internet uplinks of some parties.

The adversary wins the game if she guesses the correct scenario. If she can devise a strategy that allows her to win the game repeatedly with a probability higher than random guessing, she must have learned some information that is supposed to be protected, here the sender (e.g. that Alice is more probable the sender of the message than Charlie), since everything else was identical in both scenarios. Hence, we say that, if the adversary can find such a strategy, we do not consider the analyzed protocol secure regarding the respective privacy goal.
 
 \inlineheading{Why this is too strong}
As argued, a protocol achieving this goal would help Alice in her use case. However,  if an adversary learns who is sending any message with real information (i.e. no random bits/dummy traffic), she can distinguish both scenarios and wins the game.
As an example, consider the following two scenarios: (1) Alice and Bob send messages (2) Charlie and Dave send messages.
If the adversary can learn the active senders, she can distinguish the scenarios and win the game.
However, if she only learns the set of active senders, she may still not know who of the two active senders in the played scenario actually sent the regime-critical content.
Thus, a protocol hiding the information of who sent a message within a set of active senders is good enough for the given example.
Yet, it is considered insecure regarding the above game, since an adversary can learn the active senders.
Hence, the game defines a goal stronger than the required sender-message unlinkability.
As the \ac{ACN} in this case needs to hide the sending activity (the adversary does not know if a certain possible sender was active or not), we call the goal that is actually modeled sender unobservability.

\inlineheading{Correcting the formalization}
However, we can adjust the game of Fig. \ref{fig:game1} to model sender-message unlinkability. We desire that the only information about the communications that differs between the scenarios is who is sending which message. Thus, we allow the adversary to pick scenarios that differ in the senders, but not in the activity of the senders, i.e. the number of messages each active sender sends.
This means, we change what the adversary is allowed to submit in step 1 and what the challenger checks in step 2. So, if the adversary now wants to use Alice and Charlie, she has to use both in both scenarios, e.g. (1) Alice sends the critical message, Charlie a benign message and (2) Charlie sends the critical message, Alice the benign message. 
Hence, given an \ac{ACN} where this game cannot be won, the adversary is not able to distinguish whether Alice  or another active user sent the regime-critical message. The adversary might learn, e.g. that someone sent a regime-critical message and the identities of all active senders (here that Alice and Charlie are active senders). However, since none of this is sanctioned in the above example, Alice is safe, and we say such an \ac{ACN} provides sender-message unlinkability.

%

\inlineheading{Lessons learned}
Depending on the formalized privacy goal (e.g. sender unobservability) the scenarios are allowed to differ in certain properties of the communications (e.g. the active senders) as we have illustrated in two exemple games.
Following the standard in cryptology, we use the term \emph{privacy notion}, to describe such a formalized privacy goal that defines properties to be hidden from the adversary.

Further, the games used to prove the privacy notions only differ in how scenarios can be chosen by the adversary and hence what is checked by the challenger.
This also holds for all other privacy notions; they all define certain properties of the communication to be private and other properties that can leak to the adversary.
Therefore, their respective games are structurally identical and can be abstracted to define one general game, whose instantiations represent notions.
We explain and define this  general game in Section \ref{sec:GeneralGame}.
We then define the properties (e.g. that the set of active senders can change) in Section \ref{aspects} and build notions (e.g. for sender unobservability) upon them in Section \ref{notions}.

Additionally, we already presented the intuition that sender unobservability is stronger than sender-message unlinkability. This is not only true for this example, in fact we prove: every protocol achieving sender unobservability also achieves sender-message unlinkability. Intuitively, if whether Alice is an active sender or not is hidden, whether she sent a certain message or not is also hidden. We will prove relations between our privacy notions in Section \ref{sec:hierarchy} and show that the presented relations (depicted in Figure \ref{fig:hierarchyColored}) are complete. Before that, we argue our choice of notions in Section \ref{sec:choiceNotions}.

\section{Our Game model}
\label{sec:GeneralGame}
Our goal in formalizing the notions as a game is to analyze a given ACN protocol w.r.t. to a notion, i.e. the game is a tool to investigate if an adversary can distinguish two self-chosen, notion-compliant scenarios. Scenarios are sequences of communications.
 A \emph{communication} is described by its sender, receiver, message and auxiliary information (e.g. session identifiers) or the empty communication, signaling that nobody wants to communicate at this point. 
Some protocols might restrict the information flow to the adversary to only happen at specific points in the execution of the protocol,  e.g. because a component of the \ac{ACN} processes a batch of communications before it outputs statistics about them. 
Therefore, we introduce \emph{batches} as a sequence of communications, which is processed as a unit before the adversary observes anything\footnote{We use the word batch to designate a bunch of communications. Besides this similarity, it is not related to batch mixes.}.
When this is not needed, batches can always be replaced with single communications.

As explained in Section \ref{background}, we do not need to define a complete new game for every privacy goal, since notions only vary in the difference between the alternative scenarios chosen by the adversary. Hence, for a given ACN and notion, our general game is simply instantiated with a model of the \ac{ACN}, which we call the protocol model, and the notion. 
 The protocol model accepts  a sequence of communications as input. Similar to the real implementations the outputs of the protocol model are the observations the real  adversary can make. 
Note, the adversaries in the game and the real world have the same capabilities\footnote{A stronger game adversary also implies that the protocol is safer in the real world.}, but differ in their aims: while the real world adversary aims to find out something about the users of the system, the game adversary merely aims to distinguish the two scenarios she has constructed herself. 

In the simplest version of the game, the adversary constructs two scenarios, which are just two batches of communications and sends them to the challenger. The challenger checks that the batches are compliant with the notion. If so, the challenger tosses a fair coin to randomly decide which of the two batches it executes with the protocol model. The protocol model's output is returned to the game adversary. Based on this information, the game adversary  makes a guess about the outcome of the coin toss.

We extend this simple version of the game, to allow the game adversary to send multiple times two batches to the challenger. However, the challenger performs a single coin flip and sticks to this scenario for this game, i.e. it always selects the batches corresponding to the initial coin flip. This allows analyzing for adversaries, that are able to base their next actions in the attack on the observations they made previously. 

\iflong
Further, we allow for user (a sender or receiver) corruption, i.e. the adversary learns the user's momentary internal state, by sending corrupt queries to the challenger. Note that although the adversary decides on all the communications that happen in the alternative scenarios, she does not learn secret keys or randomness unless the user is corrupted. This allows to add several options for different corruption models to the privacy goals.
\fi

 
%

\iflong
To model all possible attacks,
\else
To unfetter our general game from the concrete adversary model,
\fi we allow the adversary to send protocol queries. This is only a theoretical formalization to reflect what information the adversary gets and what influence she can exercise. These protocol query messages are sent to the protocol model without any changes by the challenger. The protocol model covers the adversary to ensure that everything the real world adversary can do is possible in the game with some query message. For example, protocol query messages can be used to add or remove nodes from the \ac{ACN} by sending the appropriate message. 

As introduced in Section \ref{background}, we say that an adversary has an advantage in winning the game, if she guesses the challenger-selected scenario correctly with a higher probability than random guessing.
A protocol achieves a certain privacy goal, if an adversary has at most negligible advantages in winning the game. 



\subsection*{Formalization}
In this subsection, we formalize the game model to conform to the above explanation.

We use $\Pi$ to denote the analyzed \textit{\ac{ACN} protocol model}, $Ch$ for the challenger and $\mathcal{A}$ for the adversary, which is a probabilistic polynomial time algorithm.
Additionally, we use $X$ as a placeholder for the specific notion, e.g. sender unobservability, if we explain or define something for all the notions. 
 A \textit{communication} $r$  in $\Pi$ is represented by a tuple $(u,u',m,aux)$ with a sender $u$, a receiver $u'$, a message $m$,  and auxiliary information $aux$ (e.g. session identifiers).
Further, we use $\Diamond$ instead of the communication tuple $(u,u',m,aux)$ to represent that no communication occurs. 
%
Communications are clustered into \textit{batches} $\underline{r}_b=(r_{b_1},\dots, r_{b_l})$, with $r_{b_i}$ being the $i$-th communication of batch $\underline{r}_b$. Note that we use $\underline{r}$ (underlined) to identify batches and $r$ (no underline) for single communications. Batches in turn are clustered into \textit{scenarios}; the first scenario is $(\underline{r}_{0_1}, \dots,\underline{r}_{0_k} )$. A \textit{challenge}  is defined as the tuple of two scenarios $\left ( (\underline{r}_{0_1}, \dots,\underline{r}_{0_k} ), (\underline{r}_{1_1}, \dots,\underline{r}_{1_k} ) \right )$. 
All symbols used so far and those introduced later are summarized in \mbox{Tables~\ref{tab:allNotions}~--~\ref{tab:allSymbols}}  in Appendix~\ref{sec:allSymbols}.

 \inlineheading{Simple Game}
 \begin{enumerate}
 \item $Ch$ randomly picks challenge bit $b$.
\item $\mathcal{A}$ sends a batch query, containing $\underline{r}_{0}$ and $\underline{r}_{1}$, to $Ch$. 
\item $Ch$ checks if the query is valid,  i.e. both batches differ only in information that is supposed to be protected according to the analyzed notion $X$. 
\item If the query is valid, $Ch$ inputs the batch corresponding to $b$ to $\Pi$. 
\item $\Pi$'s output $\Pi(\underline{r}_{b})$ is handed to $\mathcal{A}$. 
\item After processing the information, $\mathcal{A}$ outputs her guess $g$ for $b$.
\end{enumerate}

\inlineheading{Extensions}
As explained above, there are useful extensions we make to the simple game:
\begin{description}
\item[\emph{Multiple Batches}] Steps 2-5 can be repeated.
\iflong
\item[\emph{User corruption}] Instead of Step 2, $\mathcal{A}$ can also decide to issue a corrupt query specifying a user $u$ and receive $u$'s internal state as output. This might change $\Pi$'s state, lead to different behavior of $\Pi$ in following queries and yield a higher advantage in guessing than before.
\fi
\item[\emph{Other parts of the adversary model}] Instead of Step 2, $\mathcal{A}$ can also decide to issue a protocol query, containing an input specific to $\Pi$ and receive $\Pi$'s output to it (e.g. the internal state of a router that is corrupted in this moment). This might change $\Pi$'s state. 
 \end{description}

\inlineheading{Achieving notion $X$}
Intuitively, a protocol \(\Pi\) achieves a notion \(X\) if any possible adversary has at most negligible advantage in winning the game.
To formalize the informal understanding of $\Pi$ achieving goal $X$, we need the following denotation. \iflong ${\text{Pr}[g= \langle \mathcal{A} \bigm| Ch(\Pi, X,c,b)\rangle ]}$ \else ${\text{Pr}[g= \langle \mathcal{A} \bigm| Ch(\Pi, X,b)\rangle ] }$ \fi describes the probability that $\mathcal{A}$ \iflong (with at most $c$ challenge rows, i.e. communications differing in the scenarios) \fi outputs $g$, when $Ch$ is instantiated with $\Pi$ and $X$ and the challenge bit was chosen to be $b$. With this probability, achieving a notion translates to Definition \ref{def:achieve}.


\begin{definition}[Achieving a notion $X$]\label{def:achieve}
An \ac{ACN} Protocol $\Pi$ achieves $X$, iff for all probabilistic polynomial time(PPT) algorithms $\mathcal{A}$ there exists a negligible $\delta$ such that
\small
\[ \bigm| \text{Pr}[0= \langle \mathcal{A} \mid Ch(\Pi, X,\iflong c, \fi 0)\rangle ]  -
\text{Pr}[0= \langle \mathcal{A} \mid Ch(\Pi, X,\iflong c, \fi  1)\rangle]\bigm|\leq \delta \text{.} \]
\end{definition}

\iflong
We use a variable \(\delta\), which is referred to as negligible, as an abbreviation when we actually mean a function \(\delta(\kappa)\) that is negligible in a security parameter \(\kappa\).
\fi
\iflong
\paragraph{Equivalence to Other Definitions}
Notice, that this definition is equivalent to 
\begin{align*}
 \text{(1) Pr}[0= \langle \mathcal{A} \bigm| Ch(\Pi, X,c,0)\rangle ]  &\leq \\
\text{Pr}[0= \langle \mathcal{A} \bigm| Ch(\Pi, X, c,1)\rangle]&+  \delta \text{.}
\end{align*}
and 
\begin{align*}
 \text{(2) Pr}[1= \langle \mathcal{A} \bigm| Ch(\Pi, X,c,1)\rangle ] &\leq \\
\text{Pr}[1= \langle \mathcal{A} \bigm| Ch(\Pi, X, c,0)\rangle]&+ \delta \text{.}
\end{align*}

(1): $|Pr[0\mid 0]- Pr[0 \mid 1]|\leq \delta $ for all $\mathcal{A}$ $\iff (Pr[0\mid 0]- Pr[0 \mid 1]\leq \delta $ for all $\mathcal{A}$) $\land $ $( Pr[0 \mid 1] - Pr[0\mid 0]\leq \delta $ for all $\mathcal{A}$). To every attack $\mathcal{A}$ with $Pr[0 \mid 1] - Pr[0\mid 0]> \delta$, we can construct $\mathcal{A'}$ with $Pr[0\mid 0]- Pr[0 \mid 1]> \delta$.  Since the definition requires the inequality to hold for all attacks, this is enough to prove that (1) implies the original, the other way is trivial. 
This is how we construct it: Given attack $\mathcal{A}$, we  construct $\mathcal{A'}$ by changing the batches of the first with the second scenario. Hence,  $\text{Pr}[0= \langle \mathcal{A} \bigm| Ch(\Pi, X,c,0)\rangle ] =\text{Pr}[0= \langle \mathcal{A'} \bigm| Ch(\Pi, X,c,1)\rangle ] $ and $\text{Pr}[0= \langle \mathcal{A} \bigm| Ch(\Pi, X,c,1)\rangle ] =\text{Pr}[0= \langle \mathcal{A'} \bigm| Ch(\Pi, X,c,0)\rangle ] $.

(2): To every attack $\mathcal{A}$ breaking (1), we can construct one with the same probabilities in (2). Given attacker $\mathcal{A}$, we construct $\mathcal{A'}$ as the one that changes the batches of the first with the second scenario and inverts the output of $\mathcal{A}$. Hence, $\text{Pr}[0= \langle \mathcal{A} \bigm| Ch(\Pi, X,c,0)\rangle ] =\text{Pr}[1= \langle \mathcal{A'} \bigm| Ch(\Pi, X,c,1)\rangle ] $ and $\text{Pr}[1= \langle \mathcal{A} \bigm| Ch(\Pi, X, c,0)\rangle]=\text{Pr}[0= \langle \mathcal{A'} \bigm| Ch(\Pi, X, c,1)\rangle]$. Since we can reverse this operations by applying them again, we can also translate in the other direction.

\paragraph{Differential Privacy based Definition}
For some use cases, e.g. if the court of your jurisdiction requires that the sender of a critical content can be identified with a minimal probability of a certain threshold e.g. 70\%, a non-negligible $\delta$ might be sufficient. Hence, we allow to specify the parameter of $\delta$ and extend it with  the allowed number of challenge rows $c$ to finally include the well-known concept of differential privacy as AnoA does in the following definition:

\begin{definition}[Achieving $(c,\epsilon, \delta)-X$]\label{def:achieveEpsilon} 
An \ac{ACN} protocol $\Pi$ is  $(c,\epsilon, \delta) -X$ with $c>0$, $\epsilon \geq 0$ and $0 \leq \delta \leq 1$, iff for all PPT algorithms $\mathcal{A}$:
\begin{align*}
 \text{Pr}[0= \langle \mathcal{A} \bigm| Ch(\Pi, X,c,0)\rangle ] &\leq \\
e^{\epsilon} \text{Pr}[0= \langle \mathcal{A} \bigm| Ch(\Pi, X, c,1)\rangle]&+ \delta\text{.}
\end{align*}
\end{definition}
Notice that $\epsilon$ describes how close the probabilities of guessing right and wrong have to be. This can be interpreted as the quality of privacy for this notion. While $\delta$ describes the probability with which the $\epsilon$-quality can be violated. Hence, every \ac{ACN} protocol will achieve $(0,1)-X$ for any notion $X$, but this result does not guarantee anything, since with probability $\delta =1$ the $\epsilon$-quality is not met. 

The first variant can be expressed in terms of the second as $\Pi$ achieves $X$, iff $\Pi$ is $(c,0,\delta)-X$ for a negligible $\delta$ and any $c\geq 0$.
\fi

%% file: sections/model.tex



%% file: sections/properties.tex
\section{Protected Properties}
\label{aspects}


We define properties to specify which information about the communication is allowed to be disclosed to the adversary, and which must be protected to achieve a privacy notion, as mentioned in Section \ref{background}.  
We distinguish between simple and complex properties. Simple properties can be defined with the basic game model already introduced, while complex properties require some extensions to the basic model. 

  \begin{table}
\center
\resizebox{0.48\textwidth}{!}{%
  \begin{tabular}{ c p{3.5cm} p{4.8cm} }

Symbol &Description&Translation to Game\\ \hline
$|M|$&Message Length& Messages in the two scenarios always have the same length.\\
$\EveryButSender$&Everything but Senders& Everything except the senders is identical in both scenarios.\\
$\EveryButRec/ \EveryButMsg$&Everything but Receivers/Messages& Analogous\\
$\EveryButSenderMsg$&Everything but Senders and Messages& Everything except the senders and messages is identical in both scenarios.\\
$\EveryButReceiverMsg/ \EveryButSenderRec$&Analogous& Analogous\\
\something &Something is sent& In every communication something must be sent ($\Diamond$ not allowed).\\
\nothing &Nothing&Nothing will be checked; always true.\\
$U/U'$& Active Senders/Receivers& Who sends/receives is equal for both scenarios.\\
$Q/Q'$& Sender/Receiver Frequencies&Which sender/receiver sends/receives how often  is equal for both scenarios.\\
$|U|/|U'|$& Number of Senders/ Receivers& How many senders/receivers communicate is equal for both scenarios.\\
$P/P'$& Message Partitioning per Sender/Receiver& Which messages are sent/received from the same sender/receiver  is equal for both scenarios.\\
$H/H'$& Sender/Receiver Frequency Histograms&How many senders/receivers send/receive how often is equal for both scenarios.\\
\end{tabular}}
 \caption{Simple properties; information about communications that may be required to remain private}
  \label{tab:information}
\end{table}

 \vspace{-0.3cm}
  \subsection{Simple Properties} 
  \vspace{-0.3cm}
  
 We summarize the informal meaning of all simple properties in Table \ref{tab:information} and introduce them in this section.
  
Assume an \ac{ACN}  aims to hide the sender but discloses message lengths to observers. For this case, we specify the property ($|M|$) that the message length must not differ between the two scenarios, as this information must not help the adversary to distinguish which scenario the challenger chose to play.

Next, we might want an \ac{ACN} to protect the identity of a sender, as well as any information about who sent a message, but deliberately disclose which messages are received by which receiver,  
who the receivers are, and potentially other auxiliary information. 
	We hence specify a property ($\EveryButSender$) where only the senders differ between the two scenarios\footnote{$E$ symbolizes that only this property may vary in the two submitted scenarios and everything else remains equal.}, to ensure that the adversary in our game can only win by identifying senders. 
In case the protection of the receiver identities or messages is required, the same can be defined  for receivers ($\EveryButRec$) or messages ($\EveryButMsg$).

Further, we might want the \ac{ACN} to protect senders and also the messages; leaving the receiver and auxiliary information to be disclosed to the adversary. 
This is achieved by specifying a property where only senders and messages differ between the two scenarios and everything else remains equal ($\EveryButSenderMsg$). 
Again, the same can be specified for receivers and messages ($\EveryButReceiverMsg$) or senders and receivers ($\EveryButSenderRec$).

Lastly, \ac{ACN}s might allow the adversary to learn whether a real message is sent or even how many messages are sent. We specify a property (\something) that requires real communications in both scenarios, i.e. it never happens that nothing is sent in one scenario but something is sent in the other. We ensure this by not allowing the empty communication ($\diamond$).

 However, a very ambitious privacy goal might even require that the adversary learns no information about the communication at all (\nothing). In this case, we allow any two scenarios and check nothing.
 
 \inlineheading{Formalizing those Simple Properties} 
 In the following definition all simple properties mentioned so far are formally defined. 
 Therefore, we use $\top$ as symbol for the statement that is always true.
 \begin{definition}[Properties $|M|$, $\EveryButSender$, $\EveryButSenderMsg$, \something, \nothing]\label{def:properties} 
 Let the checked batches be $\underline{r_0},\underline{r_1}$,  which include the communications \mbox{${r_0}_j \in \{ (u_{0_j},u'_{0_j},m_{0_j},aux_{0_j}), \diamond\}$} and \mbox{ ${r_1}_j \in \{ (u_{1_j},u'_{1_j},m_{1_j},aux_{1_j}), \diamond\}$} with $ j \in \{1, \dots l\} $. We say the following properties are met, iff for all $j \in \{1, \dots l\}$:
\begin{align*}
   |M|& : |{m_0}_j|=|{m_1}_j| \\[0.5em] 
  \EveryButSender&: {r_1}_j =(\mathbf{u_{1_j}},u'_{0_j},m_{0_j},aux_{0_j}) \\
   \EveryButRec&: {r_1}_j =(u_{0_j},\mathbf{u'_{1_j}},m_{0_j},aux_{0_j}) \\
    \EveryButMsg &: {r_1}_j =(u_{0_j},u'_{0_j},\mathbf{m_{1_j}},aux_{0_j}) \\[0.5 em]
     \EveryButSenderMsg&: {r_1}_j =(\mathbf{u_{1_j}},u'_{0_j},\mathbf{m_{1_j}},aux_{0_j}) \\
   \EveryButReceiverMsg&: {r_1}_j =(u_{0_j},\mathbf{u'_{1_j}},\mathbf{m_{1_j}},aux_{0_j}) \\
   \EveryButSenderRec &: {r_1}_j =(\mathbf{u_{1_j}},\mathbf{u'_{1_j}},m_{0_j},aux_{0_j}) \\[0.5 em]
   \text{\something} &: \Diamond \not \in \underline{r}_0 \land \Diamond \not \in \underline{r}_1\\
  \text{\nothing}&: \top
\end{align*}
\label{SimpleProp1}
 \end{definition} 
 
\inlineheading{More Simple Properties: Active Users, Frequencies}
The properties of Definition \ref{SimpleProp1} are important to formalize privacy, but are by themselves not sufficient.
Take the \ac{ACN} Tor as an example: While the set of active senders is trivially known to their ISPs and the guard nodes, we still require that the senders are unlinkable with the messages they are sending (and their receivers).
Similarly, the sending (receiving) frequency of a party may be important and is not formalized yet.
To formalize these properties, we use sets that capture which user sent which messages in a certain period, i.e. a batch of communications (and similarly sets to capture which user received which messages). Note that we use primes ($'$) for the corresponding sets and properties of the receivers.

\begin{definition}[Sender-Message Linking]\label{def:senderMessageSet}
 We define the sender-message linkings for scenario $b$ ($L'_{b_i}$ the receiver-message linkings are analogous) as:
 \begin{align*}
  L_{b_i}:= &\{(u,\{m_1,...,m_h\}) \bigm| u \text { sent messages }m_1, \dots , m_h\\
  &\text{ in batch }i\}\text{.}\\
  \end{align*} 
\end{definition}

The sets from Definition \ref{def:senderMessageSet} allow easy identification of who an active sender in this batch was and how often each sent something:

\begin{definition}[Active Sender Set, Frequency Set]\label{UbQb}
Let the current batch be the $k$-th one. 
For $ b \in \{0,1\}$ $U_b,Q_b$  ($U_b',Q_b'$ for $L_b'$) are defined as:
\begin{align*}
 	U_b &:= \{ u \bigm|  (u,M) \in L_{b_k}\}\hfill\\
        Q_b &:= \{ (u,n) \bigm|  (u,M) \in L_{b_k}\land |M|=n\}\\
\end{align*}
\end{definition}

Recall that we currently define properties for ACNs that allow the adversary to learn which senders are active at different times, or the number of messages they send during some periods, while hiding some other properties (e.g. which messages they have sent). 
Hence, with the respective sets for active users and user frequencies defined, we need only to request that they are equal in both scenarios:

\begin{definition}[Properties $U$, $Q$, $|U|$]
We say that the properties $U,Q,|U|$ ($U',Q', |U'|$ analogous) are met, iff:
\[U: U_0=U_1 \hspace{1,5em}  Q: Q_0=Q_1  \hspace{1,5em}  |U|: |U_0|=|U_1| \]
\end{definition}  

\inlineheading{More Simple Properties: Message Partitions, Histograms}
Other interesting properties are which messages came from a given sender and how many senders sent how many messages. If the adversary knows which messages are sent from the same sender, e.g. because of a pseudonym, she might be able to combine information from them all to identify the sender. If she knows how many senders sent how many messages, she knows the sender activity and hence can make conclusions about the nature of the senders. 

As before, we introduce auxiliary variables to formally define these two properties.
We use $M_{b,I}$ to denote the collection of messages that has been sent by the same sender (e.g. linked by a shared pseudonym) in a set of batches, and $M_{b,I,n}$ to denote the union of all these sets of cardinality $n$.
The equality of the properties in the two scenarios must pertain throughout all comparable batches in the scenarios. If this were not true, the inequality would help the adversary to distinguish the scenarios without learning the protected information e.g. identifying the sender.

\begin{definition}[Multi-Batch-Message Linkings]\label{def:messageSets} Let the current batch be the $k$-th, \mbox{$\mathcal{K}:=\{1, \dots, k\}$}, $\mathcal{P}(\mathcal{K})$ the power set of $\mathcal{K}$ and $\mathcal{U}$ the set of all possible senders ($\mathcal{U'}$ receivers).
 For $b \in \{0,1\}$ and $I \in \mathcal{P}(\mathcal{K})$: We define ($M'_{b,I}, M'_{b,I,n}$ for $L'_{b_i}$)
 \begin{itemize}
 \item the multi-batch-message-sender linking:\\ \mbox{$M_{b,I} := \cup_{u\in \mathcal{U}} \{ \cup_{i \in I}\{M| (u,M)\in L_{b_i}\}\}$} and
 \item the cardinality restricted multi-batch-message-sender linking:  $M_{b,I,n}:=\{ M \in M_{b,I} \bigm| |M|=n \}$.
 \end{itemize}
\end{definition}

As before, we define auxiliary variables capturing the information that we want to be equal in both scenarios: We define ordered sets specifying which messages are sent from the same user for any set of batches (Message Partition $P_b$) and how many users sent how many messages for any set of batches (Histogram $H_b$).  Therefore, we use a slightly unusual notation: For any set Z, we use $(Z_i)_{i\in\{1,\dots, k\}}$ to denote the sequence $(Z_1, Z_2, \dots, Z_k)$ and $\overrightarrow{\mathcal{P}}(Z)$ to denote a sorted sequence of the elements of the power set\footnote{For brevity we use $\in$ to iterate through a sequence.} of $Z$.


\begin{definition}[Message partitions, Histograms]\label{def:properties}
 
Consider the $k$-th batch, \mbox{$\mathcal{K}:=\{1, \dots, k\}$}. 
For $ b \in \{0,1\}$  $P_b, H_b$  ($P_b',H_b'$ analogous) are defined as:
 \allowdisplaybreaks
\begin{align*}
	P_b &:= (M_{b,I})_{ I \in \overrightarrow{\mathcal{P}}(\mathcal{K})}\\
	H_b &:= (\{(n, i) \bigm|   i= |M_{b,I,n}|\})_{ I \in \overrightarrow{\mathcal{P}}(\mathcal{K})}\\
\end{align*}

\vspace{-0.5cm}
Further, we say that properties $P,H$ ($P',H' $ analogous) are met, iff:
\begin{align*}
  P&: P_0=P_1 \hspace{3em}  H: H_0=H_1\\
\end{align*}
 \end{definition} 

\subsection{Complex Properties}
\label{sec:ComplexProperties}
 \vspace{-0.3cm}
 
So far, we have defined various properties to protect senders, messages, receivers, their activity, frequency and the grouping of messages. 
However, this is not sufficient to formalize several relevant privacy goals, and we must hence introduce complex properties.
 
 
 \inlineheading{Learning Sender and Receiver}
Consider that one aims to hide which sender is communicating with which receiver.
Early \acp{ACN} like classical Mix-Nets~\cite{chaum81untraceable}, and  also  Tor~\cite{dingledine04tor}, already used this goal. Therefore, we want the adversary to win the game only if she identifies both: sender and receiver of the same communication. 

An intuitive solution may be to model this goal by allowing the adversary to pick different senders and receivers ($\EveryButSenderRec$) in both scenarios (see Fig. \ref{fig:senderReceiver} (a) for an example). 
This, however, does not actually model the privacy goal: by identifying only the sender or only the receiver of the communication, the game adversary could tell which scenario was chosen by the challenger. 
We hence must extend the simple properties and introduce scenario \emph{instances} to model dependencies.

\inlineheadingTwo{Scenario instances}
 We now require the adversary to give alternative instances for both scenarios (Fig. \ref{fig:senderReceiver} (b)). The challenger chooses the scenario according to the challenge bit, which is picked randomly for every game, and the instance according to the instance bit, which is picked randomly for every challenge.
 
 Formally, we replace steps 2--5 of the game with the following steps:
 \begin{description}
\item[2.] $\mathcal{A}$ sends a batch query, containing $\underline{r}^0_{0}$,  $\underline{r}^1_{0}$, $\underline{r}^0_{1}$and $\underline{r}^1_{1}$ to $Ch$. 
\item[3.] $Ch$ checks if the query is valid according to the analyzed notion $X$. 
\item[4.] If the query is valid  and $Ch$ has not already picked an instance bit $a$ for this challenge, $Ch$ picks $a \in \{0,1\}$ randomly  and independent of $b$. Then it inputs the batch corresponding to $b$ and $a$ to $\Pi$. 
\item[5.] $\Pi$'s output $\Pi(\underline{r}^a_{b})$ is forwarded to $\mathcal{A}$. 
 \end{description}
 
This allows us to model the goal that the adversary is not allowed to learn the sender and receiver: We allow the adversary to pick two sender-receiver pairs, which she uses as instances for the first scenario. 
The mixed sender-receiver pairs must then be provided as instances for the second scenario (see Fig. \ref{fig:senderReceiver} (b)). 
We thus force the game adversary to provide alternative assignments for each scenario.
This way she cannot abuse the model to win the game by identifying only the sender or the receiver.
We call this property \textit{Random Sender Receiver} $R_{SR}$. 

This complex property is still not sufficient to model the situation in, for example, Tor:
The adversary can distinguish the scenarios without learning who sent to whom, just by learning which senders and which receivers are active. 
Hence, we further restrict the adversary picking instances where both senders and both receivers are active by defining the property \textit{Mix Sender Receiver} $M_{SR}$. 
Here, the adversary picks two instances for $b=0$ where her chosen sender-receiver pairs communicate, and two for $b=1$ where the mixed sender-receiver pairs communicate. 
The two instances simply swap the order in which the pairs communicate (Fig. \ref{fig:senderReceiver} (c)). 
This way, we force the adversary to provide alternative assignments  for each scenario where both suspected senders and both suspected receivers are active. 
This combination prevents the adversary from winning the game without learning the information that the real system is actually supposed to protect, i.e. the sender-receiver pair. 

\begin{figure}[htbp]
  \centering
  \includegraphics[width=0.48\textwidth]{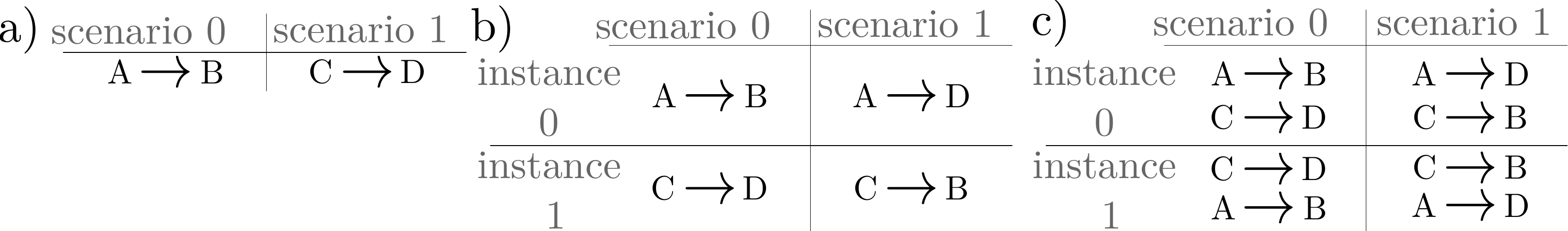}
  \caption{Examples showing the general structure of communications that differ in both scenarios: a) Naive, but incorrect b) Random Sender Receiver $R_{SR}$ c) Mixed Sender Receiver $M_{SR}$}
  \label{fig:senderReceiver}
\end{figure}
\inlineheading{Defining Complex Properties}
To simplify the formal definition of complex properties, we introduce \emph{challenge rows}. A challenge row is a pair of communications with the same index that differ in the two scenarios (e.g. ${r}_{0_j}, {r}_{1_j}$ with index $j$). 
For complex properties, the challenger only checks the differences of the challenge rows in the two scenarios.

\begin{definition}[Properties $R_{SR}$, $M_{SR}$]\label{def:complexProperties}
Let the given batches be $\underline{r}_b^a$ for instances $a \in \{0,1\}$ and scenarios $b \in \{0,1\}$, 
  $\mathsf{CR}$  the set of challenge row indexes, 
  $(u^a_0, {u'}^{\:a}_0)$ for both instances $a\in \{0,1\}$  be the sender-receiver-pairs of the first challenge row of the first scenario ($b=0$)\iflong in this challenge \fi.
Random Sender Receiver $R_{SR}$, Mixed Sender Receiver $M_{SR}$ ($R_{SM}, R_{RM}, M_{SM}, M_{RM}$ analogous) are met, iff:
\allowdisplaybreaks
\begin{align*}
   R_{SR}:\quad{r_{0}^a}_{cr}&=(\mathbf{u^{a}_0}, \mathbf{{u'}^{\:a}_0}, m^1_{0_{cr}}, aux^1_{0_{cr}})~\land \\
   {r_{1}^a}_{cr}&=(\mathbf{u^{a}_0}, \mathbf{ {u'}^{\:1-a}_0}, m^1_{0_{cr}},aux^1_{0_{cr}})\\[0.2em]
   &\forall cr \in \mathsf{CR} , a \in \{0,1\}&
\end{align*}
\begin{align*}
      M_{SR}:\quad&{r_0^a}_{cr}=(\mathbf{u^{a}_0}, \mathbf{ {u'}^{\:a}_0},m^1_{0_{cr}},aux^1_{0_{cr}})~\land \\
   &{r_0^a}_{cr+1}= (\mathbf{u^{1-a}_0}, \mathbf{ {u'}^{\:1-a}_0}, m^1_{0_{cr}},aux^1_{0_{cr}})~\land \\[0.2em]
   &{r_1^a}_{cr}=(\mathbf{u^{a}_0}, \mathbf{ {u'}^{\:1-a}_0},m^1_{0_{cr}},aux^1_{0_{cr}})~\land \\
   &{r_1^a}_{cr+1}=(\mathbf{u^{1-a}_0},\mathbf{ {u'}^{\:a}_0}, m^1_{0_{cr}},aux^1_{0_{cr}})\\[0.2em]
   &\text{for every second } cr \in \mathsf{CR} , a \in \{0,1\}& 
\end{align*}

\end{definition}
 
\inlineheading{Linking message senders}
A final common privacy goal that still cannot be covered is the unlinkability of senders over a pair of messages (\anoaULLong).
Assume a real world adversary that can determine that the sender of two messages is the same entity.
If subsequently she discovers the identity of the sender of one of the messages through a side channel, she can also link the second message to the same individual.

\inlineheadingTwo{Stages}
To model this goal, we need two scenarios  (1) both messages are sent by the same sender, and (2) each message is sent by a different sender. 
Further, the adversary picks the messages for which she wants to decide whether they are sent from the same individual, and which other messages are sent between those two messages. 
Therefore, we add the concept of \emph{stages} and ensure that only one sender sends in the challenge rows of stage 1, and in stage 2 either the same sender continues sending ($b=0$) or another sender sends those messages ($b=1$). This behavior is specified as the property \emph{Twice Sender} $T_S$. 



%
%


\begin{definition}[Property $T_S$]\label{def:complexProperties}
Let the given batches be $\underline{r}_b^a$ for instances $a \in \{0,1\}$ and scenarios $b \in \{0,1\}$,  $x$ the current stage, 
  $\mathsf{CR}$  the set of challenge row indexes, 
  $(u^a_0, {u'}^a_0)$ for both instances $a\in \{0,1\}$  be the sender-receiver-pairs of the first challenge row of the first scenario ($b=0$) \iflong in this challenge \fi in stage 1 and $(\tilde{u}^a_0,\tilde{u}'^a_0)$ the same pairs in stage 2.
Twice Sender $T_S$ is met, iff ($T_R$ analogous):
\begin{align*}
  T_{S}:&\quad x=stage1~\land\\
   &\qquad{r_0^a}_{cr}=(\mathbf{u^a_0},  {u'}^{\:0}_0,m^1_{0_{cr}},aux^1_{0_{cr}})~\land \\
   &\qquad{r_1^a}_{cr}=(\mathbf{u^a_0},  {u'}^{\:0}_0,m^1_{0_{cr}},aux^1_{0_{cr}})\\[0.2em]
   \mathbf{\lor} &\quad x=stage2~\land\\
   &\qquad{r_0^a}_{cr}= (\mathbf{u^{a}_0}, \tilde{u}'^{\:0}_0, m^1_{0_{cr}},aux^1_{0_{cr}})~\land \\
   &\qquad{r_1^a}_{cr}= (\mathbf{u^{1-a}_0}, \tilde{u}'^{\:0}_0, m^1_{0_{cr}},aux^1_{0_{cr}})\\
     &\quad\forall cr \in \mathsf{CR} , a \in \{0,1\}&
\end{align*}

\end{definition}
 Hence, we need to facilitate distinct stages for notions with the complex properties $T_S$ or $T_R$. Precisely, in step 2 of the game, the adversary is additionally allowed to switch the stages.
 
 Note that the above definition can easily be extended to having more stages and hence, more than two messages for which the adversary needs to decide whether they have originated at the same sender.
 
This set of properties allows us to specify all privacy goals that have been suggested in literature as privacy notions and additionally all that we consider important. 
It is of course difficult to claim completeness, as future \acp{ACN} may define diverging privacy goals and novel observable properties (or side-channels) may be discovered.

%% file: sections/notions.tex
\section{Privacy Notions}
\label{notions}

Given the properties above, we can now set out to express intuitive privacy goals as formal privacy notions.
We start by specifying sender unobservability as an example leading to  a general definition of our privacy notions.

Recall the first game we defined in Section \ref{background}, which corresponds to sender unobservability ($\bohliSSAc$ = S(ender) $\lnot$ O(bservability)). There, in both scenarios something has to be sent, i.e. we need to specify that sending nothing is not allowed: \something. Further, both scenarios can only differ in the senders, i.e.  we also need the property that everything but the senders is equal: $\EveryButSender$. Hence, we define sender unobservability as $\bohliSSAc :=$\something$ \land \EveryButSender$.
\footnote{Technically $\EveryButSender$ already includes \something. However, to make the differences to other notions more clear, we decide to mention both in the definition.}

We define all other notions in the same way:

\begin{definition}[Notions]
Privacy notions are defined as a boolean expression of the properties 
 according to Table~\ref{NotionsDefinition}.
\end{definition}

 \begin{table}[htb]
\center
\resizebox{0.4\textwidth}{!}{%
  \begin{tabular}{ l l }

Notion&Properties  \\ \hline
$\loopixSRTPU$ &\something $ \land \EveryButSenderRec \land M_{SR}$  \\
$\anoaREL$ &\something$\land \EveryButSenderRec \land R_{SR}$ \\
$\newCONFWOL$ &\something $ \land \EveryButMsg$ \\
$\newCONF$ &\something $ \land \EveryButMsg \land |M| $ \\
$\heviaUL$&  \something$\land Q \land Q'$\\
$\bohliSA$& \something \\
$\heviaUO$& \nothing\\ \hline
$\bohliSSAc$& \something $\land \EveryButSender$\\
$\bohliSSUPc $& \something$\land \EveryButSender \land |U|$\\
$\bohliSWUPc$ &\something$\land \EveryButSender \land H$\\
$\bohliSPSc $&\something $\land \EveryButSender \land P$\\
$\bohliSSUUc$ &\something$\land \EveryButSender \land U$ \\
$\bohliSWUUc$ &\something $\land \EveryButSender \land U \land H$\\
$\bohliSANc $&\something $\land \EveryButSender \land U \land P$\\
$\bohliSWUc$ &\something $\land \EveryButSender \land Q$\\
$\bohliSWAc$ &\something $\land \EveryButSender \land Q \land P$\\
$\anoaUL$ &\something$\land \EveryButSender \land T_S $ \\
$\bohliRSAc$ \ etc.&analogous\\ \hline
\mbox{$\newAnoaSAWOL$}&\something $ \land \EveryButSenderMsg $ \\
\mbox{$\newAnoaSA$}&\something $ \land \EveryButSenderMsg \land |M| $ \\
$\newSAUO$ &\something $  \land \EveryButSenderMsg \land R_{SM}$ \\
$\newSA$&\something $  \land \EveryButSenderMsg \land M_{SM} $ \\
\mbox{$\anoaRA$} \ etc.&analogous \\ \hline
$\sgame{X'}$& Properties of $X'$, remove $\EveryButRec$ \\
\multicolumn{2}{c}{for $X' \in \{\bohliRSAc,$ $ \bohliRSUPc,$ $\bohliRWUPc,$ $\bohliRPSc,$ $\bohliRSUUc,$} \\
\multicolumn{2}{c}{$\bohliRWUUc,$ $\bohliRANc,$ $\bohliRWUc,$ $\bohliRWAc\}$}\\
$\rgame{X}$ &analogous \\ 
\end{tabular}}
 \caption{Definition of the notions. A description of simple properties was given in Table~\ref{tab:information}. }
  \label{NotionsDefinition}
\end{table} 

\iflong
Modeling the notions as a game, the respective challenger verifies all properties (and the later introduced options) of the adversary's queries.
A complete description of the challenger can be found in Appendix \ref{sec:challenger}. 
Further, an example of how the definitions can be represented by using a challenge specific state, which the challenger maintains, is shown in Algorithms \ref{Challenger} and \ref{calcNewState}  in Appendix \ref{pseudocode}.
\else
Modeling the notions as a game, the respective challenger will check all aspects of the adversary's queries.
A complete description of the challenger can be found in Appendix \ref{sec:challenger}. 
\fi

%% file: sections/choiceOfNotions.tex
\section{On the Choice of Notions}
\label{sec:choiceNotions}

The space of possible combinations of properties, and hence of conceivable privacy notions, is naturally large.
Due to this, we verify our selection of privacy goals by finding exemple use cases.
Additionally, we demonstrate the choice and the applicability of our definition  by analyzing the privacy goals of Loopix, an \ac{ACN} that was recently published.
We additionally verify that our privacy notions include those of previous publications that suggest frameworks based on indistinguishability games, and provide a complete mapping in Section \ref{sec:ComparingFrameworks}.


\subsection{Example Use Cases for the Notions}
\label{sec:notionExamples}
We illustrate  our notions by continuing the example of an activist group trying to communicate in a repressive regime, although our notions are generally applicable. 

Recall the  general idea of an indistinguishability game from the examples in Section \ref{background}:
To prove that an \ac{ACN} hides certain properties, whatever is allowed to be learned in the actual \ac{ACN} must not help a game adversary to win. This way, she is forced to win the game solely based on those properties that are required to remain hidden.
Therefore, the information allowed to be disclosed cannot be used in the game and hence must be kept identical in both scenarios.

Before giving examples, we need to order the notions. We chose to group them semantically.
Our resulting clusters are shown as gray boxes in Figure \ref{fig:hierarchyColored}. 
Horizontally, we categorize notions that focus on receiver or sender protection (Receiver Privacy Notions or Sender Privacy Notions, respectively) or treat both with the same level of importance (Impartial Notions).
Inside those categories, we use clusters concerning the general leakage type: Both-side Unobservability means that neither senders, nor receivers or messages should be leaked.
Both-side Message Unlinkability means that it should be possible to link neither  senders nor receivers to messages. In Sender Observability, the sender of every communication can be known, but not the message she sends or to whom she sends (Receiver and Message Observability analogous). In Sender-Message Linkability, who sends which message can be known to the adversary (Receiver-Message and Sender-Receiver Linkability analogous).
\iflong
 \begin{table} [thb]
\center
\resizebox{0.45\textwidth}{!}{%
  \begin{tabular}{ c p{6cm}  }

Usage&Explanation\\ \hline
$D \in \{S,R,M\}$& Dimension $\in \{$Sender, Receiver, Message$\}$\\
Dimension $D$ not mentioned& Dimension can leak \\ 
Dimension $D$ mentioned &Protection focused on this dimension exists\\ \hline
$D \overline{O}$& not even the participating items regarding D leak,(e.g. $S\overline{O}$: not even $U$ leaks)\\
$DF \overline{L}$& participating items regarding D can leak, but not which exists how often (e.g. $SF\overline{L}$: $U$ leaks, but not $Q$)\\
$DM \overline{L}$& participating items regarding D and how often they exist can leak ( e.g. $SM\overline{L}$: $U,Q$ leaks)\\ \hline
$X -Prop, $& like X but additionally Prop can leak\\
$Prop \in \{|U|,H,P,|U'|, H',P', |M| \}$&\\  \hline
$(D_1 D_2)\overline{O}$& uses $R_{D_1 D_2}$; participating items regarding $D_1,D_2$ do not leak, (e.g. $(SR)\overline{O}$: $R_{SR}$)\\ 
$(D_1 D_2)\overline{L}$& uses $M_{D_1 D_2}$; participating items regarding $D_1,D_2$ can leak, (e.g. $(SR)\overline{L}$: $M_{SR}$)\\ 
$(2D)\overline{L}$& uses $T_{D}$; it can leak whether two participating item regarding $D$ are the same,  (e.g. $\anoaUL$: $T_{S}$)\\ \hline
$\overline{O}$&short for $S \overline{O} R\overline{O} M\overline{O} $\\
$\heviaUL$& short for $ M\overline{O}(SM\overline{L}, RM\overline{L})$\\
$S\overline{O}\{X\}$& short for $S\overline{O} M\overline{O} X$\\
$D_1 X_1[ D_2 X_2]$& $D_1$ is dominating dimension, usually $D_1$ has more freedom, i.e. $X_2$ is a weaker restriction than $X_1$ \\ \hline
$C\overline{O}$& nothing can leak (not even the existence of any communication)\\
\end{tabular}}
 \caption{Naming Scheme}
  \label{Tab:NamingScheme1}
\end{table}

We also want to explain our naming scheme, which we  summarize in Table \ref{Tab:NamingScheme1}. Our notions consider three dimensions: senders, messages and receivers.
 Each notion restricts the amount of leakage on each of those dimensions. However, only dimensions that are to be protected are part of the notion name. 
We use $\overline{O}$, short for unobservability, whenever the set of such existing items of this dimension cannot be leaked to the adversary.
 E.g. $S\overline{O}$ cannot be achieved if the set of senders $U$ is leaked. 
 Notions carrying $\overline{L}$, short for unlinkability, can leak $U$ (for sender related notions), but not some other property related to the item. 
 E.g. we use $SF\overline{L}$ if the frequency $Q$ cannot be leaked and $SM\overline{L}$, if $Q$ can be leaked, but not the sender-message relation. 
  With a ``$-Prop$'' we signal that the property $Prop$ can additionally leak to the adversary. We distinguish those properties from $U$ and $Q$ used before as they give another leakage dimension (as illustrated later in the hierarchy).
 Further, we use parentheses as in $(SR)\overline{O}$ to symbolize that if not only one set, but both sets of senders and receivers ($U$ and $U'$) are learned the notion is broken. Analogously, in $(SR)\overline{L}$ both sets can be learned but the linking between sender and receiver cannot.
For the last missing complex property, we use $(2S)\overline{L}$ to symbolize that two senders have to be linked to be the same identity to break this notion.

For readability we add some abbreviations: 
We use  $\overline{O}= S \overline{O} R\overline{O} M\overline{O} $ to symbolize unobservability on all three types and we summarize the remaining types in $ M\overline{O}(SM\overline{L}, RM\overline{L})$ to $\heviaUL$. $C\overline{O}$ symbolizes the notion in which nothing is allowed to leak.
Further, we use curly brackets to symbolize that the message cannot be leaked $ S\overline{O}\{X\}= S\overline{O} M\overline{O} X$ and we put the (in our understanding) non dominating part of the notion in brackets $S\overline{O} M\overline{O}= S\overline{O}[M\overline{O}]$.

\else
Table \ref{Tab:NamingScheme} of Appendix \ref{app:summaryNamingScheme} summarizes our naming scheme.
\fi


\vspace{-0.3cm}
\subsubsection{Impartial Privacy Notions}
\vspace{-0.2cm}
These notions treat senders and receivers equally.

\inlineheading{Message Observability}
The content of messages can be learned in notions of this group, as messages are not considered confidential. Because the real world adversary can learn the content, we must prevent her from winning the game trivially by choosing different content. 
Hence, such notions use the property that the scenarios are identical except for the senders and receivers ($\EveryButSenderRec$) to ensure that the messages are equal in both scenarios. 

\example{An activist of the group is already well-known and communication with that person leads to persecution of Alice.}

Alice needs a protocol that hides whether a certain sender and receiver  communicate with each other; cf. Section \ref{sec:ComplexProperties} motivation of the complex property $M_{SR}$.
The resulting notion is \emph{\loopixSRTPULong \ ($\loopixSRTPU$)}.

\exampleCont{Only few people participate in the protocol. Then, just using the protocol  to receive (send) something, when the well known activist is acting as sender (receiver) threatens persecution. }

Alice needs a protocol that hides whether a certain sender and receiver actively participate at the same time or not; cf. Section \ref{sec:ComplexProperties} motivation of the complex property $R_{SR}$.
The resulting notion is \emph{\anoaRELLong\  ($\anoaREL$)}.

\inlineheading{Sender-Receiver Linkability (Message Confidentiality)}
Senders and receivers can be learned  in notions of this group, because they are not considered private.
Hence, such notions include the property that the scenarios are identical, except for the messages ($\EveryButMsg$) to ensure that the sender-receiver pairs are equal in both scenarios.

\example{Alice wants to announce her next demonstration. (1) Alice does not want the regime to learn the content of her message and block this event. (2) Further,  she is afraid that the length of her messages could put her under suspicion, e.g. because activists tend to send messages of a characteristic length.}

In (1) Alice needs a protocol that  hides the content of the messages. However, the adversary is allowed to learn all other attributes,
in particular the length of the message.
Modeling this situation, the scenarios may differ solely in the message
content; all other attributes must be identical in both scenarios, as
they may not help the adversary distinguish between them.
Beyond the above-described $\EveryButMsg$, we must thus also request that the
length of the messages $|M|$ is identical in both scenarios.
The resulting notion is  \emph{\newCONFLong\ ($\newCONF$)}\footnote{We stick to our naming scheme here, although we would commonly call this confidentiality.}.

In the second case (2), the protocol is required to hide the length of the message.
The length of the messages thus may differ in the two scenarios, as the
protocol will need to hide this attribute. Hence, we remove the restriction that the message length $|M|$ has to be equal in both scenarios from the above notion and end up with \emph{\newCONFWOLLong\ $\newCONFWOL$}.

\iflong
\inlineheading{Both-Side Message Unlinkability}
Notions of this group are broken if the sender-message or receiver-message relation is revealed.

\example{The activists know  that their sending and receiving frequencies are similar to regime supporters' and that using an \ac{ACN} is in general not forbidden, but nothing else. Even if the content and length  of the message ($\newCONFWOL$) and the sender-receiver relationship ($\loopixSRTPU$)  is hidden, the regime might be able to distinguish uncritical from critical communications, e.g. whether two activists communicate  ``Today'' or  innocent users an innocent message.  In this case, the regime might learn that currently many critical communications take place and improves its measures against the activists.}

In this case, the activists want a protocol that hides the communications, i.e. relations of sender, message and receiver. However, as using the protocol is not forbidden and their sending frequencies are ordinary, the adversary can learn which users are active senders or receivers and how often they sent and receive. Modeling this, the users need to have the same sending and receiving frequencies in both scenarios $Q,Q'$, since it can be learned. However, everything else needs to be protected and hence, can be chosen by the adversary. This corresponds to the notion  \emph{\heviaULLong \ ($\heviaUL$)}.
\fi

\inlineheading{Both-Side Unobservability}
Even the activity of a certain sender or receiver is hidden in notions of this group.

\exampleCont{It is a risk for the activists, if the regime can distinguish between two leading activists exchanging the message ``today'' and two loyal regime supporters exchanging the message ``tomorrow''.}

 In this case, Alice wants to disclose nothing about senders, receivers, messages or their combination. However, the adversary can learn the total number of communications happening in the \ac{ACN}. Modeling this, we need to assure that for every communication in the first scenario, there exists one in the second. We achieve this by prohibiting the use of the empty communication with property \something. This results in the notion \emph{\bohliSALong \ ($\overline{O}$)}.

\example{The regime knows that a demonstration is close, if the total number of communications transmitted over this protocol increases. It then prepares to block the upcoming event.}

To circumvent this, Alice needs a protocol  that additionally hides the total number of communications. Modeling this, we need to allow the adversary to pick any two scenarios. Particularly, use of the empty communication  $\Diamond$ is allowed. This is represented in the property that nothing needs to be equal in the two scenarios, \nothing\ ,  and results in the notion \emph{\heviaUOLong \ ($\heviaUO$)}.
Note that this is the only notion where the existence of a communication is hidden. All other notions include \something\ and hence do not allow for the use of the empty communication.


\subsubsection{Sender (and Receiver) Privacy Notions}
\vspace{-0.2cm}
These notions allow a greater freedom in picking the senders (or
receivers: analogous notions are defined for receivers.).

\inlineheading{Receiver-Message Linkability}
The receiver-message relation can be disclosed in notions of this group.
Hence, such notions include the property that the scenarios are identical except for the senders ($\EveryButSender$) to ensure the receiver-message relations are equal in both scenarios.

In \emph{\bohliSWUcLong \ ($\bohliSWUc$)}  the total number of communications and how often each user sends can be additionally learned.
However,  who sends which message is hidden.
In \emph{\bohliSSUUcLong } \emph{($\bohliSSUUc$)}  the set of users and the total number of communications can be additionally disclosed.
However,  how often a certain user sends is hidden, since it can vary between the two scenarios.
In \emph{\bohliSSAcLong \ ($\bohliSSAc$)}, the total number of communications can additionally be disclosed.
 However, especially the set of active senders $U_b$ is hidden.

If a notion further includes the following abbreviations, the following information can  be disclosed as well:
\begin{itemize}
 \item \emph{with User Number Leak}  ($-|U|$): the number of senders that send something in the scenario
\item \emph{with Histogram Leak}  ($-H$):  the histogram of how many senders send how often 
\item \emph{with Pseudonym Leak}  ($-P$): which messages are sent from the same user 
\end{itemize} \vspace*{-\baselineskip}

\exampleSpec{Alice is only persecuted when the regime can link a message with compromising content to her {\normalfont -- she needs a protocol that at least provides $\bohliSWAc$.} 
However, since such a protocol does not hide the message content, the combination of all the messages she sent might lead to her identification. {\normalfont Opting for a protocol that additionally hides the message combination ($P$), i.e. provides $\bohliSWUc$,  can protect her from this threat.} \\
Further, assuming most users send compromising content, and Alice's message volume is high, the regime might easily suspect her to be the origin of some compromising messages even if she is careful that the combination of her messages does not reidentify her {\normalfont -- she  needs a protocol that does not disclose her sending frequencies ($Q$) although the combination of her messages ($P$) might be learned, i.e. achieving $\bohliSANc$.}
However, Alice might fear disclosing the combination of her messages {\normalfont - then she needs a protocol achieving at least $\bohliSWUUc$, which hides the frequencies ($Q$) and the message combination ($P$), but discloses the sending histogram, i.e. how many people sent how many messages ($H$).}
However, if  multiple activist groups use the \ac{ACN} actively at different time periods, disclosing the sending histogram $H$ might identify how many activist groups exist and to which events they respond by more active communication {\normalfont -- to prevent this she needs a protocol that hides the frequencies $Q$ and the histogram $H$, i.e. provides $\bohliSSUUc$.} \\
Further, not only sending a certain content, but also being an active sender (i.e. being in $U$) is prosecuted {\normalfont she might want to pick a protocol with at least $\bohliSPSc$. 
Again if she is afraid that leaking $P$ or $H$ together with the expected external knowledge of the regime would lead to her identification, she picks the corresponding stronger notion.}
If the regime knows that senders in the \ac{ACN} are activists and learns that the  number of active senders is high, it blocks the \ac{ACN}. 
{\normalfont In this case at least $\bohliSSAc$ should be picked to hide the number of senders ($|U|$).}}

\example{ For the next protest, Alice sends two messages: (1) a location, and (2) a time. If the regime learns that both messages are from the same sender, they will block the place at this time even if they do not know who sent the messages.} 
Alice then needs a protocol that hides whether two communications have the same sender or not. We already explained how to model this with complex property $T_{S}$ in  Section \ref{sec:ComplexProperties}.
The resulting notion is \anoaULLong ($\anoaUL$).

\iflong

\inlineheading{Receiver Observability}
In notions of this group the receiver of each communication can be learned.
Hence, such notions include the property that the scenarios are equal except for the senders and messages ($\EveryButSenderMsg$) to ensure that they are equal in both scenarios.

\example{Consider not only sending real messages is persecuted, but also the message content or any combination of senders and message contents is exploited by the regime. If the regime e.g. can distinguish activist Alice sending ``today'' from regime supporter Charlie sending ``see u'', it might have learned an information the activists would rather keep from the regime.  Further, either (1) the activists know that many messages of a certain length are sent or (2) they are not sure that  many messages of a certain length are sent.}

In case (1), Alice needs a \ac{ACN}, that hides the sender activity, the message content and their combination. However, the adversary can especially learn the message length. Modeling this, beyond the above described $\EveryButSenderMsg$, the message lengths have to be equal $|M|$. This results in the notion  \emph{\newAnoaSALong}  ($\newAnoaSA$).
Note that in $\newAnoaSA$ the properties of $\newCONF$ are included and further the senders are allowed to differ in the two scenarios.
The second case (2) requires a protocol that additionally hides the message length. Hence, in modeling it we remove the property that the message lengths are equal $|M|$ from the above notion. This results in \emph{\newAnoaSAWOLLong}  ($\newAnoaSAWOL$).

\example{ Alice's demonstration is only at risk if the regime can link a message with a certain content to her as a sender with a non negligible probability.} 
Then at least \emph{\newSALong\  ($\newSA$)}, which is defined analogous to $\loopixSRTPU$ is needed. 

\exampleCont{However, $\newSA$ only allows Alice to claim that not she, but Charlie sent a critical message $m_a$ and the regime cannot know or guess better. Now assume that Dave is also communicating, then the regime might be able to distinguish Alice sending $m_a$, Charlie $m_c$ and Dave $m_d$ from Alice sending $m_d$, Charlie $m_a$ and Dave $m_c$. In this case, it might not even matter that Alice can claim that Charlie possibly sent her message. The fact that when comparing all three communications that possibly happened, Alice is more likely to have sent the critical message $m_a$ means a risk for her.}

 To circumvent this problem Alice needs a protocol that not only hides the difference between single pairs of users, but any number of users. Modeling this, instead of the complex property $M_{SM}$, we need to restrict that the active senders' sending frequencies are equal, i.e. $\bohliSWUc$.

\example{In another situation our activists already are prosecuted for being a sender while a message with critical content is sent. }

In this case at least  \emph{\newSAUOLong\  ($\newSAUO$)}, which is defined analogous to $\anoaREL$ is needed.

Analogous notions are defined for receivers.

\inlineheading{Sender Privacy Notions: Both-Side Message Unlinkability}
As explained with the example before in the case that Alice does not want any information about senders, receivers and messages or their combination to leak, she would use $\overline{O}$. However, the privacy in this example can be tuned down, if she assumes that the regime does not have certain external knowledge or that the users are accordingly careful. As explained for the Sender Notions with Receiver-Message Linkability before, in this case we might decide to allow $U', |U'|,Q',H',P'$ to leak.

If a notion $X \in \{\bohliRSAc, \bohliRSUPc, \bohliRWUPc, \bohliRPSc,\bohliRSUUc,\bohliRWUUc,\bohliRANc,\bohliRWUc,\bohliRWAc\}$ is extended to \emph{Sender Unobservability by X} \emph{($ \sgame{X}$)}, the leaking of the sender-message relation is removed.
This is done by removing $\EveryButRec$.
 Since the attacker now has a greater degree of freedom in choosing the senders and is (if at all) only restricted in how she chooses the receivers and messages, this is a special strong kind of Sender Unobservability.
 Analogous notions are defined for receivers.\footnote{Note that $\sgame{\bohliRSAc}=\rgame{\bohliSSAc}=\bohliSA$.}
\else
Due to page limits the examples for the remaining notions can be found in Appendix \ref{app:examples}.
\fi


%% file: sections/application.tex
\vspace{-0.5cm}
\subsection{Analyzing Loopix's Privacy Goals}
\label{application}
\label{sec:loopix}
\vspace{-0.2cm}

 To check if we include currently-used privacy goals, we decide on a current \ac{ACN} that has defined its goals based on an existing analytical framework  and which has already been analyzed: the Loopix anonymity system \cite{piotrowska17loopix}. In this section, we show that the privacy goals of  Loopix map to notions we have defined (although the naming differs).
Loopix aims for Sender-Receiver Third-Party Unlinkability, Sender online Unobservability and Receiver Unobservability.  

\inlineheading{Sender-Receiver Third-Party Unlinkability}
Sender-Receiver Third-Party Unlinkability means that an adversary cannot distinguish scenarios where two receivers are switched:
\begin{quote}
``The senders and receivers should be unlinkable by any unauthorized party. Thus, we consider an adversary that wants to infer whether two users are communicating. We define \emph{sender-receiver third party unlinkability} as the inability of the adversary to distinguish whether $\{S_1\rightarrow R_1, S_2 \rightarrow R_2\}$ or $\{S_1 \rightarrow R_2, S_2 \rightarrow R_1\}$ for any concurrently online honest senders $S_1,S_2$ and honest receivers $R_1,R_2$ of the adversary's choice.''  \cite{piotrowska17loopix}
\end{quote}
The definition in Loopix allows the two scenarios to be distinguished by learning the first receiver.
We interpret the notion  such that it is only broken if the adversary learns a sender-receiver-pair, which we assume is what is meant in \cite{piotrowska17loopix}. This means that the sender and receiver of a communication must be learned and is exactly the goal that motivated our introduction of complex properties: $\loopixSRTPU$.

\inlineheading{Unobservability}
In sender online unobservability the adversary cannot distinguish whether an adversary-chosen sender communicates ($\{S\rightarrow\}$) or not \mbox{($\{S \centernot \rightarrow\}$)}:
\begin{quote}
``Whether or not senders are communicating should be hidden from an unauthorized third party. We define \emph{sender online unobservability} as the inability of an adversary to decide whether a specific sender $S$ is communicating with any receiver $\{S \rightarrow\}$ or not $\{S \centernot \rightarrow\}$, for any concurrently online honest sender $S$ of the adversary's choice.''  \cite{piotrowska17loopix}
\end{quote}
Receiver unobservability is defined analogously.

\iflong
 \begin{table} 
\center
\resizebox{0.3\textwidth}{!}{%
  \begin{tabular}{ c c c }

Notion&Name&Aspects\\ \hline
$\loopixSUO$ &Loopix's Sender Unobservavility& $ \EqualOrNothing$\\ 
$\loopixRUO$ & Loopix's Receiver Unobservability&$\EqualOrNothing$\\
$\loopixSUONe$ &Restricted Sender Unobservability& $\centernot \rightarrow \land \EveryButSender$  \\
$\loopixRUONe$ &Restricted Receiver Unobservability& $\centernot \rightarrow' \land \EveryButRec$ \\
\end{tabular}}
 \caption{Definition of the Loopix notions}
  \label{NotionsLoopix}
\end{table}
\fi

Those definitions are open to interpretation.
On the one hand, $\{S \centernot \rightarrow\}$ can mean that there is no corresponding communication in the other scenario.
\iflong
This corresponds to our $\Diamond$ and the definition of $\loopixSUO$  and $\loopixRUO$ according to Table \ref{NotionsLoopix}.
\else
This corresponds to our $\Diamond$ and the definition of $\loopixSUO$  and $\loopixRUO$ in  Appendix \ref{Loopix1Sketch}.
\fi
When a sender is not sending in one of the two scenarios, this means that there will be a receiver receiving in the other, but not in this scenario.
Hence, $\loopixSUO$  can be broken by learning about receivers and the two notions are equal.
\iflong
These notions are equivalent to $\heviaUO$:
\begin{theorem}
\label{loopixUO}
It holds that
\[(c, \epsilon, \delta)-\heviaUO \Rightarrow (c, \epsilon, \delta)- \loopixSUO_{CR_1}\text{.}\]
\[(c, \epsilon, \delta)-\heviaUO \Leftarrow (2c, \epsilon, \delta)- \loopixSUO_{CR_1}\text{.}\]
\end{theorem}

\begin{proofsketch}
$\heviaUO \implies \loopixSUO$  by definition.
For $\loopixSUO \implies \heviaUO$  we use the following argumentation: Given an attack on $\heviaUO$, we can construct an attack on $\loopixSUO$ with the same success. Assume a protocol has $\loopixSUO$, but not $\heviaUO$. Because it does not achieve $\heviaUO$, there exists a successful attack on $\heviaUO$. However, this implies that there exists a successful attack on $\loopixSUO$ (we even know how to construct it). This contradicts that the protocol has $\loopixSUO$.
We construct an successful attack on $\loopixSUO$ by creating two new batches  $(r_0, \Diamond  )$ and $( \Diamond , r_1)$  for every challenge row $(r_0,r_1)$ in the successful attack on $\heviaUO$.
\end{proofsketch}
\else
These notions are equivalent to $\heviaUO$ (see Appendix \ref{Loopix1Sketch}).
\fi

On the other hand, $\{S \centernot \rightarrow\}$ can mean that sender $u$ does not send anything in this challenge.
In this case, the receivers can experience the same behavior in both scenarios and the notions differ.
\iflong
We  formulate these notions as $\loopixSUONe$ and $\loopixRUONe$ according to Table \ref{NotionsLoopix}.
Therefore, we need a new property that some sender/receiver is not participating in any communication in the second scenario:

\begin{definition}[Property $\centernot\rightarrow$]
Let $u$ be the sender of the first scenario in the first challenge row of this challenge.
We say that $\centernot \rightarrow$ is fulfilled iff  for all $j:$ $  u_{1_j}\neq u$.
(Property $\centernot \rightarrow'$ is defined analogously for receivers.) 
\end{definition}

\begin{theorem}
\label{loopixUO2}
It holds that
\begin{align*} (c, \epsilon, \delta)-\bohliSSAc &\Rightarrow (c, \epsilon, \delta)- \loopixSUONe \text{ and}\\
 (c, 0, 2\delta)-\bohliSSAc &\Leftarrow (c, 0, \delta)- \loopixSUONe \text{.}
 \end{align*}
 
\end{theorem}
\begin{proofsketch}
Analogously to Theorem \ref{loopixUO}.
$\Rightarrow$: Every attack on $\loopixSUONe$  is by definition a valid attack on $\bohliSSAc$.

$\Leftarrow$: Given an attack $\mathcal{A}$ on $(c, 0, 2\delta)-\bohliSSAc$, where both scenarios of a challenge use all users (otherwise it would be a valid attack on $\loopixSUONe$).
Let $(r_{2_1}, \dots, r_{2_l})$ be the same batch as the second of  $\mathcal{A}$ except that whenever one of the two senders of the first challenge row  from the original scenarios is used, it is replaced with an arbitrary other sender (that was not used in the first challenge row of the original scenarios).
Let 
$P(0|2)$ be the probability that $\mathcal{A}$ outputs 0, when the new batches are run; $P(0|0)$ when the first scenario of $\mathcal{A}$ is run and $P(0|1)$ when the second is run.
In the worst case for the attacker $P(0|2)=\frac{P(0|0)+P(0|1)}{2}$ (otherwise we would replace the scenario $b$ where $|P(0|2)-P(0|b)|$ is minimal with the new one and get better parameters in the following calculation).
Since $\mathcal{A}$ is an attack on $(c, 0, 2\delta)-\bohliSSAc$, $P(0|0)>  P(0|1) +  2\delta$.
Transposing and inserting the worst case for $P(0|2)$ leads to: $ P(0|0) > 2P(0|2)- P(0|0)+ 2 \delta \iff  P(0|0)> P(0|2) +  \delta $.
Hence, using $\mathcal{A}$ with the adapted scenario is an attack on  $(c, 0, \delta)- \loopixSUONe$\footnote{An analogous argumentation works for $(c, \epsilon- \ln(0.5), \delta)-\bohliSSAc \Leftarrow (c, \epsilon, \delta)- \loopixSUONe$.}.
\end{proofsketch}
\else
We formulate this notion and argue its equivalence  to $\bohliSSAc$ (with a change in parameters) in Appendix \ref{Loopix2Sketch}.
\fi
This is equivalent to AnoA's sender anonymity $\alpha_{SA}$. Analogously, Loopix's corresponding receiver notion is equivalent to  $\bohliRSAc$, which is even weaker than AnoA's receiver anonymity.


\inlineheading{Remark}
We do not claim that the Loopix system achieves or does not achieve any of these notions, 
since we based our analysis on the definitions of their goals, which were not sufficient to unambiguously derive the corresponding notions.

%% file: sections/options.tex
\section{Options for Notions}
\label{sec:options}

Additionally to the properties, we define options. Options can be added to any notion and allow for a more precise mapping of real world protocols, aspects of the adversary model, or easier analysis by quantification.

\vspace{-0.3cm}
\subsection{Protocol-dependent: Sessions}\label{sec:sessions}
 \vspace{-0.3cm}
Some \ac{ACN} protocols, like e.g. Tor, use sessions. Sessions encapsulate sequences of communications from the same sender to the same receiver by using the same session identifier for them.  In reality, the adversary might be able to observe the session identifiers but (in most cases) not to link them to a specific user.

To model sessions, we therefore set the auxiliary information of a communication to the session ID ($sess$): $aux=sess$.
 However, as the adversary can choose this auxiliary information, we need to ensure that the scenarios cannot be distinguished just because the session identifier is observed.
Hence, we definine $sess$ to be a number in most communications. Only for the session notions, we require special session IDs that correspond to the current challenge $\Psi$ and stage $x$ in all challenge rows: $(x,Ch\Psi)$. In this way, they have to be the same in both scenarios and a concrete $sess$ is only used in one stage of one challenge.

 The session identifier that is handed to the \ac{ACN} protocol model is a random number that is generated  by the challenger when a new $sess$ is seen. Hence, neither leaking (it is a random number) nor linking session identifiers (it will be picked new and statistically independent for every challenge and stage) will give the attacker an advantage.

We formalize this in the following definition, where we also use `$\_$’ to declare that this part of a tuple can be any value.\footnote{E.g.
$\exists(u,m,\_) \in r$ will be true iff $\exists u': \exists(u,m,u') \in r$.}

\begin{definition}[Sessions]\label{def:sessions}
Let $x$ be the stage and $u^a_0,u^a_1,u'^a_0,u'^a_1$ be the senders and receivers of the first challenge row of this challenge $\Psi$ and stage in instance $a \in \{0,1\}$.
Property $sess$ is met, iff for all $a\in \{0,1\}$:
{\footnotesize
\begin{align*}
  sess&: \forall (r^a_0,r^a_1)\in \mathsf{CR}(\underline{r}^a_0,\underline{r}^a_1):(r^a_0,r^a_1)=\\
  &\quad (u^a_0,u'^a_0,\_,(x, Ch\Psi) ),(u^a_1,u'^a_1,\_,(x,Ch\Psi)) 
\end{align*}
}
\end{definition}
\vspace{-0.3cm}
 As not all protocols use sessions, we allow to add sessions as an option to the notion $X$ abbreviated by $\manySess{X}$.

\subsection{Adversary Model: Corruption} \label{sec:advCap}
Some adversary capabilities like user corruption imply additional checks our challenger has to do. As all properties are independent from corruption, we add corruption as an option, that can be more or less restricted as shown in Table \ref{tab:corruption options}.
 The different corruption options have implications on the challenger, when a corrupt query or a batch query arrives.
 
   \begin{table} [t]
\center
\resizebox{0.35\textwidth}{!}{%
  \begin{tabular}{ l p{7cm}}

Symbol &Description\\ \hline

$X$&Adaptive corruption is allowed.\\
$\static{X}$ & Only static corruption of users is allowed.\\
$\noCorr{X}$ & No corruption of users is allowed.\\ \hline
$X$&Corrupted users not restricted.\\
$\corrNoComm{X}$&Corrupted users are not allowed to be chosen as senders or receivers.\\
$\corrOnlyPartnerSender{X}$&Corrupted users are not  allowed to be senders.\\ 
$\corrOnlyPartnerReceiver{X}$&Corrupted users are not  allowed to be receivers.\\
$\corrStandard{X}$&Corrupted users send/receive identical messages in both scenarios.\\ 
\end{tabular}}
 \caption{Options for corruption and for corrupted communication}
  \label{tab:corruption options}
\end{table}

 \inlineheadingTwo{Check on corrupt queries}
 This check depends on whether the user corruption is adaptive, static, or not allowed at all.
  The default case for notion $X$ is adaptive corruption, i.e.
the adversary can corrupt honest users at any time. 
With static corruption $\static{X}$, the adversary has to corrupt a set of users before she sends her first batch. 
The third option, $\noCorr{X}$, is that no user corruption is allowed.
We denote the set of corrupted users as $\hat{U}$.

\begin{definition}[Corruption: Check on Corrupt Query]\label{def:corruptionCorrupt}
Let $\hat{U}$ be the set of already corrupted users, $u$  the user in the corrupt query and the bit $\mathrm{subsequent}$ be true iff at least one batch query happened. 
The following properties are met, iff:
{\footnotesize
\begin{align*}
  {corr}_{static}&: \mathrm{subsequent} \implies u \in \hat{U} \\
  {corr}_{no}&: \perp  \quad \quad {corr}_{adaptive}: \top\\
\end{align*}
}
\end{definition}
\vspace{-2em}
\inlineheadingTwo{Check on batch queries}
In reality for most ACNs the privacy goal can be broken for corrupted users, e.g. a corrupted sender has no unobservability.  
Therefore, we need to assure that the adversary cannot distinguish the scenario because the behavior of corrupted users differs.  This is done by assuring equal behavior $\corrStandard{corr}$ or banning such users from communicating $\corrNoComm{corr},\corrOnlyPartnerSender{corr},\corrOnlyPartnerReceiver{corr}$.

\begin{definition}[Corruption: Check on Batch Query]\label{def:corruption}
The following properties are met, iff for all $a \in \{0,1\}$:
{\footnotesize
\begin{align*}
  \corrNoComm{corr}&: \forall (u,u',m,aux) \in \underline{r}^a_0\cup \underline{r}^a_1: u \not \in \hat{U} \land u' \not \in \hat{U} \\
  \corrOnlyPartnerSender{corr}&: \forall (u,u',m,aux) \in \underline{r}^a_0\cup \underline{r}^a_1: u \not \in \hat{U} \\
  \corrOnlyPartnerReceiver{corr}&: \forall (u,u',m,aux) \in \underline{r}^a_0\cup \underline{r}^a_1: u' \not \in \hat{U}\\
  \corrStandard{corr}&: \forall \hat{u} \in \hat{U}:r^a_{0_i}=(\hat{u},\_,m,\_) \implies r^a_{1_i}=(\hat{u},\_,m,\_)\\
   &\quad \quad \quad \quad \land r^a_{0_i}=(\_,\hat{u},m,\_) \implies r^a_{1_i}=(\_,\hat{u},m,\_) \\
\end{align*}
}
\end{definition}
\vspace{-0.5cm}

Of course user corruption is not the only important part of an adversary model. Other adversarial capabilities can be adjusted with other parts of our framework (like the corruption of other parts of the \ac{ACN} with protocol queries).

\subsection{Easier Analysis: Quantification}\label{sec:challenges}
For an easier analysis, we allow the quantification of notions in the options. This way a reduced number of challenge rows (challenge complexity) or of challenges (challenge cardinality) can be required. The next section includes information on how results with  low challenge cardinality imply results for higher challenge cardinalities.

\inlineheading{Challenge Complexity}
\example{Consider Alice using a protocol, that achieves $\bohliSSAc$ for one challenge row ($\bohliSSAc_{\mathsf{CR}_{1}}$), but not for two ($\bohliSSAc_{\mathsf{CR}_{2}}$). This means in the case that Alice only communicates once, the adversary is not able to distinguish Alice from any other potential sender Charlie. However, if Alice communicates twice the regime might distinguish her existence from the existence of some other user, e.g. by using an intersection attack.}

To quantify how different the scenarios can be, we add the concept of \emph{challenge complexity}. Challenge complexity is measured in \textit{Challenge rows}, the pairs of communications that differ in the two scenarios as defined earlier. $c$ is the maximal allowed number of challenge rows  in the game. Additionally, we add the maximal allowed numbers of  challenge rows  per challenge $\#cr$  as option to a notion $X$ with $X_{\mathsf{CR}_{\#cr}}$.

\begin{definition}[Challenge Complexity]\label{def:challengeComplexity}
Let $\#\mathsf{CR}$ be the number of challenge rows in this challenge so far, 
We say that the following property is met, iff:
 \allowdisplaybreaks
\begin{align*}
\mathsf{CR}_{\#cr}:& \#\mathsf{CR} \leq \#cr
\end{align*}
\end{definition}
 
  \begin{table}[htb]
\center
\resizebox{0.35\textwidth}{!}{%
  \begin{tabular}{ l l }

Notion including option& Definition \\ \hline
$\manySess{X}$&Properties of $X \land sess$\\
$\corrStandard{X}$, $\corrNoComm{X}$ \ etc.& $\land \corrStandard{corr}$, $\land  \corrNoComm{corr}$ \ etc.\\
$\challengeRows{X}$& Properties of  $ X \land \mathsf{CR}_{\#cr}$\\
\end{tabular}}
 \caption{Definition of notions including the options; for all notions $X$ }
  \label{NotionsDefinition}
\end{table} 

 \inlineheading{Challenges Cardinality}
So far, our definitions focused on one challenge.  We now bound the number of challenges to $n$, as the adversary potentially gains more information the more challenges are played. While challenge complexity defines a bound on the total number of differing rows within a single challenge, cardinality bounds the total number of challenges.
Communications belonging to a challenge are identified by the challenge number $\Psi$, which has to be between 1 and $n$ to be valid. The challenge number is a part of the auxiliary information of the communication  and is only used by the challenger, not by the protocol model.

This dimension of quantification can be useful for analysis, since for certain assumptions the privacy of the $n$-challenge-case can be bounded in the privacy of the single-challenge-case as we will discuss in the next section.

%% file: sections/adversaryModel.tex
\section{Capturing Different Adversaries}
\label{sec:adversary}
The adversary model assumed in the protocol model can be further restricted by adding adversary classes, that filter the information from the adversary to the challenger and vice versa.  Potentially many such adversary classes can be defined. 

\inlineheadingTwo{Adversary Classes}
AnoA introduces adversary classes, i.e.
a PPT algorithm that can modify and filter all in- and outputs from/to the adversary. Adversary classes $\mathcal{C}$ can be included into our framework in exactly the same way: to wrap the adversary $\mathcal{A}$. Instead of sending the queries to  $Ch$, $\mathcal{A}$ sends the queries to  $\mathcal{C}$, which can modify and filter them before sending them to   $Ch$. Similarly, the answers from  $Ch$ are sent to  $\mathcal{C}$ and possibly modified and filtered before sent further to  $\mathcal{A}$. Adversary classes that fulfill reliability, alpha-renaming and simulatability (see  \cite{backes17anoa}  for definitions) are called single-challenge reducible. For such adversary classes it holds that every protocol $\Pi$ that is $(c, \epsilon,\delta)$-X for $\mathcal{C}(\mathcal{A})$, is also  $(n \cdot c, n\cdot \epsilon,n\cdot \epsilon^{n \epsilon}\delta)$-X for $\mathcal{C}(\mathcal{A})$. Even though our framework extends AnoA's in multiple ways, the proof for multi-challenge generalization of AnoA is independent from those extensions and still applies to our framework.

\input{sections/proofs/SingleChallenge.tex}

Note that since the adversary class $\mathcal{C}$ is only a PPT algorithm, $\mathcal{C}(\mathcal{A})$ still is a PPT algorithm and hence, a possible adversary against X when analyzed without an adversary class. So, while adversary classes can help to restrict the capabilities of the adversary, results shown in the case without an adversary class carry over.

\inlineheadingTwo{Using UC-Realizability}
AnoA shows that, if a protocol $\Pi$ UC-realizes an ideal functionality $\mathcal{F}$, which achieves $(c,\epsilon, \delta)-X$, $\Pi$ also achieves ($c,\epsilon, \delta+\delta'$)-X for a negligible $\delta'$. As the proof is based on the $(\epsilon, \delta)$- differential privacy definition of achieving a notion and independent from our extensions to the AnoA framework, this result still holds.

\begin{proofsketch}
AnoA's proof assumes that $\Pi$ does not achieve ($\epsilon, \delta+\delta'$)-X. Hence, there must be an attack $\mathcal{A}$ distinguishing the scenarios. With this, they build a PPT distinguisher $\mathcal{D}$ that uses $\mathcal{A}$ to distinguish $\Pi$ from $\mathcal{F}$. Since, even with our extensions $\mathcal{A}$ still is a PPT algorithm, that can be used to build distinguisher $\mathcal{D}$ and the inequalities that have to be true are the same (since the definition of achieving ($\epsilon, \delta$)-X is the same as being $(\epsilon, \delta)$-differentially private. The combination of $\Pi$ not being ($\epsilon, \delta+\delta'$)-X and $\mathcal{F}$ being $(\epsilon, \delta)-X$ results in the same contradiction as in AnoA's proof.
\end{proofsketch}

%% file: sections/proofs/SingleChallenge.tex
\begin{proofsketch}
The proof is analogous to the one in Appendix A.1 of \cite{backes17anoa}\footnote{Note, although they include the challenge number $n$ in the definition of achieving a notion, this is not used in the theorem.}: 
we only argue that our added concepts
(adaptive corruption, arbitrary sessions and grouping of challenge and input queries to batches) do not change the indistinguishability of the introduced games.

\emph{Adaptive Corruption}
In Games $G_0$ till $G_2$,  $G_3$ till $G_6$ and $G_9$ to $G_{10}$ communications that reach the protocol model are identical. Hence, also adaptive corruption queries between the batch queries will return the same results (if adaptive corruption is used probabilistic: the probability distribution for the results is equal in all these games). $G_2$ to $G_3$ and $G_7$ to $G_8$ by induction hypothesis. $G_6$ to $G_7$ and $G_8$ to $G_9$ adaptive corruption is independent from the used challenge numbers (called challenge tags in \cite{backes17anoa}).

The argumentation for sessions and batches in analogous. Notice that by using the batch concept, in some games a part of the communications of a batch might be simulated, while another part is not. 
\end{proofsketch}

%% file: sections/otherFrameworks.tex
\iflong
\section{Relations to Prior Work}
\label{sec:mappingpapers}
In this Section, we introduce existing frameworks and point out to which of our notions their notions corresponds.
\else
\subsection{Relation to Existing Analysis Frameworks}
\label{sec:mappingpapers}
In this section, we briefly introduce the existing frameworks based on indistinguishability games. 
\fi
\label{sec:ComparingFrameworks}
\iflong
We argue that our framework includes all their assumptions and notions  relevant for \acp{ACN} and thus provides a combined basis for an analysis of \acp{ACN}. 
\else
We argue that our summary of notions includes all their notions\footnote{Where necessary, we have interpreted them for \acp{ACN} and broken them down to the general observable information.}  and therefore allows a comparison along this dimension.
\fi
\label{advClass}
\iflong
For each framework, we first quickly give an idea why the properties and options match the notions of it and focus on how the concepts (like batches) relate later on. The resulting mapping is shown in Table \ref{mapping} and reasoned below.
\else
The resulting mapping is shown in Table \ref{mapping} of Appendix~\ref{app:NotionMapping}.  Since the mapping of our properties to the notions of the other frameworks is obvious in most cases, we reason the remaining cases and concepts here and refer to the long version of this paper \cite{longVersion} for the complete verification.
\fi

\iflong
 \begin{table} [b!]
\center
\resizebox{0.4\textwidth}{!}{%
 \begin{tabular}{ c  c c }

Framework & Notion & Equivalent to \\ \hline
AnoA& $\alpha_{SA}$&$ {\manySess{\static{\corrStandard{\bohliSSAc}}}}_{CR_1}$\\
&$\alpha_{RA}$&${\manySess{\static{\corrStandard{\anoaRA}}}}_{CR_1}$\\
&$ \alpha_{REL}$&${\manySess{\static{\corrStandard{\anoaREL}}}}_{CR_2}$\\
&$ \alpha_{UL}$&${\manySess{\static{\corrStandard{\anoaUL}}}}_{CR_2}$\\
& $\alpha_{sSA}$&$\manySess{\static{\corrStandard{\bohliSSAc}}}$\\
&$\alpha_{sRA}$&$\manySess{\static{\corrStandard{\anoaRA}}}$\\
&$ \alpha_{sREL}$\footnotemark&$\manySess{\static{\corrStandard{\anoaREL}}}$\\
&$ \alpha_{sUL}$\footnotemark&$\manySess{\static{\corrStandard{\anoaUL}}}$\\ \hline
Bohli's&$S/SA=R/SA$&$\bohliSA$\\
&$ R/SUP$& $\bohliRSUP$\\
&$ R/WUP$& $\bohliRWUP$\\
& $ R/PS$&$\bohliRPS$\\
&$ R/SUU$& $\bohliRSUU$\\
&$ R/WUU$&$\bohliRWUU$\\
&$ R/AN$&$\bohliRAN$\\
&$ R/WU$&$\bohliRWU$\\
& $ R/WA$&$\bohliRWA$\\
&$ S/SA^\circ$& $\bohliSSAc$\\
&$ S/SUP^\circ$&$\bohliSSUPc$\\
&$ S/WUP^\circ$&$\bohliSWUPc$\\
&$S/PS^\circ$&$\bohliSPSc$\\
&$S/SUU^\circ$&$\bohliSSUUc$\\
& $S/WUU^\circ$&$\bohliSWUUc$\\
&$S/AN^\circ$&$\bohliSANc$\\
&$S/WU^\circ$&$\bohliSWUc$\\
&$S/WA^\circ$&$\bohliSWAc$\\
&$S/X, R/X^\circ$&analogous\\ 
&$X^+$&$\corrStandard{\langle X\rangle}$\\
& $X^*$&$\corrStandard{\langle X^\circ\rangle}$\\ \hline
Hevia's&$UO$&$\noCorr{\heviaUO}$, $k=1$\\
&$SRA$&$\noCorr{\bohliSA}$, $k=1$\\
&$SA^*$&$\noCorr{\bohliRWU}$, $k=1$\\
&$SA$&$\noCorr{\bohliSSAc}$, $k=1$\\
&$UL$&$\noCorr{\heviaUL}$, $k=1$\\
&$SUL$&$\noCorr{\bohliSWUc}$, $k=1$\\
&$RA^*, RUL, RA$&analogous\\ \hline
Gelernter's&$R^{H, \tau}_{SA}$&$\noCorr{\gelernterSA} \iff \noCorr{\bohliSPSc}$, $k=1$\\
&$R^{H, \tau}_{SUL}$&$\noCorr{\gelernterSUL} \iff \noCorr{\bohliSWAc}$, $k=1$\\
&$R_X$& analogous Hevia: $\langle X\rangle$ \\
&$R^H_X$& analogous Hevia: $\corrNoComm{\langle X\rangle}$\\
&$\hat{R}^H_X$&  analogous Hevia $\corrOnlyPartnerSender{\langle X\rangle}$\\
\end{tabular}}
 \caption{Equivalences, $\langle X \rangle$ equivalence of $X$ used}
  \label{mapping}
\end{table}

\addtocounter{footnote}{-1}
\footnotetext{Under the assumption that in all cases $m_0$ is communicated like in $\alpha_{REL}$ of   \cite{backes17anoa} and in $\alpha_{SREL}$ of one older AnoA version \cite{backes14anoa}.}
\stepcounter{footnote}
\footnotetext{Under the assumption that the receiver in stage 2 can be another than in stage 1 like in $\alpha_{UL}$ of   \cite{backes17anoa}.}
\fi

\paragraph{AnoA Framework}
AnoA \cite{backes17anoa} builds its privacy notions on $(\epsilon, \delta)$~differential privacy and compares them to their interpretation of the terminology paper of Pfitzmann and Hansen \cite{pfitzmann10terminology}.

\iflong
AnoA's $\alpha_{SA}$ allows only one sender to change, the same is achieved with the combination of $\EveryButSender$ and $CR_1$. In AnoA's $\alpha_{RA}$ also the messages can differ, but have to have the same length, which we account for with using $\EveryButReceiverMsg$ and $|M|$. AnoA's $\alpha_{REL}$ will either end in one of the given sender-receiver combinations been chosen ($b=0$) or one of the mixed cases ($b=1$). This is exact the same result as $R_{SR}$ generates. For AnoA's $\alpha_{UL}$ either the same sender is used in both stages or each of the senders is used in one of the stages. This behavior is achieved by our property $T_S$. Although AnoA checks that the message length of the communication of both scenarios is equal, only the first message is used in any possible return result of $\alpha_{UL}$. Hence, not checking the length and requiring the messages to be the same as we do in $T_S$ is neither weaker nor stronger.
\else
\fi

 \iflong
Our model differs from AnoA's model in the batch queries, the adaptive corruption, the arbitrary  sessions and the use of notions instead of anonymity functions.
\else
Conceptually our model differs from AnoA's model in the definition of achieving a notion, batch queries, and the use of notions instead of anonymity functions.
\fi
\iflong
Instead of  \emph{batch queries} AnoA distinguishes between input, i.e.
communications that are equal for both scenarios, and challenge queries,  i.e.
challenge rows.
Input queries are always valid in AnoA. 
They are also valid in our model, because all the privacy aspects used for our notions equivalent to AnoA's hold true for identical batches without $\Diamond$ and $\Diamond$ is not allowed in the equivalent notions. 
\else
AnoA's definition of achieving a notion can be easily included (see Appendix \ref{app:epsilonDef}), if needed.
\fi
\iflong
In AnoA's single-message anonymity functions only a limited number of challenge queries, i.e. challenge rows, is allowed per challenge.
We ensure this restriction with $CR_{\#cr}$.
\fi
In AnoA, the adversary gets information after every communication.
 This is equivalent to multiple batches of size one in our case.
\iflong
We assume that for the analyzed protocol a protocol model can be created, which reveals the same or less information when it is invoked on a sequence of communications at once instead of being invoked for every single communication.
Our notions, which match the AnoA notions, allow for batches of size one.
 So, our batch concept neither strengthens nor weakens the adversary.
\fi

\iflong
AnoA's \emph{corruption} is static, does not protect corrupted users\footnote{Although AnoA does not explicitly state this, we understand the analysis and notions of AnoA this way, as scenarios differing in the messages corrupted users send/receive could be trivially distinguished.} and AnoA includes restrictions on \emph{sessions}.
Hence, AnoA's notions translate to ours with the static corruption  $\static{X}$, the corrupted communication have to be equal in both scenarios $\corrStandard{X}$ and the session option of our model $\manySess{X}$.
\fi

AnoA's challenger does not only check properties, but modifies the batches with the \emph{anonymity functions}. However, the modification results in one  of at most four batches. We require  those four batches (as combination of scenario and instances) as input from the adversary, because it is more intuitive that all possible scenarios stem from the adversary. This neither increases nor reduces the information the adversary learns, since she knows the challenger algorithm.

\vspace{-0.5cm}
\paragraph{Bohli's Framework}
Bohli and Pashalidis \cite{bohli11relations} build a hierarchy of application-independent privacy notions based on what they define as ``interesting properties'',  that the adversary is or is not allowed to learn.
Additionally, they compare their notions to Hevia's, which we introduce next,  and find equivalences.

\iflong
It is easy to see, that our definitions of $U,Q, H$ ($P$ is not easy and hence, explained more detailed below) match the ones of Bohli's properties (who sent, how often any sender sends and how many senders sent how often) although we do not use a function that links every output message with the sender(/receiver), but the sender-messages-sets(/receiver-messages-sets). 
Bohli and Pashalidis additionally define the restriction of picking their communications equal except for the user (depending on the current notion sender or receiver) $\circ$. This is the same as allowing only the senders  resp. receivers to differ ($\EveryButSender$ resp. $\EveryButRec$).

Conceptually, our model differs from Bohli's model in the concept of challenges, the advantage definition, the order of outputs, and the allowed behavior of corrupted users.
\else
\fi

\iflong
Bohli's notions can be understood as one \emph{challenge} ($n=1$) with arbitrarily many challenge rows (any $c$).
Further, it does not use a multiplicative term in its \emph{advantage} ($\epsilon=0$).
Then $\delta$ equals the advantage, which has to be 0 to unconditionally provide a privacy notion or negligible to computationally provide this notion.
\else
\fi

\iflong
\else
To achieve the mapping, we need to interpret one property of Bohli's framework for \acp{ACN}. Our message partitionings ($P,P'$) group the messages by their sender/receiver. However, Bohli's corresponding linking relation groups the indexes of the outputs of the analyzed system. Since messages are usually the interesting output elements, the adversary tries to link in \acp{ACN}; we consider this as a suitable mapping when analyzing \acp{ACN}.
\fi
\iflong
 Bohli's framework assumes that the protocol outputs information as an information vector, where each entry belongs exactly to one communication.
The adversary's goal in Bohli's framework is to link the index number of the output vector with the sender or receiver of the corresponding communication.

All except one of their properties can be determined given the batches of both scenarios.
However, the linking relation property that partitions the index numbers of the output vector by user (sender or receiver depending on the notion), can only be calculated once the output order is known.
Since our notions shall be independent from the analyzed protocol, the challenger cannot know the protocol and the way the output order is determined.
Running the protocol on both scenarios might falsely result in differing output orders for non-deterministic protocols.

 Thus, we adapt the  linking relation for \acp{ACN}  to be computable based on the batches.
The interesting output elements the adversary tries to link in \acp{ACN} are messages.
Hence, here the linking relation partitions the set of all messages into the sets of messages sent/received by the same user, which can be calculated based on the batches.
This adaption is more restrictive for an adversary, since the partition of output numbers can be equal for both scenarios even though the sent messages are not.
However, if the adversary is able to link the output number to the message, she can calculate our new linking relations based on Bohli's.
\fi

\iflong
Further, Bohli's framework allows for notions, where the \emph{behavior of corrupted users} differs in the two scenarios.
This means privacy of corrupted users is provided, i.e.
the adversary wins if she can observe the behavior of corrupted users. 
Those notions are the ones without the option $\corrStandard{X}$.

To match our batch query, Bohli's input queries, which include communications of both scenarios, have to be combined with a nextBatch query, which signals to hand all previous inputs to the protocol.
\fi


\vspace{-0.5cm}
\paragraph{Hevia's Framework}
Hevia and Micciancio \cite{hevia08indistinguishability} define scenarios based on message matrices.
Those message matrices specify who sends what message to whom. 
Notions restrict different communication properties like the number or set of sent/received messages per fixed user, or the number of total messages.
Further, they construct a hierarchy of their notions and give optimal \ac{ACN} protocol transformations that, when applied, lead from weaker to stronger notions.

\iflong
Mapping of the properties follows mainly from Bohli's and the equivalences between Bohli and Hevia (including the one we correct in the following paragraph). Besides this,  only Hevia's Unobservability ($UO$), where the matrices  can be picked arbitrary, is new. However, this corresponds to our \nothing \ property, that always returns TRUE and allows any arbitrary scenarios.

Our model differs from Hevia's, since ours considers the order of communications,  allows adaptive attacks and corruption.

Our game allows to consider the \emph{order of communications}.
\else
In contrast, our model considers the order of communications. 
\fi
Analyzing protocol models that ignore the order will lead to identical results. However, protocol models that consider the order do not achieve a notion -- although they would in Hevia's framework, if an attack based on the order exists.
\iflong \footnote{
Creating an adapted version left a degree of freedom.
Our choice of adaptation  corresponds with the interpretation of Hevia's framework that was used, but not made explicit in Bohli's framework.}\fi

\iflong
Most of Hevia's notions are already shown to match Bohli's with only one batch ($k=1$) and no corruption ($\noCorr{X}$)  \cite{bohli11relations}.
\else
Most of Hevia's notions are already shown to match Bohli's with only one batch \cite{bohli11relations}.
\fi
However, we have to correct two mappings:
in \cite{bohli11relations} Hevia's strong sender anonymity (${SA^*}$), which requires  the number of messages a receiver receives  to be the same in both scenarios was mistakenly matched to  Bohli's  sender weak unlinkability ($S/WU^+$), in which every sender sends the same number of messages in both scenarios.
\iflong 
The needed restriction is realized in Bohli's $R/WU^+$ instead.
The proof is analogous to Lemma 4.3 in \cite{bohli11relations}.
\else
 Hence, the sender and receiver restrictions become confused and it needs to be mapped to Bohli's receiver weak unlinkability ($R/WU^+$) instead.
\fi
The same reasoning leads to Bohli's  sender weak unlinkability ($S/WU^+$) as the mapping for Hevia's strong receiver anonymity (${RA^*}$).

\paragraph{Gelernter's Framework}
Gelernter and Herzberg \cite{gelernter13limits} extend Hevia's framework to include corrupted participants.
Additionally, they show that under this strong adversary  an \ac{ACN} protocol achieving the strongest notions exists.
However, they prove that any \ac{ACN} protocol with this strength has to be inefficient, i.e.
the message overhead is at least linear in the number of honest senders.
Further, they introduce  relaxed privacy notions that can be efficiently achieved.

\iflong
The notions of Gelernter's framework build on Hevia's and add corruption, which is covered in our corruption options.
Only the relaxed notions $R^{H, \tau}_{SA}$ and $R^{H, \tau}_{SUL}$ are not solely a corruption restriction.
\else
The notions of Gelernter's framework build on Hevia's and add corruption, which we do not discuss in  this work, but include in the long version of this paper \cite{longVersion}. However,  the relaxed notions are not solely an extension regarding corruption.
\fi
\iflong
We define new notions  as $\gelernterSA$=\something $\land G$ and  $\gelernterSUL$= \something $\land Q \land G$ that are equivalent to some of the already introduced notions to make the mapping to the Gelernter's notions obvious.
 They use a new property $G$, in which scenarios are only allowed to differ in the sender names.
 \begin{definition}[Property $G$]
 Let $\mathcal{U}$ be the set of all senders,  $s_{b_i}= \{(u,\{m_1, \dots ,m_h\})\bigm| u $ send message $m_1, \dots , m_h$ in batch $i\}$  the sender-messages sets for scenario $b\in \{0,1\}$.
We say that $G$ is met, iff a permutation $perm$ on $\mathcal{U}$ exists such that for all $ (u,M)\in s_{0_k}: (\text{perm}(u),M) \in s_{1_k}$.
 \end{definition}
Note that  Gelernter's relaxed notions (indistinguishability between permuted scenarios) is described by our property $G$, the need for the existence of such a permutation.

\begin{theorem}It holds that
\begin{align*}
(c, \epsilon, \delta)-\gelernterSA &\iff (c, \epsilon, \delta)-\bohliSPSc,\\
(c, \epsilon, \delta)-\gelernterSUL &\iff (c, \epsilon, \delta)-\bohliSWAc\text{.}
\end{align*}
\end{theorem}

\begin{proofsketch}
Analogous to Theorem \ref{loopixUO}.

$\gelernterSA \Rightarrow \bohliSPSc$: Every attack on $\bohliSPSc$ is valid against $\gelernterSA$: Since $P$ is fulfilled, for every sender  $u_0$ in the first scenario, there exists a sender $\tilde{u}_0$ in the second scenario sending the same messages.
Hence, the permutation between senders of the first and second scenario exists.

	$\gelernterSA \Leftarrow \bohliSPSc$: Every attack on $\gelernterSA$ is valid against $\bohliSPSc$: Since there exists a permutation between the senders of the first and second scenario sending the same messages, the partitions of messages sent by the same sender are equal in both scenarios, i.e.
$P$ is fulfilled.

$\gelernterSUL \iff \bohliSWAc:$
	 $Q$ is required in both notions by definition.
Arguing that $P$ resp.
$G$ is fulfilled given the other attack is analogous to $\gelernterSA \iff \bohliSPSc$.

\end{proofsketch}

\else
In Appendix \ref{OtherFrameworksProof} we formalize them and shown to be equivalent to two previously defined notions.
\fi

%% file: sections/hierarchy.tex
\vspace{-2em}
\section{Hierarchy}
\label{sec:hierarchy}
\vspace{-1em}

\iflong
\begin{figure*}[h!]
\center
\includegraphics[width=0.95\textwidth]{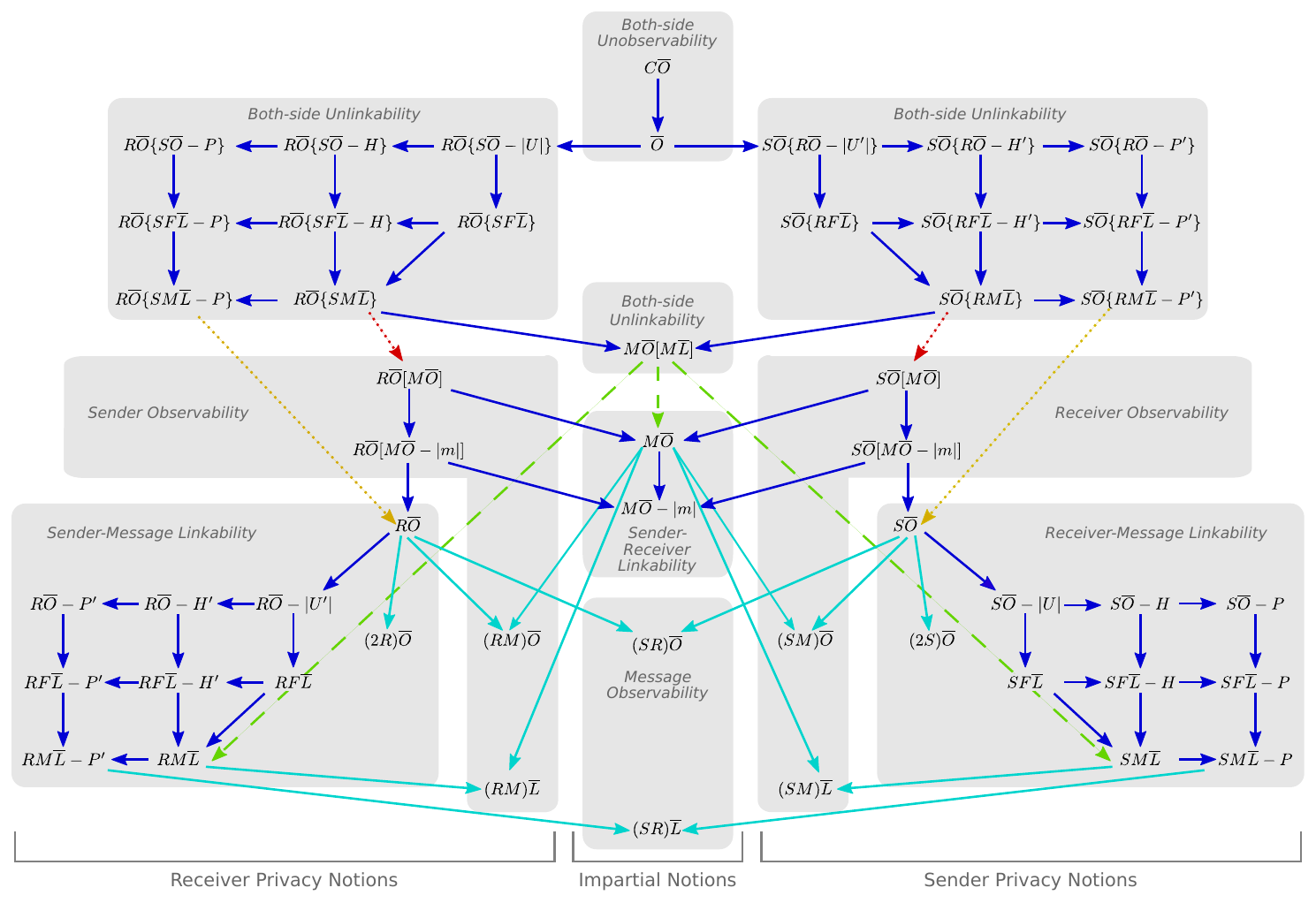}
\caption{Our new hierarchy of privacy notions divided into sender, receiver and impartial notions and clustered by leakage type}\label{fig:hierarchyColored}
\end{figure*}
\else
\begin{figure*}[h!]
\center
\includegraphics[width=0.95\textwidth]{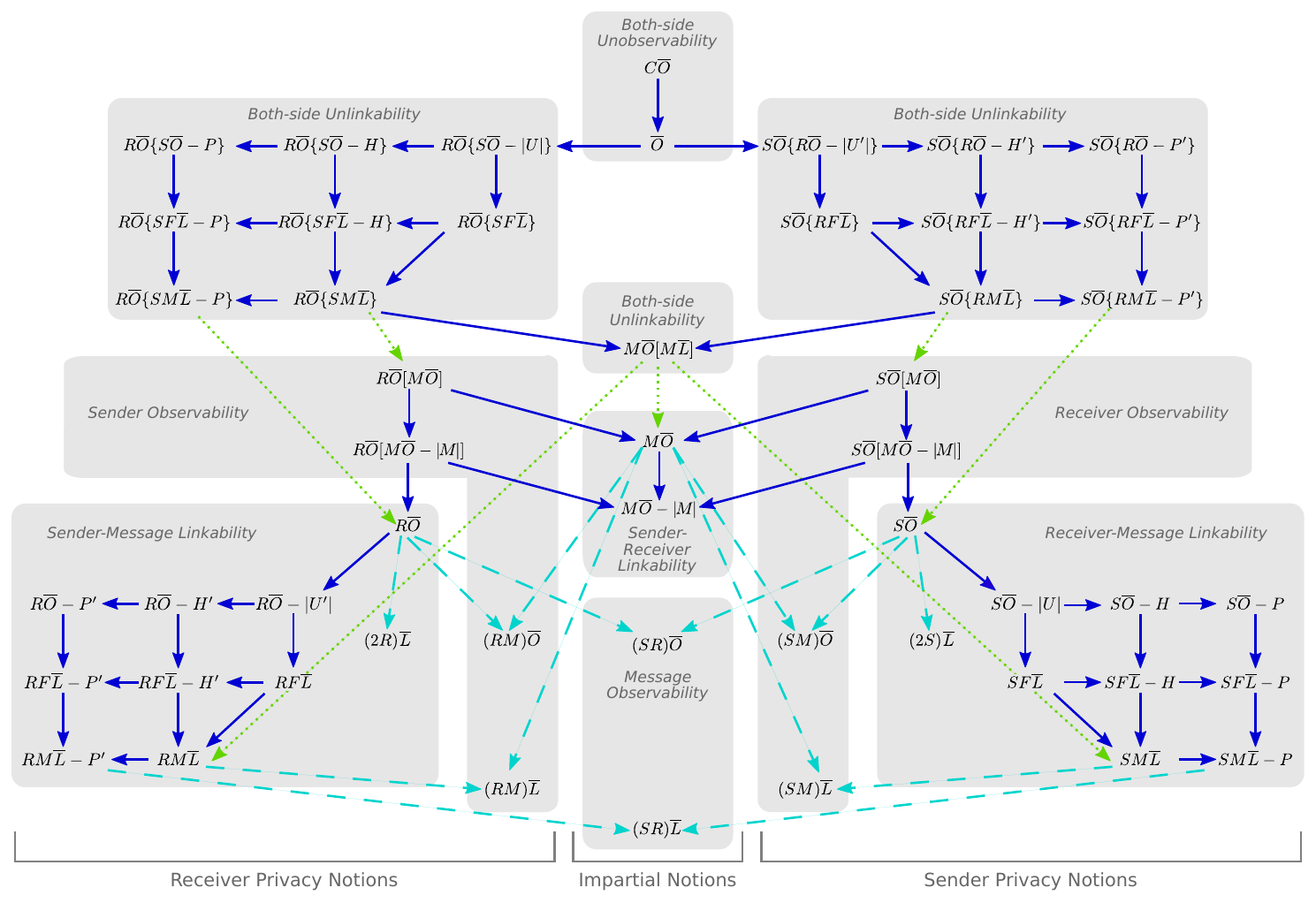}
\caption{Our hierarchy of privacy notions divided into sender, receiver and impartial notions and clustered by leakage type. Table~\ref{NotionsDefinition} provides definitions for the presented notions based on properties. Table~\ref{tab:allNotions} gives an overview on all properties. For a summary of the naming scheme, see Table \ref{Tab:NamingScheme} of Appendix \ref{app:summaryNamingScheme}. } \label{fig:hierarchyColored}
\end{figure*}
\fi

\iflong
\begin{figure}[thb]
\begin{center}
\resizebox{0.25\textwidth}{!}{%
\begin{tikzpicture}[>=triangle 45,font=\sffamily]
        \node (X) at (0,0) {$\corrStandard{X}$};
    \node (XSpecial) [above =0.3cm of X] {$X \iff X_{CR_c}$};
    \node (corrNoComm)[below  = 0.6cm of X] {$\corrNoComm{X}$};
    \node (invis)[below  = 0.1cm of X] {};
    \node(corrSender)[left = 0.5cm of X] {$\corrOnlyPartnerSender{X}$};
    \node(corrRec)[right = 0.5cm of X] {$\corrOnlyPartnerReceiver{X}$};
    \node (static)[right = 0.5cm of corrRec] {$\static{X}$};
    \node (noCorr)[right = 1.5cm of corrNoComm] {$\noCorr{X}$};
    \node (manySess) [ left=0.5cm of corrSender] {$\manySess{X}$};
    
    \draw [semithick,->] (XSpecial) -- (static);
        \draw [semithick,->] (XSpecial) -- (X);
    \draw [semithick,->] (XSpecial) -- (corrSender);
    \draw [semithick,->] (XSpecial) -- (corrRec);
        \draw [semithick,->] (corrSender)  -- (corrNoComm);
    \draw [semithick,->] (corrRec) -- (corrNoComm);
     \draw [semithick,->] (X) -- (corrNoComm);
    \draw [semithick,->] (static) -- (noCorr);
        \draw [semithick,->] (corrNoComm) -- (noCorr);
    \draw [semithick,->] (XSpecial) -- (manySess);
 
\end{tikzpicture}}
\caption{Additional implications for corruption and sessions} \label{HierarchyExtended}
\end{center}
\end{figure}
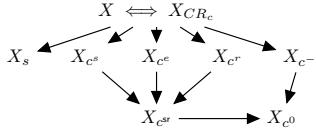
\fi
Next, we want to compare all notions and establish their hierarchy. To do this, for any pair of notions we analyze which one is stronger than, i.e. implies, the other. This means, any \ac{ACN} achieving the stronger notion also achieves the weaker (implied) one. Our result is shown in Figure~\ref{fig:hierarchyColored}, where all arrow types represent implications, and is proven as Theorem \ref{the:impl} below. 
Further, obvious implications between every notion $\sgame{X}$, $\rgame{X}$ and $X$ exist, since $\sgame{X}$ only adds more possibilities to distinguish the scenarios. However, to avoid clutter we do not show them in Figure  \ref{fig:hierarchyColored}. 
\iflong
To ease understanding the hierarchy for the first time, we added Appendix  \ref{sec:HierarchyAndTables} where it is plotted together with the most important symbol tables. Further, the same hierarchy exists between  notions with the same session, corruption and quantification options. 

Further, we add a small hierarchy for the options that holds by definition in Figure \ref{HierarchyExtended}.
\fi


\begin{theorem}
\label{the:impl}
The implications shown in Figure \ref{fig:hierarchyColored} hold.
\end{theorem}
\iflong
\input{sections/proofs/hierarchyProof.tex}
So far we have proven that implications between notions exist. Further, we assure that the hierarchy is complete, i.e. that there exist no more implications between the notions of the hierarchy:
\else
\input{sections/proofs/hierarchySketch.tex}
\fi

%
%
%
\begin{theorem}
\label{the:noImpl}
For all notions $X_1$ and $X_2$ of our hierarchy, where  $X_1 \implies X_2$ is not proven or implied by transitivity, there exists an \ac{ACN} protocol achieving $X_1$, but not $X_2$.
\end{theorem}
\iflong
\input{sections/proofs/completeProof.tex}
\else
\begin{proofsketch}
We construct the protocol in the following way: Given a protocol $\Pi$ that achieves $X'_1$ ($X_1$ itself or a notion that implies $X_1$), let protocol $\Pi'$ run $\Pi$ and additionally output some information $I$. We argue that learning $I$ does not lead to any advantage in distinguishing the scenarios for $X_1$. Hence, $\Pi'$ achieves $X_1$. We give an attack against $X_2$ where learning $I$ allows the scenarios to be distinguished. Hence, $\Pi'$ does not achieve $X_2$.  
Further, we use the knowledge that $\implies$ is transitive\footnote{If $X_1 \implies X_2$ and $X_1 \centernot \implies X_3$, it follows that $X_2 \centernot \implies X_3$.}.

Some concrete cases are shown in Appendix \ref{completeSketch}. We provide the complete list of proofs in the long version of this paper \cite{longVersion}.
\end{proofsketch}
\fi

%% file: sections/proofs/hierarchyProof.tex
\begin{proof}

Analogous to Theorem \ref{loopixUO}: We prove every implication $X_1 \Rightarrow X_2$ by an indirect proof of the following outline: Given an attack on $X_2$, we can construct an attack on $X_1$ with the same success. Assume a protocol has $X_1$, but not $X_2$. Because it does not achieve $X_2$, there exists a successful attack on $X_2$. However, this implies that there exists a successful attack on $X_1$ (we even know how to construct it). This contradicts that the protocol has $X_1$.\footnote{In AnoA, Bohli's and Hevia's framework some of these implications are proved for their notions in the same way.} Due to this construction in the proof the implications are transitive.

We use different arrow styles in Figure \ref{fig:hierarchyColored} to partition the implications into those with analogous proofs.

 \begin{tikzpicture}[>=triangle 45,font=\sffamily]
    \node (A) at (0,0) {};
    \node (B)[right = 1cm of A] {};
    \draw [semithick,line width=0.5mm,->,green,  dashed] (A) -- (B);
\end{tikzpicture}
The dashed, green implications 
hold, because of the following and analogous proofs:

\begin{claim}($\heviaUL \implies \bohliSWUc$) \label{claimGreen}
If protocol $\Pi$ achieves $(c, \epsilon, \delta)- \heviaUL$, it achieves $(c, \epsilon, \delta)-\bohliSWUc$. 
\end{claim}
\begin{proof}
Given a valid attack $\mathcal{A}$ on $\bohliSWUc$. We show that $\mathcal{A}$ is a valid attack on $\heviaUL$:
Because of $\bohliSWUc$,  $Q$ is fulfilled. Because of $\EveryButSender$, the receivers of the communications input to the protocol are the same in both scenarios. Hence, every receiver receives the same number of messages, i.e. $Q'$ is fulfilled. 
So, every attack against $(c, \epsilon, \delta)-\bohliSWUc$ is valid against $(c, \epsilon, \delta)- \heviaUL$. 

Now, assume a protocol $\Pi$ that achieves  $(c, \epsilon, \delta)- \heviaUL$, but not $(c, \epsilon, \delta)-\bohliSWUc$. Because it does not achieve  $(c, \epsilon, \delta)-\bohliSWUc$, there has to exist an successful attack $\mathcal{A}$ on $(c, \epsilon, \delta)-\bohliSWUc$, i.e.\[
\text{Pr}[0= \langle \mathcal{A} \bigm| Ch(\Pi, \bohliSWUc,c,0)\rangle ] >\]\[e^{\epsilon} \cdot \text{Pr}[0= \langle \mathcal{A} \bigm| Ch(\Pi, \bohliSWUc, c,1)\rangle]+ \delta \text{.}\] 
We know $\mathcal{A}$ is also valid against $(c, \epsilon, \delta)- \heviaUL$. Thus, it exists a attack, with \[\text{Pr}[0= \langle \mathcal{A} \bigm| Ch(\Pi, \heviaUL,c,0)\rangle ] >\]\[e^{\epsilon} \cdot \text{Pr}[0= \langle \mathcal{A} \bigm| Ch(\Pi, \heviaUL, c,1)\rangle]+ \delta \text{,} \] which contradicts the assumption that $\Pi$ achieves $(c, \epsilon, \delta)- \heviaUL$.
\end{proof}

 \begin{tikzpicture}[>=triangle 45,font=\sffamily]
    \node (A) at (0,0) {};
    \node (B)[right = 1cm of A] {};
    \draw [semithick,line width=0.5mm,->,yellow, dotted] (A) -- (B);
\end{tikzpicture}
The dotted, yellow implications hold, because of the following and analogous proofs:

\begin{claim}($\bohliRWA \implies \bohliSSAc$)
If protocol $\Pi$ achieves $(c, \epsilon, \delta)-\bohliRWA$, it achieves $(c, \epsilon, \delta)-\bohliSSAc$. 
\end{claim}
\begin{proof}
Given a valid attack $\mathcal{A}$ on $\bohliSSAc$. We show that $\mathcal{A}$ is a valid attack on $\bohliRWA$:
Because of $\EveryButSender$  of $\bohliSSAc$ the receiver-message pairs of the communications input to the protocol are the same in both scenarios. Hence, every receiver receives the same messages, i.e.  $Q'$ and $P'$ are  fulfilled. 
So, every attack against $(c, \epsilon, \delta)-\bohliSSAc$ is valid against $(c, \epsilon, \delta)- \bohliRWA$. 

Now, the proof by contradiction is done analogous to the proof of Claim \ref{claimGreen}.
\end{proof}

 \begin{tikzpicture}[>=triangle 45,font=\sffamily]
    \node (A) at (0,0) {};
    \node (B)[right = 1cm of A] {};
    \draw [semithick,line width=0.5mm,->,blue] (A) -- (B);
\end{tikzpicture}
All dark blue implications 
 $(c, \epsilon, \delta)-X_1 \Rightarrow (c, \epsilon, \delta)-X_2$ follow from the definition of the notions: Every valid attack against $X_2$ is valid against $X_1$. This holds because $U \Rightarrow |U|$, $H \Rightarrow |U|$, $Q \Rightarrow U$, $P \Rightarrow H$ , $Q \Rightarrow H$, \something \ $ \Rightarrow $ \nothing\ and obviously for any properties $A$ and $B$: $A \land B \Rightarrow A$ resp.  $A \land B \Rightarrow B$. 


 \begin{tikzpicture}[>=triangle 45,font=\sffamily]
    \node (A) at (0,0) {};
    \node (B)[right = 1cm of A] {};
    \draw [semithick,line width=0.5mm,->,red, dotted] (A) -- (B);
\end{tikzpicture}
The dotted, red implications $(c, \epsilon, \delta)-X_1 \Rightarrow (c, \epsilon, \delta)- X_2$ hold, because of the following and analogous proofs:

\begin{claim}($\bohliSWU \implies \anoaRAWOL$)
If protocol $\Pi$ achieves $(c, \epsilon, \delta)-\bohliSWU$, it achieves $(c, \epsilon, \delta)-\anoaRAWOL$. 
\end{claim}
\begin{proof}
Given attack $\mathcal{A}_1$ on $(c,\epsilon,\delta)\anoaRAWOL$. We show that $\mathcal{A}_1$  is valid against $(c, \epsilon, \delta)-\bohliSWU$:
Sending nothing is not allowed in $\anoaRAWOL$ and hence, will not happen (\something) and because of $\EveryButReceiverMsg$, every sender sends equally often in both scenarios, i.e. Q is fulfilled.

Now, the proof by contradiction is done analogous to the proof of Claim \ref{claimGreen}.
\end{proof}

 \begin{tikzpicture}[>=triangle 45,font=\sffamily]
    \node (A) at (0,0) {};
    \node (B)[right = 1cm of A] {};
    \draw [semithick,line width=0.5mm,->,cyan] (A) -- (B);
\end{tikzpicture}
The cyan implications $(c, \epsilon, \delta)-X_1 \Rightarrow (c, \epsilon, 2\delta)- X_2$ hold, because of the following and analogous proofs \footnote{For $\bohliSSAc \Rightarrow\anoaREL $ (resp. $\bohliSSAc  \Rightarrow \newSAUO$) pick challenge rows differently; for $b=0: a=a'$  and for $b=1:a=1-a'$ to ensure that receivers (resp. messages) are equal.\\
For $ (2c, \epsilon,  \delta)-\bohliSWUc \Rightarrow (c,\epsilon, \delta)- \newSA$, \mbox{(resp. $\bohliRWUc \Rightarrow \newRA$, $\bohliSWAc \Rightarrow \loopixSRTPU$)} replace the challenge row with the corresponding two rows.}:

\begin{claim}($\bohliRSAc \implies \anoaREL$)
If protocol $\Pi$ achieves $(c, \epsilon, \delta)-\bohliRSAc$, it achieves $(c, \epsilon, 2\delta)-\anoaREL$.
\end{claim}
\begin{proof}
We first argue the case of one challenge and later on extend it to multiple challenges.
Given attack $\mathcal{A}_2$ on $(1, \epsilon, 2\delta)-\anoaREL$. We construct two attacks $\mathcal{A}_1'$ and $\mathcal{A}_1''$ against $\bohliRSAc$ and show that one of those has at least the desired success.

We construct attacks  $\mathcal{A}_1'$ and $\mathcal{A}_1''$.  We therefore pick $a'=0$ and $a''=1$.  Those shall replace $a$, which would be picked randomly by the challenger in $\anoaREL$ to determine the batch instance.
In $\mathcal{A}_1'$ we use the communications of $\mathcal{A}_2$ corresponding to $a'=0$ (for $b=0$ and $b=1$) as challenge row, whenever a batch in $\mathcal{A}_2$ includes a challenge row. In $\mathcal{A}_1''$ we analogously use the communications corresponding to $a''=1$.

We show that both $\mathcal{A}_1'$ and $\mathcal{A}_1''$ are valid against $(1, \epsilon, \delta)-\bohliRSAc$:
Sending nothing is also not allowed in $\anoaREL$ and hence, will not happen (\something) and because of the fixed $a=a'$ or $a=a''$,  the senders of challenge rows are the same in both scenarios. Since also messages are equal  in  $\anoaREL$, the sender-message pairs are fixed ($\EveryButRec$). Hence, $\mathcal{A}_1'$ and  $\mathcal{A}_1''$ are valid against $\bohliRSAc$. 

Since $\mathcal{A}_2$ is an successful attack on $(1, \epsilon, 2\delta)-\anoaREL$ and $\mathcal{A}_1'$ and $\mathcal{A}_1''$ against $\bohliRSAc$ only fix the otherwise randomly picked $a$:
\begin{align*}
0.5&\text{Pr}[0= \langle \mathcal{A}_1' \bigm| Ch(\Pi, \bohliRSAc,c,0)\rangle ]+\\
 0.5 &\text{Pr}[0= \langle \mathcal{A}_1'' \bigm| Ch(\Pi, \bohliRSAc,c,0)\rangle ] \\
 >&e^{\epsilon} \cdot(0.5\text{Pr}[0= \langle \mathcal{A}_1' \bigm| Ch(\Pi, \bohliRSAc,c,1)\rangle ]+\\
  0.5 &\text{Pr}[0= \langle \mathcal{A}_1'' \bigm| Ch(\Pi, \bohliRSAc,c,1)\rangle ]) +2 \delta\text{.}
\end{align*}

So,
\begin{align*}
 0.5&\text{Pr}[0= \langle \mathcal{A}_1' \bigm| Ch(\Pi, \bohliRSAc,c,0)\rangle ]\\
 >&e^{\epsilon} \cdot(0.5\text{Pr}[0= \langle \mathcal{A}_1' \bigm| Ch(\Pi, \bohliRSAc,c,1)\rangle ]+  \delta\\
 &\text{or}\\ 
 0.5&\text{Pr}[0= \langle \mathcal{A}_1'' \bigm| Ch(\Pi, \bohliRSAc,c,0)\rangle ] \\
 >&e^{\epsilon} \cdot(0.5\text{Pr}[0= \langle \mathcal{A}_1'' \bigm| Ch(\Pi, \bohliRSAc,c,1)\rangle ]+ \delta\\
 \end{align*}
 must hold true (otherwise we get a contradiction with the above inequality). Hence, $\mathcal{A}_1'$ or $\mathcal{A}_1''$ has to successfully break $(1, \epsilon, \delta)-\bohliRSAc$.

In case of multiple challenges: the instance bit is picked randomly for every challenge. Hence, we need to construct one attack for every possible combination of instance bit picks, i.e. $2^c$ attacks in total\footnote{Note that this does not contradict the PPT requirement of our definition as only finding the right attack is theoretically exponential in the number of challenges allowed. However, the attack itself is still PPT (and might be even easier to find for a concrete protocol).} from which each is a PPT algorithm and at least one is at least as successful as the attack on $(c, \epsilon, 2\delta)-\anoaREL$.

Now, the proof by contradiction is done analogous to the proof of Claim \ref{claimGreen}.
\end{proof}

\inlineheading{Remark}
$c$ can be any value in the proofs (esp. $c=1$ or $c>1$) the proposed constructions apply changes for each challenge row.

Additionally, the corrupt queries are not changed by the proposed constructions. Hence, the implications hold true between those notions as long as they have the same corruption options. Analogously sessions are not modified by the constructions and the same implications hold true between notions with equal session options. 
 
\end{proof}

\iflong
Additional implications based on corruption and sessions are shown in Figure \ref{HierarchyExtended}. Most of them hold by definition. Only the equivalence with and without challenge row restriction  per challenge is not so easy to see and proven below.
\fi 
\begin{theorem}
\label{CRTheorem}
For all $X$: 
\begin{enumerate}
\item $(c,\epsilon, \delta)-X \Rightarrow (c,\epsilon,\delta)-X_{CR_{c'}}$. 
\item $(c,\epsilon, \delta)-X \Leftarrow (c,\epsilon,\delta)-X_{CR_{c'}}$, if number of challenges not restricted. 
\end{enumerate}
\end{theorem}

\begin{proof}
1. Trivial: Given an attack valid on $X$ restricted regarding the challenge rows per challenge. This attack is also valid against $X$ without challenge row restriction. 

2. We need to construct a new attack:
\begin{translation}\label{BohliToAnoA}
Given an attack  $\mathcal{A}_2$, we construct attack $\mathcal{A}_1$. 
Let $\bar{n}$ be the number of previous challenges used in $\mathcal{A}_1$ so far. 
For every batch query $bq$ with $n''$ challenge rows replace the challenge tag of the $1$st $,\dots, c'$st challenge row with $\bar{n}+1 $; continue with the next $c'$ challenge rows and the increased challenge tag until no challenge rows are left.
Use all other queries as $\mathcal{A}_2$ does, give the answers to $\mathcal{A}_2$ and output whatever $\mathcal{A}_2$ outputs.
\end{translation}

Given  attack $\mathcal{A}_2$ on $X$. We construct an attack $\mathcal{A}_1$ with the same success against $X_{CR_1}$ by using Attack Construction \ref{BohliToAnoA}.
We show that $\mathcal{A}_1$ is valid against $X_{CR_1}$:
Attack Construction \ref{BohliToAnoA} assures, that at most $c'$ challenge rows are used in every challenge ($CR_{c'}$), all other aspects of $X$ are fullfilled in $\mathcal{A}_2$, too.
Since $\mathcal{A}_1$ perfectly simulates the given attack $\mathcal{A}_2$, it has the same success.

Now, the proof by contradiction is done analogous to the proof of Claim \ref{claimGreen}.

\end{proof}

%

%% file: sections/proofs/hierarchySketch.tex
\begin{proofsketch}
We prove every implication $X_1 \Rightarrow X_2$ by an indirect proof of the following outline: Given an attack on $X_2$, we can construct an attack on $X_1$ with the same success. Assume a protocol has $X_1$, but not $X_2$. Because it does not achieve $X_2$, there exists a successful attack on $X_2$. However, this implies that there exists a successful attack on $X_1$ (we even know how to construct it). This contradicts that the protocol has $X_1$.\footnote{In AnoA, Bohli's and Hevia's framework some of these implications are proved for their notions in the same way.} Due to this construction of the proof, the implications are transitive.

We use different arrow styles in Figure \ref{fig:hierarchyColored} to partition the implications into those with analogous proofs.

 \begin{tikzpicture}[>=triangle 45,font=\sffamily]
    \node (A) at (0,0) {};
    \node (B)[right = 1cm of A] {};
    \draw [semithick,line width=0.5mm,->,blue] (A) -- (B);
\end{tikzpicture}
 follow from the definition of the notions.

 \begin{tikzpicture}[>=triangle 45,font=\sffamily]
    \node (A) at (0,0) {};
    \node (B)[right = 1cm of A] {};
    \draw [semithick,line width=0.5mm,->,green, dotted] (A) -- (B);
\end{tikzpicture}
hold, because of the following and analogous arguments:
every attack against $ \bohliSSAc$  is valid against $ \bohliRWA$: Because of $\EveryButSender$ the receiver-message pairs of the communications input to the protocol are the same in both scenarios. Hence, every receiver receives the same messages, i.e.  $Q'$ and $P'$ are  fulfilled. 



 \begin{tikzpicture}[>=triangle 45,font=\sffamily]
    \node (A) at (0,0) {};
    \node (B)[right = 1cm of A] {};
    \draw [semithick,line width=0.5mm,->,cyan, dashed] (A) -- (B);
\end{tikzpicture}
$X_1 \Rightarrow  X_2$ hold, because of the following and analogous arguments:\footnote{For $\bohliSSAc \Rightarrow\anoaREL $ (or $\bohliSSAc  \Rightarrow \newSAUO$) pick challenge rows differently; for $b=0: a=a'$ and for $b=1:a=1-a'$ to ensure that receivers (or messages) are equal.\\
For \mbox{$ \bohliSWUc \Rightarrow \newSA$}, (or \mbox{$\bohliRWUc \Rightarrow \newRA$}, \mbox{$\bohliSWAc \Rightarrow \loopixSRTPU$}) replace the challenge row with the corresponding two rows.}
given attack $\mathcal{A}_2$ on $\anoaREL$. We construct two attacks $\mathcal{A}_1'$ and $\mathcal{A}_1''$ against $\bohliRSAc$ and show that one of those has at least the desired success.

We construct attacks  $\mathcal{A}_1'$ and $\mathcal{A}_1''$ by picking $a'=0$ and $a''=1$.  These shall replace $a$, which would be picked randomly by the challenger in $\anoaREL$ to determine the instance.
In $\mathcal{A}_1'$ we use the communications of $\mathcal{A}_2$ corresponding to $a'=0$ (for $b=0$ and $b=1$) as the challenge row, whenever a batch in $\mathcal{A}_2$ includes a challenge row. In $\mathcal{A}_1''$ we analogously use the communications corresponding to $a''=1$.

$\mathcal{A}_1'$ and $\mathcal{A}_1''$ are valid against $\bohliRSAc$:
Because of the fixed $a=a'$ or $a=a''$,  the senders of challenge rows are the same in both scenarios. Since messages are also  equal  in  $\anoaREL$, the sender-message pairs are fixed ($\EveryButRec$).
Since $\mathcal{A}_2$ is an successful attack on $\anoaREL$ and $\mathcal{A}_1'$ and $\mathcal{A}_1''$ against $\bohliRSAc$ only fix the otherwise randomly-picked $a$, one of the two newly-constructed attacks successfully breaks $\bohliRSAc$. For the case of multiple challenges we refer the reader to the extended version.
\end{proofsketch}

Further, our hierarchy is complete in the sense that no implications are missing: 

\iflong
\todo[inline]{If this is ever used, here are n's that should not be here}
\subsection{For CR Option}
\label{CRProofSketch}
\vspace{-1em}
\begin{theorem}
For all $X$: 
\begin{enumerate}
\item $(c,n,\epsilon, \delta)-X \Rightarrow (c,n,\epsilon,\delta)X_{CR_{\#cr}}$. 
\item $(c,n,\epsilon, \delta)-X \Leftarrow (c,n',\epsilon,\delta)X_{CR_{\#cr}}$ with $n \leq n'$ 
\end{enumerate}
\end{theorem}

\begin{proofsketch}
1.) follows from definition.

2.) Given an attack  $\mathcal{A}_2$, we construct attack $\mathcal{A}_1$. 
Let $\bar{n}$ be the number of previous challenges used in $\mathcal{A}_1$ so far. 
For every batch query $bq$ with $n''$ challenge rows replace the challenge number of the $1$st $,\dots, {\#cr}$'st challenge row with $\bar{n}+1$ in $\mathcal{A}_1$. Continue for the next $\#cr$ challenge rows with the next higher challenge number. Repeat till all challenge rows have a new number.
Use all other queries as $\mathcal{A}_2$ does, give the answers to $\mathcal{A}_2$ and output whatever $\mathcal{A}_2$ outputs.

We show that $\mathcal{A}_1$ is valid against $X_{CR_{\#cr}}$:
The attack construction assures, that at most $\#cr$ challenge row are used in every challenge ($CR_{\#cr}$), all other aspects of $X$ are fullfilled in $\mathcal{A}_2$, too.
Since $\mathcal{A}_1$ perfectly simulates the given attack $\mathcal{A}_2$, it has the same success.
\end{proofsketch}
\fi

%% file: sections/proofs/completeProof.tex
\begin{proof}

\inlineheading{Overview}
We construct the protocol in the following way: Given a protocol $\Pi$ that achieves $X'_1$ ($X_1$ itself or a notion that implies $X_1$), let protocol $\Pi'$ run $\Pi$ and additionally output some information $I$. We argue that learning $I$ does not lead to any advantage distinguishing the scenarios for $X_1$. Hence, $\Pi'$ achieves $X_1$. We give an attack against $X_2$ where learning $I$ allows to distinguish the scenarios. Hence, $\Pi'$ does not achieve $X_2$.  
Further, we use the knowledge that $\implies$ is transitive\footnote{If $X_1 \implies X_2$ and $X_1 \centernot \implies X_3$, it follows that $X_2 \centernot \implies X_3$.} and give the systematic overview over all combinations from sender or impartial notions to all other notions and which proof applies for which relation in Table~\ref{tab:newRel}. Since receiver notions are completely analogous to sender notions, we spare this part of the table.

In Table~\ref{tab:newRel} ``$\Rightarrow$'' indicates that the notion of the column implies the one of the row. ``$=$'' is used when the notions are equal. 
$P_n$ indicates that the counterexample is described in a proof with number $n$. 
$P_A$ and $P_B$ are proofs covering multiple counterexamples and $P_A^*$ is a special one of those counterexamples. 
$P_n'$ means the proof analogous to $P_n$, but for the receiver notion.
$(P_n)$ means that there cannot be an implication, because otherwise it would be a contradiction with proof $P_n$ since our implications are transitive.

\begin{table*}[t!]
  \center
\resizebox{0.95\textwidth}{!}{%
  \begin{tabular}{c c c c c c c c c c c c c c c c c c c}
  &  $\heviaUO$ &$ \bohliSA$ &$\heviaUL$ &$\newCONFWOL$&$\newCONF$ &$\anoaREL$ & $\loopixSRTPU$ &  
 $ \bohliRSUP$ & $\bohliRWUP $ & $\bohliRPS$ & $\bohliRSUU $ & $\bohliRWUU $ & $\bohliRAN $ & $\bohliRWU$ & $\bohliRWA $ &
  $\newAnoaSAWOL$ &  $\newAnoaSA$ 
   \\

  \hline
  
  $\heviaUO$ &$=$&$\Rightarrow$&$\Rightarrow$&$\Rightarrow$&$\Rightarrow$&$\Rightarrow$&$\Rightarrow$&
 $\Rightarrow$ &$\Rightarrow$&$\Rightarrow$&$\Rightarrow$&$\Rightarrow$&$\Rightarrow$&$\Rightarrow$&$\Rightarrow$&
$\Rightarrow$ &$\Rightarrow$ \\

$ \bohliSA$&$P_1$&$=$&$\Rightarrow$&$\Rightarrow$&$\Rightarrow$&$\Rightarrow$&$\Rightarrow$&
$\Rightarrow$&$\Rightarrow$&$\Rightarrow$&$\Rightarrow$&$\Rightarrow$&$\Rightarrow$&$\Rightarrow$&$\Rightarrow$&
$\Rightarrow$ &$\Rightarrow$\\

$\heviaUL$&($P_1$)&($P_2$)&=&$\Rightarrow$&$\Rightarrow$&$P_2$&$\Rightarrow$&
($P_2$)&($P_2$)&($P_2$)&($P_2$)&($P_2$)&($P_2$)&($P_2$)&($P_2$)&
($P_2$)&($P_2$)\\

$\newCONFWOL$&($P_1$)&($P_3$)&($P_3$)&=&$\Rightarrow$&($P_2$)&$P_3$&
($P_3$)&($P_3$)&($P_3$)&($P_3$)&($P_3$)&($P_3$)&($P_3$)&($P_3$)&
($P_3$)&($P_3$)\\

$\newCONF$&($P_1$)&($P_3$)&($P_3$)&($P_{17}$)&=&($P_2$)&$(P_3)$&
($P_3$)&($P_3$)&($P_3$)&($P_3$)&($P_3$)&($P_3$)&($P_3$)&($P_3$)&
($P_3$)&($P_3$)\\

$\anoaREL$&($P_1$)&($P_5$)&($P_5$)&($P_6$)&($P_6$)&=&$P_{16}$&
($P_{16}$)&($P_{16}$)&($P_{16}$)&($P_{16}$)&($P_{16}$)&($P_{16}$)&($P_{16}$)&($P_{16}$)&
($P_{16}$)&($P_{16}$)\\

 $\loopixSRTPU$&($P_1$)&($P_5$)&($P_5$)&($P_6$)&($P_6$)&($P_7$)&=&
($P_{8}$)&($P_{8}$)&($P_{8}$)&($P_{8}$)&($P_{8}$)&($P_{8}$)&($P_{8}$)&($P_{8}$)&
($P_8'$)&($P_8'$)\\

 \hline

$ \bohliRSUP$&($P_1$)&($P_4$)&$\Rightarrow$&$\Rightarrow$&$\Rightarrow$&$\Rightarrow$&$\Rightarrow$&
=&$\Rightarrow$&$\Rightarrow$&$\Rightarrow$&$\Rightarrow$&$\Rightarrow$&$\Rightarrow$&$\Rightarrow$&
$\Rightarrow$&$\Rightarrow$\\

 $\bohliRWUP $&($P_1$)&($P_4$)&$\Rightarrow$&$\Rightarrow$&$\Rightarrow$&$\Rightarrow$&$\Rightarrow$&
($P_A$)&=&$\Rightarrow$&($P_A$)&$\Rightarrow$&$\Rightarrow$&$\Rightarrow$&$\Rightarrow$&
$\Rightarrow$&$\Rightarrow$\\

 $\bohliRPS$&($P_1$)&($P_6$)&($P_6$)&$(P_6)$&$P_6$&$\Rightarrow$&$\Rightarrow$&
($P_A$)&($P_A$)&=&($P_A$)&($P_A$)&$\Rightarrow$&($P_A$)&$\Rightarrow$&
($P_6$)&($P_6$)\\

 $\bohliRSUU $&($P_1$)&($P_4$)&$\Rightarrow$&$\Rightarrow$&$\Rightarrow$&$\Rightarrow$&$\Rightarrow$&
($P_A$)&($P_A$)&($P_A$)&=&$\Rightarrow$&$\Rightarrow$&$\Rightarrow$&$\Rightarrow$&
$\Rightarrow$&$\Rightarrow$\\

 $\bohliRWUU $&($P_1$)&($P_4$)&$\Rightarrow$&$\Rightarrow$&$\Rightarrow$&$\Rightarrow$&$\Rightarrow$&
($P_A$)&($P_A$)&($P_A$)&($P_A$)&=&$\Rightarrow$&$\Rightarrow$&$\Rightarrow$&
$\Rightarrow$&$\Rightarrow$\\

 $\bohliRAN $&($P_1$)&($P_6$)&($P_6$)&($P_6$)&($P_6$)&$\Rightarrow$&$\Rightarrow$&
($P_A$)&($P_A$)&($P_A$)&($P_A$)&($P_A$)&=&($P_A$)&$\Rightarrow$&
($P_6$)&($P_6$)\\

 $\bohliRWU$&($P_1$)&($P_4$)&$\Rightarrow$&$\Rightarrow$&$\Rightarrow$&$\Rightarrow$&$\Rightarrow$&
($P_A$)&($P_A$)&($P_A$)&($P_A$)&($P_A$)&($P_A$)&=&$\Rightarrow$&
$\Rightarrow$&$\Rightarrow$\\

 $\bohliRWA $&($P_1$)&($P_6$)&($P_6$)&($P_6$)&($P_6$)&$\Rightarrow$&$\Rightarrow$&
($P_A$)&($P_A$)&($P_A$)&($P_A$)&($P_A$)&($P_A$)&($P_A$)&=&
($P_6$)&($P_6$)\\

 $\newAnoaSAWOL$&($P_1$)&($P_5$)&$P_5$&$\Rightarrow$&$\Rightarrow$&$\Rightarrow$&$\Rightarrow$&
$(P_8)$&$(P_8)$&$(P_8)$&$(P_8)$&$(P_8)$&$(P_8)$&$(P_8)$&$(P_8)$&
=&$\Rightarrow$\\
 
 $\newAnoaSA$&($P_1$)&($P_5$)&$(P_5)$&$\Rightarrow$&$\Rightarrow$&$\Rightarrow$&$\Rightarrow$&
$(P_8)$&$(P_8)$&$(P_8)$&$(P_8)$&$(P_8)$&$(P_8)$&$(P_8)$&$(P_8)$&
($P_{17}$)&=\\

 $\bohliSSAc$&($P_1$)&($P_5$)&($P_5$)&($P_6$)&($P_6$)&$\Rightarrow$&$\Rightarrow$&
$(P_8)$&$(P_8)$&$(P_8)$&$(P_8)$&$(P_8)$&$(P_8)$&$(P_8)$&$(P_8)$&
($P_6$)&($P_6$)\\

 $\bohliSSUPc $&($P_1$)&($P_5$)&($P_5$)&($P_6$)&($P_6$)&$P_7$&$\Rightarrow$&
$(P_8)$&$(P_8)$&$(P_8)$&$(P_8)$&$(P_8)$&$(P_8)$&$(P_8)$&$(P_8)$&
($P_B$)&($P_B$)\\

 $\bohliSWUPc$&($P_1$)&($P_5$)&($P_5$)&($P_6$)&($P_6$)&($P_7$)&$\Rightarrow$&
$(P_8)$&$(P_8)$&$(P_8)$&$(P_8)$&$(P_8)$&$(P_8)$&$(P_8)$&$(P_8)$&
($P_B$)&($P_B$)\\

 $\bohliSPSc $&($P_1$)&($P_5$)&($P_5$)&($P_6$)&($P_6$)&($P_7$)&$\Rightarrow$&
$(P_8)$&$(P_8)$&$(P_8)$&$(P_8)$&$(P_8)$&$(P_8)$&$(P_8)$&$(P_8)$&
($P_B$)&($P_B$)\\

 $\bohliSSUUc $&($P_1$)&($P_5$)&($P_5$)&($P_6$)&($P_6$)&($P_7$)&$\Rightarrow$&
$(P_8)$&$(P_8)$&$(P_8)$&$(P_8)$&$(P_8)$&$(P_8)$&$(P_8)$&$(P_8)$&
($P_B$)&($P_B$)\\

 $\bohliSWUUc$&($P_1$)&($P_5$)&($P_5$)&($P_6$)&($P_6$)&($P_7$)&$\Rightarrow$&
$(P_8)$&$(P_8)$&$(P_8)$&$(P_8)$&$(P_8)$&$(P_8)$&$(P_8)$&$(P_8)$&
($P_B$)&($P_B$)\\

 $\bohliSANc $&($P_1$)&($P_5$)&($P_5$)&($P_6$)&($P_6$)&($P_7$)&$\Rightarrow$&
$(P_8)$&$(P_8)$&$(P_8)$&$(P_8)$&$(P_8)$&$(P_8)$&$(P_8)$&$(P_8)$&
($P_B$)&($P_B$)\\

 $\bohliSWUc$&($P_1$)&($P_5$)&($P_5$)&($P_6$)&($P_6$)&($P_7$)&$\Rightarrow$&
$(P_8)$&$(P_8)$&$(P_8)$&$(P_8)$&$(P_8)$&$(P_8)$&$(P_8)$&$(P_8)$&
($P_B$)&($P_B$)\\

 $\bohliSWAc$&($P_1$)&($P_5$)&($P_5$)&($P_6$)&($P_6$)&($P_7$)&$\Rightarrow$&
$(P_8)$&$(P_8)$&$(P_8)$&$(P_8)$&$(P_8)$&$(P_8)$&$(P_8)$&$(P_8)$&
($P_B$)&($P_B$)\\

$\anoaUL$&($P_1$)&($P_5$)&($P_5$)&($P_6$)&($P_6$)&$P_{10}$&$P_{11}$&
($P_{11}$)&($P_{11}$)&($P_{11}$)&($P_{11}$)&($P_{11}$)&($P_{11}$)&($P_{11}$)&($P_{11}$)&
($P_{11}$)&($P_{11}$)\\

$\newSAUO $&($P_1$)&($P_5$)&($P_5$)&($P_6$)&($P_6$)&$P_{12}$&$P_{13}$&
($P_{13}$)&($P_{13}$)&($P_{13}$)&($P_{13}$)&($P_{13}$)&($P_{13}$)&($P_{13}$)&($P_{13}$)&
($P_{13}$)&($P_{13}$)\\

$\newSA$&($P_1$)&($P_5$)&($P_5$)&($P_6$)&($P_6$)&$P_{14}$&$P_{15}$&
($P_{15}$)&($P_{15}$)&($P_{15}$)&($P_{15}$)&($P_{15}$)&($P_{15}$)&($P_{15}$)&($P_{15}$)&
($P_{15}$)&($P_{15}$)\\

  \hline
    \\
  \\
  
   \hline

&$\bohliSSAc$& $\bohliSSUPc $ & $\bohliSWUPc$ & $\bohliSPSc $ & $\bohliSSUUc $ & $\bohliSWUUc$ & $\bohliSANc $ & $\bohliSWUc$ & $\bohliSWAc$ & 
$\anoaUL$ &$\newSAUO $ &$\newSA$ & 
  $\bohliSSUP$ & $\bohliSWUP$ & $\bohliSPS$ & $\bohliSSUU$ & $\bohliSWUU$ \\
  \hline
  
  $\heviaUO$&$\Rightarrow$&$\Rightarrow$&$\Rightarrow$&$\Rightarrow$&$\Rightarrow$&$\Rightarrow$&$\Rightarrow$&$\Rightarrow$&$\Rightarrow$&
$\Rightarrow$  &$\Rightarrow$&$\Rightarrow$&
 $\Rightarrow$ &$\Rightarrow$&$\Rightarrow$&$\Rightarrow$&$\Rightarrow$\\
  
$ \bohliSA$&$\Rightarrow$&$\Rightarrow$&$\Rightarrow$&$\Rightarrow$&$\Rightarrow$&$\Rightarrow$&$\Rightarrow$&$\Rightarrow$&$\Rightarrow$&
$\Rightarrow$  &$\Rightarrow$&$\Rightarrow$&
$\Rightarrow$  &$\Rightarrow$&$\Rightarrow$&$\Rightarrow$&$\Rightarrow$\\
  
$\heviaUL$&($P_2$)&($P_A'^*$)&($P_A'^*$)&($P_A'^*$)&($P_A'^*$)&($P_A'^*$)&($P_A'^*$)&$\Rightarrow$&$\Rightarrow$&
  ($P_4'$) &$\Rightarrow$&$\Rightarrow$&
  ($P_2$)&($P_2$)&($P_2$)&($P_2$)&($P_2$)\\
  
$\newCONFWOL$&($P_3$)&($P_3$)&($P_3$)&($P_3$)&($P_3$)&($P_3$)&($P_3$)&($P_3$)&($P_3$)&
  ($P_4'$)&$\Rightarrow$&$\Rightarrow$&
  ($P_3$)&($P_3$)&($P_3$)&($P_3$)&($P_3$)\\  
  
  $\newCONF$&($P_3$)&($P_3$)&($P_3$)&($P_3$)&($P_3$)&($P_3$)&($P_3$)&($P_3$)&($P_3$)&
  ($P_4'$)&$P_{17}'$&$P_{18}'$&
  ($P_3$)&($P_3$)&($P_3$)&($P_3$)&($P_3$)\\  

$\anoaREL$&($P_{16}$)&($P_{16}$)&($P_{16}$)&($P_{16}$)&($P_{16}$)&($P_{16}$)&($P_{16}$)&($P_{16}$)&($P_{16}$)&
   ($P_4'$)&($P_{17}'$)&$P_{18}'$)&
  ($P_{16}$)&($P_{16}$)&($P_{16}$)&($P_{16}$)&($P_{16}$)\\

 $\loopixSRTPU$&($P_8'$)&($P_8'$)&($P_8'$)&($P_8'$)&($P_8'$)&($P_8'$)&($P_8'$)&($P_8'$)&($P_8'$)&
   ($P_4'$)&($P_{19}$)&($P_{20}$)&
  ($P_{8}$)&($P_{8}$)&($P_{8}$)&($P_{8}$)&($P_{8}$)\\

 \hline
 
$ \bohliRSUP$&$\Rightarrow$&$\Rightarrow$&$\Rightarrow$&$\Rightarrow$&$\Rightarrow$&$\Rightarrow$&$\Rightarrow$&$\Rightarrow$&$\Rightarrow$&
  $\Rightarrow$&$\Rightarrow$&$\Rightarrow$&
  ($P_A$)&($P_A$)&($P_A$)&($P_A$)&($P_A$)\\

 $\bohliRWUP $&$\Rightarrow$&$\Rightarrow$&$\Rightarrow$&$\Rightarrow$&$\Rightarrow$&$\Rightarrow$&$\Rightarrow$&$\Rightarrow$&$\Rightarrow$&
$\Rightarrow$  &$\Rightarrow$&$\Rightarrow$&
  ($P_A$)&($P_A$)&($P_A$)&($P_A$)&($P_A$)\\

 $\bohliRPS$&$\Rightarrow$&$\Rightarrow$&$\Rightarrow$&$\Rightarrow$&$\Rightarrow$&$\Rightarrow$&$\Rightarrow$&$\Rightarrow$&$\Rightarrow$&
$\Rightarrow$  &$\Rightarrow$&$\Rightarrow$&
  ($P_A$)&($P_A$)&($P_A$)&($P_A$)&($P_A$)\\
  
 $\bohliRSUU $&$\Rightarrow$&$\Rightarrow$&$\Rightarrow$&$\Rightarrow$&$\Rightarrow$&$\Rightarrow$&$\Rightarrow$&$\Rightarrow$&$\Rightarrow$&
$\Rightarrow$&$\Rightarrow$&$\Rightarrow$&
  ($P_A$)&($P_A$)&($P_A$)&($P_A$)&($P_A$)\\
  
 $\bohliRWUU $&$\Rightarrow$&$\Rightarrow$&$\Rightarrow$&$\Rightarrow$&$\Rightarrow$&$\Rightarrow$&$\Rightarrow$&$\Rightarrow$&$\Rightarrow$&
  $\Rightarrow$&$\Rightarrow$&$\Rightarrow$&
  ($P_A$)&($P_A$)&($P_A$)&($P_A$)&($P_A$)\\
  
 $\bohliRAN $&$\Rightarrow$&$\Rightarrow$&$\Rightarrow$&$\Rightarrow$&$\Rightarrow$&$\Rightarrow$&$\Rightarrow$&$\Rightarrow$&$\Rightarrow$&
  $\Rightarrow$&$\Rightarrow$&$\Rightarrow$&
  ($P_A$)&($P_A$)&($P_A$)&($P_A$)&($P_A$)\\
  
 $\bohliRWU$&$\Rightarrow$&$\Rightarrow$&$\Rightarrow$&$\Rightarrow$&$\Rightarrow$&$\Rightarrow$&$\Rightarrow$&$\Rightarrow$&$\Rightarrow$&
  $\Rightarrow$&$\Rightarrow$&$\Rightarrow$&
  ($P_A$)&($P_A$)&($P_A$)&($P_A$)&($P_A$)\\
  
 $\bohliRWA $&$\Rightarrow$&$\Rightarrow$&$\Rightarrow$&$\Rightarrow$&$\Rightarrow$&$\Rightarrow$&$\Rightarrow$&$\Rightarrow$&$\Rightarrow$&
  $\Rightarrow$&$\Rightarrow$&$\Rightarrow$&
  ($P_A$)&($P_A$)&($P_A$)&($P_A$)&($P_A$)\\
  
  $\newAnoaSAWOL$&$\Rightarrow$&$\Rightarrow$&$\Rightarrow$&$\Rightarrow$&$\Rightarrow$&$\Rightarrow$&$\Rightarrow$&$\Rightarrow$&$\Rightarrow$&
  $\Rightarrow$&$\Rightarrow$&$\Rightarrow$&
  $(P_8)$&$(P_8)$&$(P_8)$&$(P_8)$&$(P_8)$\\
 
 $\newAnoaSA$&$\Rightarrow$&$\Rightarrow$&$\Rightarrow$&$\Rightarrow$&$\Rightarrow$&$\Rightarrow$&$\Rightarrow$&$\Rightarrow$&$\Rightarrow$&
  $\Rightarrow$&$\Rightarrow$&$\Rightarrow$&
  $(P_8)$&$(P_8)$&$(P_8)$&$(P_8)$&$(P_8)$\\

 $\bohliSSAc$&=&$\Rightarrow$&$\Rightarrow$&$\Rightarrow$&$\Rightarrow$&$\Rightarrow$&$\Rightarrow$&$\Rightarrow$&$\Rightarrow$&
  $\Rightarrow$&$\Rightarrow$&$\Rightarrow$&
  $(P_8)$&$(P_8)$&$(P_8)$&$(P_8)$&$(P_8)$\\
  
 $\bohliSSUPc $&$P_B$&=&$\Rightarrow$&$\Rightarrow$&$\Rightarrow$&$\Rightarrow$&$\Rightarrow$&$\Rightarrow$&$\Rightarrow$&
   ($P_4'$)&$P_{19}$&$\Rightarrow$&
  $(P_8)$&$(P_8)$&$(P_8)$&$(P_8)$&$(P_8)$\\
  
 $\bohliSWUPc$&($P_B$)&$P_B$&=&$\Rightarrow$&&$\Rightarrow$&$\Rightarrow$&$\Rightarrow$&$\Rightarrow$&
   ($P_4'$)&($P_{19}$)&$\Rightarrow$&
  $(P_8)$&$(P_8)$&$(P_8)$&$(P_8)$&$(P_8)$\\
  
 $\bohliSPSc $&($P_B$)&$P_B$&$P_B$&=&$P_B$&$P_B$&$\Rightarrow$&$P_B$&$\Rightarrow$&
   ($P_4'$)&($P_{19}$)&$P_{20}$&
  $(P_8)$&$(P_8)$&$(P_8)$&$(P_8)$&$(P_8)$\\
  
 $\bohliSSUUc $&($P_B$)&$P_B$&$P_B$&$P_B$&=&$\Rightarrow$&$\Rightarrow$&$\Rightarrow$&$\Rightarrow$&
  ($P_4'$) &($P_{19}$)&$\Rightarrow$&
  $(P_8)$&$(P_8)$&$(P_8)$&$(P_8)$&$(P_8)$\\
  
 $\bohliSWUUc$&($P_B$)&$P_B$&$P_B$&$P_B$&$P_B$&=&$\Rightarrow$&$\Rightarrow$&$\Rightarrow$&
   ($P_4'$)&($P_{19}$)&$\Rightarrow$&
  $(P_8)$&$(P_8)$&$(P_8)$&$(P_8)$&$(P_8)$\\
  
 $\bohliSANc $&($P_B$)&$P_B$&$P_B$&$P_B$&$P_B$&$P_B$&=&$P_B$&$\Rightarrow$&
  ($P_4'$) &($P_{19}$)&($P_{20}$)&
  $(P_8)$&$(P_8)$&$(P_8)$&$(P_8)$&$(P_8)$\\
  
 $\bohliSWUc$&($P_B$)&$P_B$&$P_B$&$P_B$&$P_B$&$P_B$&$P_B$&=&$\Rightarrow$&
   ($P_4'$)&($P_{19}$)&$\Rightarrow$&
$(P_8)$&$(P_8)$&$(P_8)$&$(P_8)$&$(P_8)$\\
  
 $\bohliSWAc$&($P_B$)&$P_B$&$P_B$&$P_B$&$P_B$&$P_B$&$P_B$&$P_B$&=&
  ($P_4'$) &($P_{19}$)&($P_{20}$)&
  $(P_8)$&$(P_8)$&$(P_8)$&$(P_8)$&$(P_8)$\\

$\anoaUL$&($P_{11}$)&($P_{11}$)&($P_{11}$)&($P_{11}$)&($P_{11}$)&($P_{11}$)&($P_{11}$)&($P_{11}$)&($P_{11}$)&
  =&$P_{21}$&$P_{22}$&
  ($P_{11}$)&($P_{11}$)&($P_{11}$)&($P_{11}$)&($P_{11}$)\\
  
$\newSAUO $&($P_{13}$)&($P_{13}$)&($P_{13}$)&($P_{13}$)&($P_{13}$)&($P_{13}$)&($P_{13}$)&($P_{13}$)&($P_{13}$)&
   ($P_4'$)&=&$P_{23}$&
 ($P_{13}$) &($P_{13}$)&($P_{13}$)&($P_{13}$)&($P_{13}$)\\
  
$\newSA$&($P_{15}$)&($P_{15}$)&($P_{15}$)&($P_{15}$)&($P_{15}$)&($P_{15}$)&($P_{15}$)&($P_{15}$)&($P_{15}$)&
  ($P_4'$)&($P_{19}$)&=&
  ($P_{15}$)&($P_{15}$)&($P_{15}$)&($P_{15}$)&($P_{15}$)\\

  \hline
    \\
  \\
  
   \hline

  &   $\bohliSAN$ & $\bohliSWU $ & $\bohliSWA$ & 
 $\anoaRAWOL $& $\anoaRA $ & 
 $\bohliRSAc$ & $\bohliRSUPc$ & $\bohliRWUPc$ & $\bohliRPSc$ & $\bohliRSUUc$ & $\bohliRWUUc$ & $\bohliRANc $ & $\bohliRWUc$ & $\bohliRWAc$ & 
 
 $\newAnoaUL$ & $\newRAUO $ & $\newRA $ \\ \hline
  
  $\heviaUO$ &$\Rightarrow$&$\Rightarrow$&$\Rightarrow$&$\Rightarrow$&
  $\Rightarrow$&
  $\Rightarrow$&$\Rightarrow$&$\Rightarrow$&$\Rightarrow$&$\Rightarrow$&$\Rightarrow$&$\Rightarrow$&$\Rightarrow$&$\Rightarrow$&
$\Rightarrow$  &$\Rightarrow$&$\Rightarrow$\\
  
$ \bohliSA$ &$\Rightarrow$&$\Rightarrow$&$\Rightarrow$&$\Rightarrow$&
 $\Rightarrow$ &
 $\Rightarrow$ &$\Rightarrow$&$\Rightarrow$&$\Rightarrow$&$\Rightarrow$&$\Rightarrow$&$\Rightarrow$&$\Rightarrow$&$\Rightarrow$&
$\Rightarrow$  &$\Rightarrow$&$\Rightarrow$\\
  
$\heviaUL$ &($P_2$)&($P_2$)&($P_2$)&
 ($P_2$) & ($P_2$) &
  ($P_2$)&($P_A^*$)&($P_A^*$)&($P_A^*$)&($P_A^*$)&($P_A^*$)&($P_A^*$)&$\Rightarrow$&$\Rightarrow$&
  ($P_4$)&$\Rightarrow$&$\Rightarrow$\\
  
$\newCONFWOL$ &($P_3$)&($P_3$)&($P_3$)&
  ($P_3$)& ($P_3$)&
  ($P_3$)&($P_3$)&($P_3$)&($P_3$)&($P_3$)&($P_3$)&($P_3$)&($P_3$)&($P_3$)&
  ($P_4$)&$\Rightarrow$&$\Rightarrow$\\
  
  $\newCONF$ &($P_3$)&($P_3$)&($P_3$)&
  ($P_3$)& ($P_3$)&
  ($P_3$)&($P_3$)&($P_3$)&($P_3$)&($P_3$)&($P_3$)&($P_3$)&($P_3$)&($P_3$)&
  ($P_4$)&($P_{17}$)&($P_{18}$)\\
  
$\anoaREL$ &($P_{16}$)&($P_{16}$)&($P_{16}$)&
  ($P_{16}$)&($P_{16}$)&
  ($P_{16}$)&($P_{16}$)&($P_{16}$)&($P_{16}$)&($P_{16}$)&($P_{16}$)&($P_{16}$)&($P_{16}$)&($P_{16}$)&
  ($P_4$)&($P_{17}$)&($P_{18}$)\\
  
 $\loopixSRTPU$ &($P_{8}$)&($P_{8}$)&($P_{8}$)&
  ($P_{8}$)&($P_{8}$)&
  ($P_{8}$)&($P_{8}$)&($P_{8}$)&($P_{8}$)&($P_{8}$)&($P_{8}$)&($P_{8}$)&($P_{8}$)&($P_{8}$)&
  ($P_4$)&($P_{17}$)&($P_{18}$)\\
  
   \hline

$ \bohliRSUP$ &($P_A$)&($P_A$)&($P_A$)&
  ($P_A$)&($P_A$)&
  $P_A$&$\Rightarrow$&$\Rightarrow$&$\Rightarrow$&$\Rightarrow$&$\Rightarrow$&$\Rightarrow$&$\Rightarrow$&$\Rightarrow$&
  $P_4$&$\Rightarrow$&$\Rightarrow$\\
  
 $\bohliRWUP $ &($P_A$)&($P_A$)&($P_A$)&
  ($P_A$)&($P_A$)&
  $P_A$&$P_A$&$\Rightarrow$&$\Rightarrow$&$P_A$&$\Rightarrow$&$\Rightarrow$&$\Rightarrow$&$\Rightarrow$&
  ($P_4$)&$\Rightarrow$&$\Rightarrow$\\
  
 $\bohliRPS$ &($P_A$)&($P_A$)&($P_A$)&
 ($P_A$) &($P_A$)&
  $P_A$&$P_A$&$P_A$&$\Rightarrow$&$P_A$&$P_A$&$\Rightarrow$&$P_A$&$\Rightarrow$&
  ($P_4$)&$P_{24}$&$P_{9}$\\
  
 $\bohliRSUU $ &($P_A$)&($P_A$)&($P_A$)&
  ($P_A$)&($P_A$)&
 $P_A$ &$P_A$&$P_A$&$P_A$&$\Rightarrow$&$\Rightarrow$&$\Rightarrow$&$\Rightarrow$&$\Rightarrow$&
  ($P_4$)&$\Rightarrow$&$\Rightarrow$\\
  
 $\bohliRWUU $ &($P_A$)&($P_A$)&($P_A$)&
  ($P_A$)&($P_A$)&
  $P_A$&$P_A$&$P_A$&$P_A$&$P_A$&$\Rightarrow$&$\Rightarrow$&$\Rightarrow$&$\Rightarrow$&
  ($P_4$)&$\Rightarrow$&$\Rightarrow$\\
  
 $\bohliRAN $ &($P_A$)&($P_A$)&($P_A$)&
  ($P_A$)&($P_A$)&
  $P_A$&$P_A$&$P_A$&$P_A$&$P_A$&$P_A$&$\Rightarrow$&$P_A$&$\Rightarrow$&
  ($P_4$)&($P_{24}$)&($P_{9}$)\\
  
 $\bohliRWU$ &($P_A$)&($P_A$)&($P_A$)&
 ($P_A$) &($P_A$)&
 $P_A$ &$P_A$&$P_A$&$P_A$&$P_A$&$P_A$&$P_A^*$&$\Rightarrow$&$\Rightarrow$&
  ($P_4$)&$\Rightarrow$&$\Rightarrow$\\
  
 $\bohliRWA $ &($P_A$)&($P_A$)&($P_A$)&
  ($P_A$)&($P_A$)&
 $P_A$ &$P_A$&$P_A$&$P_A$&$P_A$&$P_A$&$P_A$&$P_A$&$\Rightarrow$&
  ($P_4$)&($P_{24}$)&($P_{9}$)\\
  
  $\newAnoaSAWOL$ &$(P_8)$&$(P_8)$&$(P_8)$&
  $(P_8)$&  $(P_8)$&
  $(P_8)$&$(P_8)$&$(P_8)$&$(P_8)$&$(P_8)$&$(P_8)$&$(P_8)$&$(P_8)$&$P_8$&
  ($P_4$)&$\Rightarrow$&$\Rightarrow$\\
 
 $\newAnoaSA$ &$(P_8)$&$(P_8)$&$(P_8)$&
  $(P_8)$&  $(P_8)$&
  $(P_8)$&$(P_8)$&$(P_8)$&$(P_8)$&$(P_8)$&$(P_8)$&$(P_8)$&$(P_8)$&$(P_8)$&
  ($P_4$)&$P_{17}$&$P_{18}$\\

 $\bohliSSAc$ &$(P_8)$&$(P_8)$&$(P_8)$&
  $(P_8)$&  $(P_8)$&
  $(P_8)$&$(P_8)$&$(P_8)$&$(P_8)$&$(P_8)$&$(P_8)$&$(P_8)$&$(P_8)$&$(P_8)$&
  ($P_4$)&($P_{17}$)&($P_{18}$)\\
  
 $\bohliSSUPc $ &$(P_8)$&$(P_8)$&$(P_8)$&
  $(P_8)$&  $(P_8)$&
  $(P_8)$&$(P_8)$&$(P_8)$&$(P_8)$&$(P_8)$&$(P_8)$&$(P_8)$&$(P_8)$&$(P_8)$&
  ($P_4$)&($P_{17}$)&($P_{18}$)\\
  
 $\bohliSWUPc$ &$(P_8)$&$(P_8)$&$(P_8)$&
  $(P_8)$&  $(P_8)$&
 $(P_8)$ &$(P_8)$&$(P_8)$&$(P_8)$&$(P_8)$&$(P_8)$&$(P_8)$&$(P_8)$&$(P_8)$&
  ($P_4$)&($P_{17}$)&($P_{18}$)\\
  
 $\bohliSPSc $ &$(P_8)$&$(P_8)$&$(P_8)$&
  $(P_8)$&  $(P_8)$&
  $(P_8)$&$(P_8)$&$(P_8)$&$(P_8)$&$(P_8)$&$(P_8)$&$(P_8)$&$(P_8)$&$(P_8)$&
  ($P_4$)&($P_{17}$)&($P_{18}$)\\
  
 $\bohliSSUUc $ &$(P_8)$&$(P_8)$&$(P_8)$&
 $(P_8)$ &  $(P_8)$&
 $(P_8)$ &$(P_8)$&$(P_8)$&$(P_8)$&$(P_8)$&$(P_8)$&$(P_8)$&$(P_8)$&$(P_8)$&
  ($P_4$)&($P_{17}$)&($P_{18}$)\\
  
 $\bohliSWUUc$ &$(P_8)$&$(P_8)$&$(P_8)$&
  $(P_8)$&  $(P_8)$&
  $(P_8)$&$(P_8)$&$(P_8)$&$(P_8)$&$(P_8)$&$(P_8)$&$(P_8)$&$(P_8)$&$(P_8)$&
  ($P_4$)&($P_{17}$)&($P_{18}$)\\
  
 $\bohliSANc $ &$(P_8)$&$(P_8)$&$(P_8)$&
  $(P_8)$&  $(P_8)$&
  $(P_8)$&$(P_8)$&$(P_8)$&$(P_8)$&$(P_8)$&$(P_8)$&$(P_8)$&$(P_8)$&$(P_8)$&
  ($P_4$)&($P_{17}$)&($P_{18}$)\\
  
 $\bohliSWUc$ &$(P_8)$&$(P_8)$&$(P_8)$&
  $(P_8)$&  $(P_8)$&
  $(P_8)$&$(P_8)$&$(P_8)$&$(P_8)$&$(P_8)$&$(P_8)$&$(P_8)$&$(P_8)$&$(P_8)$&
  ($P_4$)&($P_{17}$)&($P_{18}$)\\
  
 $\bohliSWAc$ &$(P_8)$&$(P_8)$&$(P_8)$&
  $(P_8)$&  $(P_8)$&
  $(P_8)$&$(P_8)$&$(P_8)$&$(P_8)$&$(P_8)$&$(P_8)$&$(P_8)$&$(P_8)$&$(P_8)$&
  ($P_4$)&($P_{17}$)&($P_{18}$)\\

$\anoaUL$ &($P_{11}$)&($P_{11}$)&($P_{11}$)&
  ($P_{11}$)&($P_{11}$)&
  ($P_{11}$)&($P_{11}$)&($P_{11}$)&($P_{11}$)&($P_{11}$)&($P_{11}$)&($P_{11}$)&($P_{11}$)&($P_{11}$)&
  ($P_4$)&($P_{17}$)&($P_{18}$)\\
  
$\newSAUO $ &($P_{13}$)&($P_{13}$)&($P_{13}$)&
  ($P_{13}$)&  ($P_{13}$)&
  ($P_{13}$)&($P_{13}$)&($P_{13}$)&($P_{13}$)&($P_{13}$)&($P_{13}$)&($P_{13}$)&($P_{13}$)&($P_{13}$)&
  ($P_4$)&($P_{17}$)&($P_{18}$)\\
  
$\newSA$ &($P_{15}$)&($P_{15}$)&($P_{15}$)&
($P_{15}$)  &($P_{15}$)  &
 ($P_{15}$) &($P_{15}$)&($P_{15}$)&($P_{15}$)&($P_{15}$)&($P_{15}$)&($P_{15}$)&($P_{15}$)&($P_{15}$)&
  ($P_4$)&($P_{17}$)&($P_{18}$)\\

\end{tabular}}
 \caption{Completeness; proofs for all relations between the notions}
  \label{tab:newRel}
\end{table*}

\inlineheading{Numbered Proofs} The proofs identified in Table \ref{tab:newRel} follow. 
\begin{lemma}{$P_4$.} \label{lemmaTable1}
$(c,\epsilon,\delta)-\bohliRSUP \centernot \implies(c^*,\epsilon^*,\delta^*) - \newAnoaUL$ for any $\epsilon^* \geq 0, \delta^* < 1,  c^*\geq 2$  and any $\epsilon \geq 0, \delta < 1 , c\geq 1$.
\end{lemma}

\begin{proof} 
Given a protocol $\Pi$, that achieves $(c,\epsilon,\delta)-\bohliRSUP$. Let protocol $\Pi'$ be the protocol, that behaves like $\Pi$ and additionally publishes the current number of receivers $|U'|$ after every batch. Since in $\bohliRSUP$ the number of receivers always needs to  be identical in both scenarios, outputting it will not lead to new information for the adversary. So, $\Pi'$ achieves $(c,\epsilon,\delta)-\bohliRSUP$. 

Fix $\epsilon^* \geq 0, \delta^* < 1, c^*\geq 2$ arbitrarily. Let $u'_0,u'_1$ be valid receivers and $m_0$ a valid message. $(c^*,\epsilon^*,\delta^*)- \newAnoaUL$ of $\Pi'$ can be broken by the following attack: An adversary $\mathcal{A}$ inputs a batch query with two valid challenge rows with differing receivers\footnote{For a simplified representation we only present the parts of the communication that differ in both scenarios and spare the senders of the communications.}: $((u'_0,m_0),\text{SwitchStageQuery},(u'_0, m_0))$ as instance 0 of the first scenario, $((u'_1,m_0),\text{SwitchStageQuery},(u'_1, m_0))$ as instance 1 of the first scenario  and $((u'_0,m_0),\text{SwitchStageQuery},(u'_1, m_0))$ as instance 0 and  $((u'_1,m_0),\text{SwitchStageQuery},(u'_0, m_0))$ as instance 1 of the second scenario. If $\Pi'$ outputs the number of receivers as being 1, both messages have been received by the same user and $\mathcal{A}$ outputs 0. Otherwise the number of receivers is 2 and the messages have been received by different users. It outputs 1 in this case and wins the game with certainty. Hence,  $\text{Pr}[0= \langle \mathcal{A} \bigm| Ch(\Pi', \newAnoaUL,c^*,0)\rangle ] =1$ and $\text{Pr}[0= \langle \mathcal{A} \bigm| Ch(\Pi', \newAnoaUL, c^*,1)\rangle]=0$ and thus 
$\text{Pr}[0= \langle \mathcal{A} \bigm| Ch(\Pi', \newAnoaUL,c^*,0)\rangle ] >e^{\epsilon^*} \cdot \text{Pr}[0= \langle \mathcal{A} \bigm| Ch(\Pi', \newAnoaUL, c^*,1)\rangle]+ \delta^*$ (since $1 > e^{\epsilon^*} \cdot 0+ \delta^*$).\footnote{Notice, that any protocol would achieve $(c^*,\epsilon^*,\delta^*) - \newAnoaUL$ for $c^*= 1$ since no complete challenge is possible there.}
\end{proof}

Tables \ref{Construction Idea}  summarizes the ideas of proofs that work analogously to Lemma \ref{lemmaTable1}.  $I=(S,m)$ means the sender(S)-message(m) pairs are published, i.e. it is leaked who sent which message. 1.(S,m) means that the first sender-message pair is revealed. |m| without brackets means the set of all message lengths is published; |U|  the number of senders. The other abbreviations are used analogously. 
The attack is shortened to the format $\langle$(communications of instance 0 scenario 0),(communications of instance 1 of scenario 0)$\rangle$,$\langle$(communications of instance 0 scenario 1),(communications of instance 1 of scenario 1)$\rangle$ (if both instances of the scenario are equal, we shorten to:$\langle$(communications of instance 0 scenario 0)$\rangle$,$\langle$(communications of instance 0 scenario 1)$\rangle$ ) and all not mentioned elements  are equal in both scenarios. 
$m_0,m_1, m_2, m_3$ are messages with $|m_0|<|m_1|$, $|m_2|=|m_3|$ and $m_0 \neq m_1 \neq m_2\neq m_3$; $u_0,u_1,u_2$ senders and $u'_0,u'_1,u'_2$ receivers.

\begin{table} [thb]
  \center
\resizebox{0.48\textwidth}{!}{
  \begin{tabular}{c p{2.6cm} p{2.6cm} p{1cm} p{5.7cm} }

$P_n$ & $X_1$&$X_2$&$I$ & attack\\ \hline
$P_4$&$\bohliRSUP$&$ \newAnoaUL$& $|U'|$& $\langle((u'_0,m_0), \text{switchStage},(u'_0,m_0)),$ $((u'_1,m_0), \text{switchStage},(u'_1,m_0))\rangle,$ $\langle((u'_0,m_0), \text{switchStage},(u'_1,m_0)),$ $((u'_1,m_0), \text{switchStage},(u'_0,m_0))\rangle$\\
\hline
$P_6$&$\bohliRPS$&$\newCONF$& $m$&$((m_2)),$ $((m_3))$\\ 
$P_{24}$&$\bohliRPS$ &$ \newRAUO$&$|U'|, m$& $\langle((u'_0,m_0), (u'_0,m_2)),$ $((u'_0,m_0),(u'_1,m_3))\rangle,$ $\langle((u'_0,m_0), (u'_0,m_3)),$ $((u'_0,m_0),(u'_1,m_2))\rangle$\\ 
$P_{9}$&$\bohliRPS$&$\newRA$&$P'$&$\langle((u'_0,m_2), (u'_0,m_0),(u'_1,m_1)),$ $((u'_0,m_2), (u'_1,m_1),(u'_0,m_0))\rangle,$ $\langle ((u'_0,m_2),(u'_0,m_1),(u'_1,m_0)),$ $ ((u'_0,m_2),(u'_1,m_0),(u'_0,m_1))\rangle$\\ 
\hline
$P_5$&$\newAnoaSAWOL$&$ \heviaUL $&$1. R$&$((u'_0), (u'_1)),$ $ ((u'_1),(u'_0))$\\ 
$P_8$&$\newAnoaSAWOL$&\mbox{$ \bohliRWAc$}& $1. R$& $((u'_0), (u'_1)),$ $ ((u'_1),(u'_0))$ \\ 
$P_{17}$&$\newAnoaSA$&$ \newRAUO$&$|U'|, |m|$& $\langle((u'_0,m_2), (u'_0,m_0)),$ $ ((u'_0,m_2), (u'_1,m_1))\rangle$ $\langle ((u'_0,m_2),(u'_0,m_1)),$ $ ((u'_0,m_2),(u'_1,m_0))\rangle$\\
$P_{18}$&$\newAnoaSA$&$ \newRA$&$(R, |m|)$& $\langle((u'_0,m_0),(u'_1,m_1)),$ $((u'_1,m_1),(u'_0,m_0))\rangle$ $\langle ((u'_0,m_1),(u'_1,m_0)),$ $((u'_1,m_0),(u'_0,m_1))\rangle$\\
\hline
$P_7$& $\bohliSSUPc$& $\anoaREL$& $|U'|, |U|$&$\langle((u_0,u'_0), (u_0,u'_0)),$ $((u_0,u'_0),(u_1,u'_1))\rangle$ $\langle ((u_0,u'_0),(u_0,u'_1))$ $((u_0,u'_0),(u_1,u'_0))\rangle$ \\ 
$P_{19}$&  $ \bohliSSUPc$&$ \newSAUO$&$|U|, |m|$& $\langle((u_0,m_2), (u_0,m_0)),$ $((u_0,m_2),(u_1,m_1))\rangle$ $ \langle((u_0,m_2),(u_0,m_1)),$ $((u_0,m_2),(u_1,m_0))\rangle$\\
  \hline
$P_{20}$&  $\bohliSPSc$&$\newSA$& P& $\langle((u_0,m_2),(u_0,m_0),(u_1,m_1)),$ $((u_0,m_2),(u_1,m_1),(u_0,m_0))\rangle$, $\langle((u_0,m_2),(u_0,m_1),(u_1,m_0)),$ $((u_0,m_2),(u_1,m_0),(u_0,m_1))\rangle$\\
 \hline
$P_n'$&Receiver \mbox{notions}& analogous& analogous &\\ \hline
$P_1$&$\bohliSA$&$\heviaUO$&\nothing &$((u_0)),$ $(\Diamond)$\\ 
\hline
$P_2$&$\heviaUL$&$\anoaREL$& $Q,Q'$& $\langle((u_0,u'_0)), $ $( (u_1,u'_1))\rangle$ $\langle((u_0,u'_1)), $ $((u_1,u'_0))\rangle$\\
\hline
$P_3$&$\newCONFWOL$&$\loopixSRTPU$&(S,R)&$\langle((u_0,u'_0), (u_1,u'_1)), $ $( (u_1,u'_1),(u_0,u'_0))\rangle$ $\langle((u_0,u'_1), (u_1,u'_0)), $ $((u_1,u'_0), (u_0,u'_1))\rangle$\\
\hline
$P_n'$&$\heviaUL$/ $\newCONF$&Receiver notions& analogous& \\
\end{tabular}}
 \caption{Counter example construction idea with $X'_1=X_1$}
  \label{Construction Idea}
\end{table}

\begin{lemma}{$P_{10}$.} \label{lemmaTable2}
$(c,\epsilon, \delta)-\anoaUL \centernot \implies (c^*,\epsilon^*,\delta^*)-\loopixSRTPU$  for any $\epsilon^* \geq 0, \delta^* < 1, c^*\geq 2$ and  for any $\epsilon \geq 0, \delta < 1, c \geq 2$.
\end{lemma}

\begin{proof} 
Given a protocol $\Pi$ with $(c,\epsilon, \delta)-\heviaUO$. Let $\Pi'$ behave like $\Pi$ and additionally publish the first sender-receiver-pair. 
Since $\Pi$ does not leak any information except how many messages are sent in total, $\Pi'$ does not leak any information except the first sender-receiver-pair and how many messages are sent in total. Hence, $\Pi'$ has $(c,\epsilon, \delta)-\anoaUL$, because who sends the second time is concealed (otherwise $(c,\epsilon, \delta)-\heviaUO$ of $\Pi$ could be broken based on this, which would be a contradiction to our assumption.).

Fix $\epsilon^* \geq 0, \delta^* < 1, c^*\geq 2$ and let $u_0,u_1$ be valid senders and $u'_0,u'_1$ valid receivers.
Then the following attack on  $(c^*,\epsilon^*,\delta^*)-\loopixSRTPU$ is possible:  
The adversary $\mathcal{A}$ creates a batch query containing only challenge rows: 
$((u_0,u'_0),(u_1,u'_1))$ as instance 0 of the first scenario, $ ((u_1,u'_1),(u_0,u'_0))$ as instance 1 of the first scenario and $((u_0,u'_1),(u_1,u'_0))$ as instance 0 of the second scenario, $ ((u_1,u'_0),(u_0,u'_1))$ as instance 1 of the second scenario.
If the published first sender-receiver pair is $ u_1,u'_1$ or $u_0,u'_0$ , $\mathcal{A}$ outputs 0. Otherwise it outputs 1. Obviously, the adversary wins the game with certainty with this strategy. Hence,  $\text{Pr}[0= \langle \mathcal{A} \bigm| Ch(\Pi', \loopixSRTPU,c^*,0)\rangle ] =1$ and $\text{Pr}[0= \langle \mathcal{A} \bigm| Ch(\Pi', \loopixSRTPU, c^*,1)\rangle]=0$ and thus 
$\text{Pr}[0= \langle \mathcal{A} \bigm| Ch(\Pi', \loopixSRTPU,c^*,0)\rangle ] >e^{\epsilon^*} \cdot \text{Pr}[0= \langle \mathcal{A} \bigm| Ch(\Pi', \loopixSRTPU, c^*,1)\rangle]+ \delta^*$ (since $1 > e^{\epsilon^*} \cdot 0+ \delta^*$)
\end{proof}

The proofs in Table \ref{Construction Idea2} are done analogously to Lemma \ref{lemmaTable2}. This time, analogously proved relations are added in angle brackets.

\begin{table} [thb]
  \center
\resizebox{0.48\textwidth}{!}{%
  \begin{tabular}{ c p{2.8cm} p{2.5cm} p{1.5cm} p{4cm} }
$P$ &$X_1$&$X_2$&$I$ & attack\\ \hline
$P_{10}$-$P_{15}$&$\anoaUL$ $ \langle \newSAUO$, $\newSA \rangle$&$\anoaREL$ $\langle \loopixSRTPU \rangle$&1. $(S,R)$&$\langle((u_0,u'_0)),((u_1,u'_1))\rangle$, $\langle((u_0,u'_1)),((u_1,u'_0))\rangle $ \\
$P_{21},P_{22}$&$\anoaUL$ &$\newSAUO$ $\langle \newSA \rangle$  &$ 1. (S,|m|)$&$\langle((u_0,m_0)), ((u_1,m_1))\rangle$, $\langle((u_0,m_1)), ((u_1,m_0))\rangle$\\
 \hline  
$P_{n'}$&Receiver notions& analogous\\ 
\end{tabular}}
 \caption{Counter example construction idea with $X'_1=\heviaUO$}
  \label{Construction Idea2}
\end{table}

\begin{lemma}{$P_B$.}\label{BohliLemma}
For $X_1,X_2$ $\in$ $\{\bohliSSAc, \bohliSSUPc$, $\bohliSWUPc$, $ \bohliSPSc$, $ \bohliSSUUc$, $\bohliSWUUc$, $ \bohliSANc$, $\bohliSWUc$, $\bohliSWAc \}$: If not $X_1 \implies X_2$ in our hierarchy, $(c,\epsilon, \delta)-X_1 \centernot \implies (c^*,\epsilon^*, \delta^*)-X_2$  for any $\epsilon^* \geq 0, \delta^* < 1, c^*\geq 1$ and for any $\epsilon \geq 0, \delta < 1, c \geq 1$.
\end{lemma}

\begin{proof} 
If not $X_1 \implies X_2$ in our hierarchy, then in $X_1$ the adversary is more restricted, i.e. it exists a property $Prop \in\{U,|U|, Q, H, P\}$, which has to be equal in both scenarios for $X_1$, but neither has to be equal nor is implied to be equal for $X_2$ (See Figure \ref{propertiesRestricted} to see which property can be used for which choice of $X_1$ and $X_2$).  

\begin{figure}
\center
\includegraphics[width=0.35\textwidth]{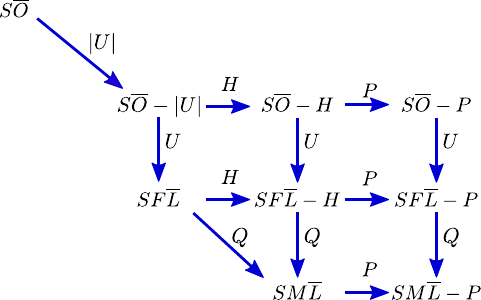}
\caption{The properties at the implication arrows are restricted for the weaker notion (and all notions those notion implies), but not for the stronger notion.} \label{propertiesRestricted}
\end{figure}

We now assume a protocol $\Pi$, that achieves $(c,\epsilon, \delta)-X_1$. Let $\Pi'$ be the protocol that additionally publishes $Prop$. $\Pi'$ still achieves $(c,\epsilon, \delta)-X_1$, since in all attacks on $X_1$ not more information than in $\Pi$ are given to the adversary (The adversary knows $Prop$ already since it is equal for both scenarios). However, since $Prop$ does not have to be equal in $X_2$, the adversary can pick the scenarios such that it can distinguish them in $\Pi'$ based on $Prop$ with certainty. Hence for an arbitrary $\epsilon^* \geq 0, \delta^* < 1, n^*=c^*\geq 1$,  $\text{Pr}[0= \langle \mathcal{A} \bigm| Ch(\Pi', X_2,c^*,0)\rangle ] =1$ and $\text{Pr}[0= \langle \mathcal{A} \bigm| Ch(\Pi', X_2, c^*,1)\rangle]=0$ and thus 
$\text{Pr}[0= \langle \mathcal{A} \bigm| Ch(\Pi', X_2,c^*,0)\rangle ] >e^{\epsilon^*} \cdot \text{Pr}[0= \langle \mathcal{A} \bigm| Ch(\Pi', X_2, c^*,1)\rangle]+ \delta^*$ (since $1 > e^{\epsilon^*} \cdot 0+ \delta^*$).
\end{proof}

\begin{lemma}{$P_A$.}
For $X_2$ $\in$ $\{\bohliRSAc, \bohliRSUPc$, $\bohliRWUPc$, $ \bohliRPSc$, $ \bohliRSUUc$, $\bohliRWUUc$, $ \bohliRANc$, $\bohliRWUc$, $\bohliRWAc \}$, with $X_1\centernot \implies X_2$:  $(c,\epsilon, \delta)-\bohliPrefixS{X_1} \centernot \implies (c^*,\epsilon^*, \delta^*)-X_2$  for any $\epsilon^* \geq 0, \delta^* < 1, c^*\geq 1$ and for any $\epsilon \geq 0, \delta < 1, c \geq 1$.
\end{lemma}

\begin{proof} 
If $X_1 \centernot \implies X_2$, then it exists a property $Prop \in\{U',|U'|, Q', H', P'\}$, which has to be equal in both scenarios for $X_1$, but neither has to equal nor is implied to be equal in $X_2$ (see Proof to Lemma \ref{BohliLemma} for details).

We now assume a protocol $\Pi$, that achieves $((c,\epsilon, \delta)-\bohliPrefixS{X_1}$. Let $\Pi'$ be the protocol that additionally outputs $Prop$. Since the properties of $X_1$ are also checked  in attacks for  $(c,\epsilon, \delta)-\bohliPrefixS{X_1}$ every valid attack on  $(c,\epsilon, \delta)-\bohliPrefixS{X_1}$ will result in the same version of $Prop$ for both scenarios. Hence, $\Pi'$ does not output new information to the adversary (compared to $\Pi$) and still achieves $(c,\epsilon, \delta)-\bohliPrefixS{X_1}$. 

However, since $Prop$ does not have to be equal in $X_2$, the adversary can pick the scenario such that it can distinguish them with certainty in $\Pi'$ based on $Prop$. Hence for an arbitrary $\epsilon^* \geq 0, \delta^* < 1, n^*=c^*\geq 1$,  $\text{Pr}[0= \langle \mathcal{A} \bigm| Ch(\Pi', X_2,c^*,0)\rangle ] =1$ and $\text{Pr}[0= \langle \mathcal{A} \bigm| Ch(\Pi', X_2, c^*,1)\rangle]=0$ and thus 
$\text{Pr}[0= \langle \mathcal{A} \bigm| Ch(\Pi', X_2,c^*,0)\rangle ] >e^{\epsilon^*} \cdot \text{Pr}[0= \langle \mathcal{A} \bigm| Ch(\Pi', X_2, c^*,1)\rangle]+ \delta^*$ (since $1 > e^{\epsilon^*} \cdot 0+ \delta^*$).
\end{proof}

 \begin{lemma}(Proofs $P_{23}$, $P_{23}'$, $P_{16}$) $\newSAUO \centernot \Rightarrow \newSA$  ($\newRAUO \centernot \Rightarrow \newRA, \anoaREL \centernot \Rightarrow \loopixSRTPU$ analogous) 
 \end{lemma}

\begin{proof}

\begin{construction}
Given a protocol $\Pi$ that achieves $\newSAUO$ \ and does not allow to recognize duplicated sender-message pairs, we construct protocol $\Pi'$.
Let $\Pi'$ run $\Pi$ and additionally output a bit for every batch that contains only two communications. For some fixed senders $u_0 \neq u_1$ and messages $m_0\neq m_1$ the protocol will additionally output a bit according to Table~\ref{protocolComplicated}. For any other batch with two communications it will pick the output bit randomly. 
\end{construction}
 \begin{table}[thb]
\center
\resizebox{0.3\textwidth}{!}{
  \begin{tabular}{ c |c|| c| c }

$(u_0,m_0)$&$(u_1,m_0)$&$(u_0,m_0)$&$(u_1,m_0)$\\
$(u_1,m_1)$&$(u_0,m_1)$&$(u_0,m_0)$&$(u_1,m_0)$\\ \hline
$(u_1,m_1)$&$(u_0,m_1)$&$(u_1,m_1)$&$(u_0,m_1)$\\ 
$(u_0,m_0)$&$(u_1,m_0)$& $(u_1,m_1)$&$(u_0,m_1)$\\ \hline
output 0& output 1& output 1& output 0\\
\end{tabular}}
 \caption{Additional output of $\Pi'$}
  \label{protocolComplicated}
\end{table}

Protocol $\Pi'$ outputs the challenge bit $b$ for $\newSA$  and hence does not achieve $\newSA$. However, it achieves $\newSAUO$. The strategy of entering the same challenge row twice does not work  because the advantage gained in the cases of $\newSA$ (i.e.  $a_1\neq a_2$) is annihilated in the cases of duplicated communications (i.e. $a_1=a_2$). Using one equal communication in both scenarios  and a challenge row leads to another compensating distribution: The probability for a correct output is 0.25, for a wrong output 0.25 and for a random output 0.5.  Another attack strategy is not possible since the additional output is only given for batches of size 2.
\end{proof}


\inlineheading{Remark}
If a protocol does not achieve $X_2$ for $c$, then it does not achieve $X_2$ for any $c'>c$ (because every attack with only $c$ CR is also valid if more than $c$ CRs are allowed.)
Further, notice that for some notions a minimum of two CRs is required (e.g. $\anoaUL$)

Additionally, for corruption options: since the proposed attacks do not use any corruption, they are valid for any corruption option. Analogously, since the proposed attacks do not use different sessions in both scenarios, they are valid for any session option.

\end{proof}

%% file: sections/howToUse.tex
\section{How to Use}\label{howToUse}


The framework described above offers the opportunity to thoroughly analyze \acp{ACN}. To perform such an analysis, we advice a top-down approach as follows.

\begin{enumerate}
  \item In case the \ac{ACN} under analysis can be instantiated to protect against different adversaries, fix those parameters.
  \item Extract capabilities of the adversary and general protocol properties from the \ac{ACN} description:
    Specify the allowed \emph{user corruption}. 
    Is it none, static, adaptive? See Table~\ref{tab:corruption options}.
    Are \emph{sessions} (channels) constructed that link messages from the same sender? See Section~\ref{sec:sessions}.
    Extract all other capabilities to include them in the protocol model.
  \item Simplify the \ac{ACN} protocol in a protocol model:
    Generate a simplified protocol (\emph{ideal functionality}) without cryptography by assuming secure communication.
    Show indistinguishability between this ideal functionality and the real-world protocol using a \emph{simulation based proof}.
    Previous work~\cite{backes_provably_2012} can guide the modeling step. See Section \ref{advClass} (UC-realizability) for how the result of the simplified protocol can be transferred to the real-world protocol.
  \item Extract properties based upon the input to the adversary from the ideal functionality:
    Start with \emph{simple properties}, see Table~\ref{tab:information}.
    What does the adversary learn from the protocol execution?
    Continue with \emph{complex properties}.
    See Section~\ref{sec:ComplexProperties}.
  \item After mapping all properties from the protocol and adversary model to our framework, a privacy notions must be selected.
    Either the description of the \ac{ACN} already specifies (in-)formally which privacy goal should be achieved, or the \ac{ACN} under analysis should be shown to achieve a certain notion.
    See Table~\ref{NotionsDefinition} for an overview of our defined notions.
  \item As it is easier to show that a certain notion is not fulfilled compared to show that it is fulfilled, we propose to start with the strongest notions extracted this way. 
    A notion is not fulfilled if the functionality (and thus the protocol) leaks a property to the adversary that he is not allowed to learn for the given notion.
    If it is not obvious that a notion is not fulfilled,  check if the notion can be proven for the  protocol model.
    The related work of Gelernter~\cite{gelernter13limits} and Backes~\cite{backes17anoa} serve as examples for such proofs. 
\end{enumerate}

If the proof cannot be constructed or  $\delta=1$, a weaker notion can be selected for analysis. 
It might also help to consider the case of a limited number of challenge rows (see Section~\ref{sec:challenges}) and limit the adversary by using a single-challenge reducable adversary class (see Section~\ref{advClass} Adversary Class).
In case the proof goes through and yields $\epsilon = 0$ and a negligible $\delta$, the protocol was shown to achieve the selected notion as per Definition~\ref{def:achieve}. If $\epsilon > 0$ or a non negligible $\delta < 1$, the protocol achieves the selected notion as per Definition~\ref{def:achieveEpsilon}. 

If the protocol supports different adversaries, the steps described above can be repeated. This typically leads to adjusting the ideal functionality or adding different adversary classes (see Section~\ref{advClass}) and thus fulfilling different properties of our framework. Analysis results under a variation of \ac{ACN} parameters may achieve different notions in our hierarchy (Figure~\ref{fig:hierarchyColored} and Figure~\ref{HierarchyExtended}). Based on our established relations between notions, analysis results can be compared for various parameters or parameter ranges, as well as against results of other \acp{ACN}.

%% file: sections/discussion.tex
\vspace{-0.3cm}
\section{Discussion} \label{discussion}
\vspace{-0.2cm}
In this section, we present the lessons learned while creating our framework.

\inlineheading{Learning about privacy goals}
The need for formal definitions is emphasized by  the mapping of Loopix's privacy goals to notions as example that  less formal alternatives leave room for interpretation. Further, a result like our hierarchy would be much harder to achieve without formal definitions.

These definitions allow us to point out the relation of privacy and confidentiality ($\newCONF$). The way we ordered the notions in the hierarchy allows easy identification  the notions implying $\newCONF$ (the middle of the upper part). Note that any privacy notion implying $\newCONF$ can be broken by distinguishing the message's content. Further, nearly all those notions also imply $\newCONFWOL$ and  hence, all such notions can be broken by learning the message length.

Our formal definitions also enabled the comparison of existing frameworks. Excluding differences in the adversarial model, quantifications and restrictions that do not apply to all \acp{ACN}, we observe that equivalent definitions are often defined independently by the authors of the analytical frameworks. For this reason, we included the notions of the other frameworks in our hierarchy in Figure \ref{fig:hierarchyold} of Appendix \ref{hierarchyFrameworks}.
$\bohliSA$, $\bohliSPSc$, $\bohliSWAc$, $\bohliSWU$ and $\bohliSWUc$ are defined (under different names) in multiple works;  $\bohliSSAc$ is even defined in all works.

Although previous work includes equivalent definitions, we realized that some notions are still missing. For example, we added weak notions like $\newSA$, $\newRA$ and $\loopixSRTPU$ because they match our understanding of anonymity. Our understanding was confirmed by the analysis of Loopix' goals. Further, we defined all analogous notions for  all communication parties involved (senders and receivers) as real-world application define which party is more vulnerable. For the concrete applications we refer the reader to Section \ref{sec:notionExamples}.

Consequently, we present a broad selection of privacy notions. We are aware that understanding them all in detail might be a challenging task, so we want to provide some intuitions and preferences, based on what we know and conjecture. We expect the lower part of the hierarchy to be more important for $\acp{ACN}$ as \cite{gelernter13limits} already includes an inefficiency result for $\bohliSSAc$ and thus for all notions implying $\bohliSSAc$ \iflong(see Figure \ref{fig:hierarchyInefficient} of the Appendix)\fi. As a first guess, we think $\bohliSSAc$, if higher overhead is manageable, $\bohliSSUUc$, $\bohliSWUc$, $\newSA$ (and receiver counterparts),  $\newCONF$ and $\loopixSRTPU$ are the most popular notions for \acp{ACN}.
Further,  we want to add some results concerning two well-known systems to ease intuition. \cite{backes17anoa}'s analysis of Tor results in a  small, but non-negligible probability  to break  $\bohliSSAc$ and thus Tor does not achieve $\bohliSSAc$ with our strict definition. Classical DC-Nets, on the other hand, do achieve at least $\bohliSPSc$ \cite{gelernter13limits}.  
\iflong
We present our selection of notions also graphically in Figure \ref{fig:hierarchyAnnotated} of the Appendix.
\fi

\inlineheading{Correcting Inconsistencies}
While the above similarities most likely stem from the influence of prior informal work on privacy goals, attempts to provide concrete mappings have led to contradictions.
The AnoA framework maps its notions to their interpretation of Pfitzmann and Hansen's terminology.
Pfitzmann and Hansen match their terminology to the notions of Hevia's framework.
This means that, notions of AnoA and Hevia's framework are indirectly mapped.
However, those notions are not equivalent.
While AnoA's sender anonymity and Hevia's sender unlinkability are both mapped to Pfitzmann and Hansen's sender anonymity, they differ: In Hevia's sender unlinkability the number of times every sender sends can leak to the adversary, but in AnoA's sender anonymity  it cannot.

We belive that  AnoA's sender anonymity  should be called sender unobservability, which is also our name for the corresponding notion.
This follows the naming proposal of Pfitzmann and Hansen and their mapping to Hevia.
It is also more suitable because AnoA's sender anonymity can be broken by learning whether a certain sender is active, i.e. sends a real message, in the system ($u \in U_b$).
In order to achieve this notion, all senders have to be unobservable.
To verify this, we looked at how the notions of AnoA have been used.
For example in \cite{chaum16cmix} the protocol model contains an environment that lets all senders send randomly.
Hence, $U_b$ is hidden by the use of this environment.
We consider that the information that is allowed to be disclosed should instead be part of the notion and not modified by an environment.
Only then are the notions able to represent what information is protected by the protocol.


Another lesson learned by comparing privacy notions is the power of names, because they introduce intuitions.
The fact that Hevia's strong sender anonymity is equivalent to Bohli's receiver weak unlinkability seems counter-intuitive, since a sender notion is translated to a receiver notion. This might also be the reason for the incorrect mapping in \cite{bohli11relations}. 
However, Bohli's receiver weak unlinkability is named this way because receivers are the ``interesting'' users, whose communication is restricted.
It does not restrict senders in any way and hence should be, in most cases, easier to break according to some information about the sender. This is why we and Hevia have classified it as a sender notion.
 An analogous argument explains why Bohli's receiver weak anonymity $R/WA$ implies the restricted case of Bohli's sender strong anonymity $S/SA^\circ$.
 
\iflong
\inlineheading{Use and Limits}
Because there is no restriction in the use of protocol queries, the only restriction to what can be analyzed is what is modeled in the protocol model.
So, if the protocol model includes e.g. insider corruption and active behavior of the insider, like delaying or modifying of messages, those functionality can be used via protocol queries. 
The same applies to timing; if the protocol model specifies that it expects protocol queries telling it, that $x$ seconds passed and the adversary gets meaningful answers after this protocol messages and only an empty answer after batch queries (because they are only processed after some time passed),
attacks on this can be analyzed. 
However, it needs to be specified in the protocol model. 
This model also defines the exact meaning of a batch query, whether messages of one batch are sent at the same time or in a sequence without interruption and specifies whether a synchronous or asynchronous communication model is used.  

Defining the protocol model with the strongest adversary imaginable and restricting it later on with adversary classes is a way to limit the work, when analyzing against different adversary models.
 We decided not to increase the complexity of the framework further by adding interfaces for dimensions of adversary models to the protocol model, i.e. adding more dedicated query types instead of the versatile protocol query. So far, our decisions for query types are driven by the related work. Differentiating between the different possible use cases of protocol queries and defining a set of  adversary classes defining typical adversaries based on this is future work. One of them should allow to limit the amount of corrupted users compared to the amount or honest users or specify n-1 attacks.
Although we presented all notions that we deemed important, there might still be use cases that are not covered. With our properties as building blocks, we conjecture that it is easy to add the needed properties and use them in combination with ours. Further, for adding new notions to the hierarchy, our proofs can be used as templates.
\else
 \inlineheading{Long Version  \cite{longVersion}} Besides giving more technical details, we focus on making our results easier to apply for practitioners 
in the long version of thus paper by  presenting an analysis framework, along with a how-to-use section. 
The extended version includes different parts of the adversary model, like user corruption and limiting the number of adversarial users,  and discusses how typical attacks, like n-1, intersection and active attacks (e.g. delaying or dropping messages), apply to our framework. 
To further simplify the proofs practitioners have to make, it allows privacy goals to be quantified by using multiple challenges or multiple challenge rows and includes results on how the limited case of challenge rows generalizes to more, such that only the limited case needs to be proven. Further, as we are aware that our strict definition of achieving a notion might not work for some practical cases, we point out the relaxed definition that allows for a non-negligible distinguishing probability.
However, none of those extensions limits or contradicts the results regarding the hierarchy of privacy notions built from observable properties that we presented here, as they work independently.  

\fi

\iflong
Another possible extension is including and extending more results of the existing frameworks.
Such results are Bohli's closed hierarchy or Gelernter's inefficiency result.
Further, in Hevia's framework techniques to achieve a stronger \ac{ACN} are included.
For instance, given an  \ac{ACN} achieving $\bohliSWUc$ adding a certain cover traffic creates an \ac{ACN} achieving $\bohliSSAc$.
Those techniques hide certain information that is allowed to leak in the weaker but not in the stronger notion.
The proof of our Theorem \ref{the:noImpl} already includes some information that is allowed to leak in the weaker but not in the stronger notion.
Hence, it is a good starting point for finding more such techniques that help understanding and constructing \acp{ACN} better.
We make the conjecture that adapting the proofs of all these results for our framework is possible.
 However, at the moment we leave proving of these results as future work.
 \fi

%% file: sections/conclusion.tex
\section{Conclusion and Future Work}
\label{conclusion}

\iflong
We have presented a framework of privacy notions for sharper analysis of \acp{ACN} that, to the best of our knowledge, includes  more  notions and assumptions than all existing frameworks based on indistinguishability games. To achieve this, we
\else
We
\fi
  expressed privacy goals formally as privacy notions. We first presented their basic building blocks: properties.
Those properties cover the observable information of communications, which is either required to remain private or allowed to be learned by an adversary, depending on the goal.
\iflong
Furthermore, we checked the sanity of the notions by finding exemplary use cases and by providing a mapping of the privacy goals of a current ACN to them.
\else
We formally specified privacy goals from \acp{ACN} and sorted them into a proven hierarchy, according to their strength.
\fi
\iflong
Our framework allows to compare and understand the differences in privacy goals. We proved the relations between our notions. This means that, for every pair of notions, we know which one is stronger than the other or if they are separate. This way, we resolved inconsistencies between the existing frameworks and built the basis to understand the strengths and weaknesses of \acp{ACN} better, which helps building improved \acp{ACN}.
Further, it creates a unified basis for the analysis  and comparison of \acp{ACN}.
\else
This means, for every pair of notions, we know which one is the stronger; or if they do not imply each other.
As a result, we resolved inconsistencies between existing analytical frameworks and built the foundations to understand the strengths and weaknesses of \acp{ACN} better, which helps analyzing and building improved \acp{ACN}.
\fi

\inlineheading{Future Work}
\iflong
Although our framework allows to analyze all types of attacks with the versatile protocol queries, the protocol model must support those attacks without systematic guidance by interfaces of the framework.
Restrictions of such attacks thus cannot be expressed formally as part of the notion and hence are not easily represented.
In future work we want to introduce more dedicated queries to also formalize other attack dimensions and based on this adversary classes for typical attackers.
\fi
As we mentioned in the discussion, providing more intuitions and understanding the significance of notions is necessary. Therefore, analogous to the analysis of Loopix's privacy goals, more current ACNs can be analyzed to understand which parts of the hierarchy they cover. This can also identify gaps in research; privacy goals for which \acp{ACN} are currently missing. Further, a survey of goals in greater depth would be useful to identify the most important notions in the hierarchy and to provide intuitions and thus ease deciding on the correct notions for practitioners.

Additionally, such a survey helps to understand the relationships between currently-employed privacy enhancing technologies.
Finally, this understanding and the knowledge about how notions are related and differ can be used to define general techniques that strengthen \acp{ACN}.

Beyond that, an investigation of the applicability of our \iflong framework \else notions and hierarchy \fi  to other areas, like e.g. anonymous payment channels, would be interesting.

\section*{Acknowledgements}

We would like to thank the anonymous reviewers for their helpful comments and feedback. Our work was partially funded by the  German  Research Foundation (DFG) within the Research Training Group GRK 1907, the German Federal Ministry of Education and Research (BMBF) within the EXPLOIDS project grant no. 16KIS0523 and European Union's Horizon 2020 project SAINT grant no. 740829.


%% file: sections/appendix/challenger.tex
\vspace{-0.5cm}
\iflong 
\fi
\section{Challenger}
\label{sec:challenger}
\iflong
\vspace{-0.5cm}
This section describes the queries to the challenger $Ch(\Pi, X,c, n,b)$.
Pseudocode of our challenger is shown in Algorithm \ref{Challenger} of Appendix \ref{challengerAppendix}.
\else
This section describes the queries to the challenger $Ch(\Pi, X,b)$.
\fi

\inlineheading{Batch Query}
The batches $\underline{r}_0,\underline{r}_1$ that the adversary chooses for the two scenarios are represented  
 in batch queries.
 \iflong
When the challenger receives a batch query, it will first check if their challenge number $\Psi$ is valid, i.e. $\Psi \in \{1, \dots, n\}$.
Further, the challenger will validate the communications that would be input to $\Pi$ for $b=0$ and $b=1$  as explained below.
\else
When the challenger receives a batch query, it will validate the communications that would be input to $\Pi$ for $b=0$ and $b=1$  as explained below.
\fi
\iflong If the game is not aborted so far, the challenger \iflong will retrieve or create the current state $s$ of the challenge $\Psi$, which stores \else uses stored \fi information to calculate the \iflong aspects. \else properties. \fi \fi
\iflong
Afterwards it checks if the allowed total number of challenge rows $c$ is met.
\fi
If all criteria are met so far, it checks that the \iflong aspects \else properties \fi of the privacy notion $X$ are met by using \iflong the current state of the challenge $s$ \else stored information about the past batches and \fi \iflong,  the set of corrupted users $\hat{U}$,\fi the instances for both scenarios $\underline{r_0}^a, \underline{r^a}_1, a \in \{0,1\}$. Finally, it runs the instance belonging to the challenge bit $b$ of this game and the for this challenge randomly chosen instance bit $a$, if the \iflong aspects \else properties \fi are matched.
 Otherwise, it returns $\perp$ and aborts the experiment.
Running the scenario in the \ac{ACN} protocol will return information that is  forwarded to the adversary\iflong (or adversary class)\fi.
This information is  what an adversary is assumed to be able to observe.

\iflong
\inlineheading{Corrupt Query}
Corrupt queries represent adaptive, momentary corruption of users (senders or receivers).
If the corrupt query is valid, the challenger forwards it to the \ac{ACN} protocol.
The \ac{ACN} protocol returns the current state of the user to the challenger, who forwards it to the adversary.
Active attacks based on corruption are realized with protocol queries if the protocol model allows for them.
\fi

\inlineheading{Protocol Query}
Protocol queries allow the adversary e.g. to compromise parts of the network\iflong(not the users)\fi, set parameters of the \ac{ACN} protocol or use other functionalities modeled in the protocol model, like e.g. active attacks.
The meaning and validity of those queries is specific to the analyzed \ac{ACN} protocol.

 \inlineheading{Switch Stage Query}
If this query occurs and it is allowed, i.e. the notion contains a relevant property, the stage is changed from 1 to 2. 

\inlineheading{Validate Communications}
If the analyzed \ac{ACN} protocol specifies restrictions of senders and receiver-message pairs, their validity is checked by this function. 

\iflong
\inlineheading{Run Protocol}
\iflong
Run protocol first creates a  new random session identifier if there is not already one for this session identifier of the adversary chosen $\text{session}$ with the extension $ID$.
This is done to ensure that the \ac{ACN} protocol is not broken only because the session identifier is leaked.
Afterwards it passes the communications to the \ac{ACN} protocol formalization.
\else
Run protocol passes the chosen communications to the \ac{ACN} protocol formalization.
\fi
\fi

\inlineheading{Remark to simple properties and instances}
 In case the notion only uses simple properties, the challenger will pick $a=0$ and check the properties for
 $\underline{r_1}_j = \underline{r_1^0}_j $ and  $\underline{r_0}_j = \underline{r_0^0}_j$.
 In case the notion uses a combination of simple and complex properties, the challenger will check the simple properties for any pair $\underline{r_1}_j = \underline{r_1^a}_j $ and  $\underline{r_0}_j = \underline{r_0^{a'}}_j$ resulting by any $a, a' \in \{0,1\}$.


%% file: sections/appendix/pseudocode.tex
\FloatBarrier

\label{challengerAppendix}

\begin{algorithm}[h!]
\scriptsize
$\hat{U}=\emptyset$\\
$\text{stage}=1$\\
 \uponMes{(Batch, $\underline{r}^0_0, \underline{r}^1_0, \underline{r}^0_1,\underline{r}^1_1, \Psi$)}{
 	\If {$\Psi \notin \{1,\dots, n\}$}
   		{output $\perp$}  
   	 \uIf{$\Psi \in T$}
      		{Retrieve $s:= s_\Psi$}
      	\Else{
      	$s:= initializeState$ \\
      	\uIf{X uses only simple properties}
      		{$a\gets 0$}
      	\Else{
      	$a\gets ^R \{0,1\}$}
      	add $\Psi$ to $T$}
      	
   	\If{$\lnot Validate(r)$}
   			{output $\perp$}       	
      	Compute $c_t=c_t+|CR(\underline{r}^0_0, \underline{r}^1_0, \underline{r}^0_1,\underline{r}^1_1)|$\\
      	\If{$c_t>c$}{output $\perp$}
 	$s'= calculateNewState( stage, s,\underline{r}^0_0,\underline{r}^1_0, \underline{r}^0_1,\underline{r}^1_1))$\\
 	\uIf {$checkFor(X, \Psi, s',\hat{U}, \underline{r}_0= (\underline{r}^0_0,\underline{r}^1_0), \underline{r}_1= (\underline{r}^0_1,\underline{r}^1_1)) $}
 		{$(\underline{r}, s_\Psi) \gets (\underline{r}^a_b,s'$)}
 	\Else{ output $\perp$}
      	Store $s_\Psi$\\
      	RunProtocol ($\underline{r}$) }
      	
\uponMes{(Protocol, x)}{ 
	\If{$x$ allowed}{Send $x$ to $\Pi$ }
	 }
	
\uponMes{(Corrupt, u)}{ 
	\If{$X$= $\noCorr{X'}$ ($X' \in$ Privacy notions)}{
	output $\perp$}
	\If{$X$= $\static{X'}$ ($X' \in$ Privacy notions) and a batch query occurred before and $u\notin \hat{U}$}
		{output $\perp$}
	$\hat{U}$= $\hat{U}\cup \{u\}$ \\
	Send internal state of $u$ to $\mathcal{A}$
}


\uponMes{(SwitchStage)}{
	\If {$\lnot X$ includes $T_S$ or $T_R$} 
	{output $\perp$}
	$\text{stage}=2$
	}
	
\Validate{($\underline{r}^0_0=(S^0_{0_i},R^0_{0_i},m^0_{0_i},aux^0_{0_i})_{i \in \{1,\dots, l\}}, \underline{r}^1_0, \underline{r}^0_1,\underline{r}^1_1$)}{
\For{$r =(S,R,m,aux) \in \{\underline{r}^{a'}_{b'} \mid a',b' \in \{0,1\}\}$}
	{\If {$\lnot Validate(S,R,m)$} 
	{output FALSE}
	}
	output TRUE
	}

\RunProtocol{ ($\underline{r}=(S_i,R_i,m_i,aux_i)_{i \in \{1,\dots, l\}}$)}{
      	\For {$r_i \in \underline{r}$}{
      		\If{$aux_i=(session_i, ID_i)$}{
      			\uIf{$\not \exists y: (\text{session},y, ID_i)\in S_i$}{ 
      				$y' \gets \{0,1\}^k$\\
      				Store ($\text{session},y',ID$) in $S_i$
      				}
			\Else{
				$y':=y$ from $(\text{session},y, ID_i)\in S_i$ \\ 
				}
			Run $\Pi$ on $r_i$ with session ID $y'$\\ Forward responses sent by $\Pi$ to $\mathcal{A}$
      			}
      		
		\Else
		{Run $\Pi$ on $r_i$\\ Forward responses sent by $\Pi$ to $\mathcal{A}$}
      	}}
      	
\caption{Challenger $Ch(\Pi, X, c, n, b)$}
\label{Challenger}
\end{algorithm}

\FloatBarrier

\section{Notions in Pseudocode}

\label{pseudocode}
CalcNewState always calculates the states for all user roles (senders and receivers). This is for improved readability. It would be sufficient to calculate the parts of the state needed for the current notion.
\begin{algorithm} 	[h!]
\footnotesize
\initializeState{}{
 	$s=(1,1,(\tilde{s},\tilde{s},\tilde{s},\tilde{s},\tilde{r},\tilde{r},\tilde{r},\tilde{r}), 0, \emptyset, \emptyset)$\\
 	return $s$}

\calcNewState{$(newStage,  s=(stage, session, users,  cr, $ $s_{sender}=(L^0_{0_i},L^1_{0_i}, L^0_{1_i},L^1_{1_i})_{i \in \{1, \dots, k-1\}}$, $s_{rec}=(L'^0_{0_i},L'^1_{0_i}, L'^0_{1_i},L'^1_{1_i})_{i \in \{1, \dots, k-1\}}),$ $\underline{r}^0_0=(S^0_{0_i}, R^0_{0_i}, m^0_{0_i}, aux^0_{0_i})_{i \in \{1,..,l\}},$ $\underline{r}^1_0=(S^1_{0_i}, R^1_{0_i}, m^1_{0_i}, aux^1_{0_i})_{i \in \{1,..,l\}},$ $ \underline{r}^0_1=(S^0_{1_i}, R^0_{1_i}, m^0_{1_i}, aux^0_{1_i})_{i \in \{1,..,l\}})$, $  \underline{r}^1_1=(S^1_{1_i}, R^1_{1_i}, m^1_{1_i}, aux^1_{1_i})_{i \in \{1,..,l\}})$}{
		\For {$a \in \{0,1\}$}{
			\For {$b \in \{0,1\}$}{
		$L^a_{b_k}=\{(u,M)\bigm|$ $M = \cup_{j:S^a_{b_j}=u} m^a_{b_i}\}$\\
		$L'^a_{b_k}=\{(u,M)\bigm|$ $ M = \cup_{j:R^a_{b_j}=u} m^a_{b_i}\}$\\}}
      			{$cr=cr+|CR(\underline{r}^0_0,\underline{r}^1_0,\underline{r}^0_1,\underline{r}^1_1)|$}\\
		\If{users=$(\tilde{s},\tilde{s},\tilde{s},\tilde{s},\tilde{r},\tilde{r},\tilde{r},\tilde{r})\land cr>0$}{
			 $((S^0_0,R^0_0,\_,\_),$ $(S^1_0,R^1_0,\_,\_),(S^0_1,R^0_,\_,\_),$ $(S^1_1,R^1_1,\_,\_), \dots)= CR(\underline{r}^0_0,\underline{r}^1_0, \underline{r}^0_1,\underline{r}^1_1)$\\
			users= ($S^0_0,S^1_0,S^0_1,S^1_1,R^0_0,R^1_0,R^0_1,R^1_1$)
		}
		\If{users= ($S^0_0,S^1_0,S^0_1,S^1_1,R^0_0,R^1_0,R^0_1,R^1_1)$ $\land \exists (r^0_0,r^1_0,r^0_1, r^1_1)\in CR(\underline{r}^0_0,\underline{r}^1_0, \underline{r}^0_1,\underline{r}^1_1): (r^0_0,r^1_0,r^0_1, r^1_1)\neq ((S^0_0,R^0_0,\_,\_,\_),(S^1_0,R^1_0,\_,\_,\_),$ $(S^0_1,R^0_1,\_,\_,\_),(S^1_1,R^1_1,\_,\_,\_))$}
			{session=$\perp$}
	$\text{stage}=\text{newStage}$\\
	output $s$
}

\caption{State Management}
\label{calcNewState}
\end{algorithm}

For the simple properties $checkFor$ uses $s^0_{b_k}$ from $s_{sender}$ resp.  $s'^0_{b_k}$ from $s_{rec}$ to calculate $U_b,Q_b,P_b$ and $H_b$ and compares them like in Definition \ref{def:properties}. For the complex properties the senders and receivers of the first challenge row are stored in the $users$-part  and the current stage in the $stage$-part of s. With this complex properties are computed as stated in Definition \ref{def:complexProperties}. Further, for the sessions-aspect the $session$-part of the state is set to $\perp$ if another sender-receiver-pair is used. With this and the $users-$ and $stage$-information the Definition \ref{def:sessions} can be checked. For the corruption it gets all the required information direct as input and can check it like defined in Definition \ref{def:corruption}.  For the challenge complexity the number of challenge rows of this challenge is counted in the $cr$-part of the state and hence, Definition \ref{def:challengeComplexity} can be calculated.

%% file: sections/appendix/epsilonAchieving.tex
\vspace{-0.5cm}
\section{Achieving $(\epsilon, \delta)$-X}
\label{app:epsilonDef}
For some use cases, e.g. if the court of your jurisdiction requires that the sender of a critical content can be identified with a minimal probability of a certain threshold e.g. 70\%, a non-negligible $\delta$ is suitable. Hence, we allow to specify the parameter of $\delta$ and include the well-known concept of differential privacy \cite{dwork14algorithmic} as AnoA does in the following Definition:


\begin{definition}[Achieving $(\epsilon, \delta)-X$]\label{def:achieveEpsilon} 
An \ac{ACN} protocol $\Pi$ is  $(\epsilon, \delta)$ -$X$ with $\epsilon \geq 0$ and $0 \leq \delta \leq 1$, iff for all PPT algorithms $\mathcal{A}$:
\begin{align*}
 \text{Pr}[0= \langle \mathcal{A} \bigm| Ch(\Pi, X,0)\rangle ] &\leq \\
e^{\epsilon} \text{Pr}[0= \langle \mathcal{A} \bigm| Ch(\Pi, X,1)\rangle]&+ \delta\text{.}
\end{align*}
\end{definition}
Note that $\epsilon$ describes how close the probabilities of guessing right and wrong have to be. This can be interpreted as the quality of privacy for this notion. While $\delta$ describes the probability with which the $\epsilon$-quality can be violated. Hence, every \ac{ACN} protocol will achieve $(0,1)-X$ for any notion $X$, but this result does not guarantee anything, since with probability $\delta =1$ the $\epsilon$-quality is not met. 

Note  $\Pi$ is $(0,\delta)-X$ for a negligible $\delta$ is equivalent to the first definition of $\Pi$ achieves $X$.

%% file: sections/appendix/moreExamples.tex
\section{Remaining Examples}
\label{app:examples}

\inlineheading{Impartial Notions: Both-Side Message Unlinkability}
Notions of this group are broken if the sender-message or receiver-message relation is revealed.

\example{The activists know  that their sending and receiving frequencies are similar to regime supporters' and that using an \ac{ACN} is in general not forbidden, but nothing else. Even if the content and length  of the message ($\newCONFWOL$) and the sender-receiver relationship ($\loopixSRTPU$)  is hidden, the regime might be able to distinguish uncritical from critical communications, e.g. whether two activists communicate  ``Today'' or  innocent users an innocent message.  In this case, the regime might learn that currently many critical communications take place and improves its measures against the activists.}

In this case, the activists want a protocol that hides the communications, i.e. relations of sender, message and receiver. However, as using the protocol is not forbidden and their sending frequencies are ordinary, the adversary can learn which users are active senders or receivers and how often they sent and receive. Modeling this, the users need to have the same sending and receiving frequencies in both scenarios $Q,Q'$, since it can be learned. However, everything else needs to be protected and hence, can be chosen by the adversary. This corresponds to the notion  \emph{\heviaULLong \ ($\heviaUL$)}.

\inlineheading{Sender Privacy Notions: Receiver Observability}
In notions of this group the receiver of each communication can be learned.
Hence, such notions include the property that the scenarios are equal except for the senders and messages ($\EveryButSenderMsg$) to ensure that they are equal in both scenarios.

\example{Consider not only sending real messages is persecuted, but also the message content or any combination of senders and message contents is exploited by the regime. If the regime e.g. can distinguish activist Alice sending ``today'' from regime supporter Charlie sending ``see u'', it might have learned an information the activists would rather keep from the regime.  Further, either (1) the activists know that many messages of a certain length are sent or (2) they are not sure that  many messages of a certain length are sent.}

In case (1), Alice needs a \ac{ACN}, that hides the sender activity, the message content and their combination. However, the adversary can especially learn the message length. Modeling this, beyond the above described $\EveryButSenderMsg$, the message lengths have to be equal $|M|$. This results in the notion  \emph{\newAnoaSALong}  ($\newAnoaSA$).
Note that in $\newAnoaSA$ the properties of $\newCONF$ are included and further the senders are allowed to differ in the two scenarios.
The second case (2) requires a protocol that additionally hides the message length. Hence, in modeling it we remove the property that the message lengths are equal $|M|$ from the above notion. This results in \emph{\newAnoaSAWOLLong}  ($\newAnoaSAWOL$).

\example{ Alice's demonstration is only at risk if the regime can link a message with a certain content to her as a sender with a non negligible probability.} 
Then at least \emph{\newSALong\  ($\newSA$)}, which is defined analogous to $\loopixSRTPU$ is needed. 

\exampleCont{However, $\newSA$ only allows Alice to claim that not she, but Charlie sent a critical message $m_a$ and the regime cannot know or guess better. Now assume that Dave is also communicating, then the regime might be able to distinguish Alice sending $m_a$, Charlie $m_c$ and Dave $m_d$ from Alice sending $m_d$, Charlie $m_a$ and Dave $m_c$. In this case, it might not even matter that Alice can claim that Charlie possibly sent her message. The fact that when comparing all three communications that possibly happened, Alice is more likely to have sent the critical message $m_a$ means a risk for her.}

 To circumvent this problem Alice needs a protocol that not only hides the difference between single pairs of users, but any number of users. Modeling this, instead of the complex property $M_{SM}$, we need to restrict that the active senders' sending frequencies are equal, i.e. $\bohliSWUc$.

\example{In another situation our activists already are prosecuted for being a sender while a message with critical content is sent. }

In this case at least  \emph{\newSAUOLong\  ($\newSAUO$)}, which is defined analogous to $\anoaREL$ is needed.

Analogous notions are defined for receivers.

\inlineheading{Sender Privacy Notions: Both-Side Message Unlinkability}
As explained with the example before in the case that Alice does not want any information about senders, receivers and messages or their combination to leak, she would use $\overline{O}$. However, the privacy in this example can be tuned down, if she assumes that the regime does not have certain external knowledge or that the users are accordingly careful. As explained for the Sender Notions with Receiver-Message Linkability before, in this case we might decide to allow $U', |U'|,Q',H',P'$ to leak.

If a notion $X \in \{\bohliRSAc, \bohliRSUPc, \bohliRWUPc, \bohliRPSc,\bohliRSUUc,\bohliRWUUc,\bohliRANc,\bohliRWUc,\bohliRWAc\}$ is extended to \emph{Sender Unobservability by X} \emph{($ \sgame{X}$)}, the leaking of the sender-message relation is removed.
This is done by removing $\EveryButRec$.
 Since the attacker now has a greater degree of freedom in choosing the senders and is (if at all) only restricted in how she chooses the receivers and messages, this is a special strong kind of Sender Unobservability.
 Analogous notions are defined for receivers.\footnote{Note that $\sgame{\bohliRSAc}=\rgame{\bohliSSAc}=\bohliSA$.}

%% file: sections/appendix/namingScheme.tex
\iflong 
\fi
\vspace{-2em}
\section{Additional Tables and Lists}
\vspace{-2.7em}
\iflong
\else

\label{app:summaryNamingScheme}

 \begin{table} [thb]
\center
\resizebox{0.48\textwidth}{!}{%
  \begin{tabular}{ c p{6cm}  }

Usage&Explanation\\ \hline
$D \in \{S,R,M\}$& Dimension $\in \{$Sender, Receiver, Message$\}$\\
Dimension $D$ not mentioned& Dimension can leak \\ 
Dimension $D$ mentioned &Protection focused on this dimension exists\\ \hline
$D \overline{O}$& not even the  active participating items regarding D leak,(e.g. $S\overline{O}$: not even $U$ leaks)\\
$DF \overline{L}$& active participating items regarding D can leak, but not which exists how often (e.g. $SF\overline{L}$: $U$leaks, but not $Q$)\\
$DM \overline{L}$& active participating items regarding D and how often they exist can leak ( e.g. $SM\overline{L}$: $U,Q$ leaks)\\ \hline
$X -Prop, $& like X but additionally Prop can leak\\
$Prop \in \{|U|,H,P,|U'|, H',P', |M| \}$&\\  \hline
$(D_1 D_2)\overline{O}$& uses $R_{D_1 D_2}$; active participating items regarding $D_1,D_2$ do not leak, (e.g. $(SR)\overline{O}$: $R_{SR}$)\\ 
$(D_1 D_2)\overline{L}$& uses $M_{D_1 D_2}$; active participating items regarding $D_1,D_2$ can leak, (e.g. $(SR)\overline{L}$: $M_{SR}$)\\ 
$(2D)\overline{L}$& uses $T_{D}$; one active participating item regarding $D$ has to be identified twice,  (e.g. $\anoaUL$: $T_{S}$)\\ \hline
$\overline{O}$&short for $S \overline{O} R\overline{O} M\overline{O} $\\
$\heviaUL$& short for $ M\overline{O}(SM\overline{L}, RM\overline{L})$\\
$S\overline{O}\{X\}$& short for $S\overline{O} M\overline{O} X$\\
$D_1 X_1[ D_2 X_2]$& $D_1$ is dominating dimension, usually $D_1$ has more freedom, i.e. $X_2$ is a weaker restriction than $X_1$ \\ \hline
$C\overline{O}$& nothing can leak (not even the existence of any communication)\\
\end{tabular}}
 \caption{Naming Scheme}
  \label{Tab:NamingScheme}
\end{table}

\label{app:NotionMapping}

\iflong
 \begin{table} [b!]
\center
\resizebox{0.4\textwidth}{!}{%
 \begin{tabular}{ c  c c }

Framework & Notion & Equivalent to \\ \hline
AnoA& $\alpha_{SA}$&$ {\manySess{\static{\corrStandard{\bohliSSAc}}}}_{CR_1}$\\
&$\alpha_{RA}$&${\manySess{\static{\corrStandard{\anoaRA}}}}_{CR_1}$\\
&$ \alpha_{REL}$&${\manySess{\static{\corrStandard{\anoaREL}}}}_{CR_2}$\\
&$ \alpha_{UL}$&${\manySess{\static{\corrStandard{\anoaUL}}}}_{CR_2}$\\
& $\alpha_{sSA}$&$\manySess{\static{\corrStandard{\bohliSSAc}}}$\\
&$\alpha_{sRA}$&$\manySess{\static{\corrStandard{\anoaRA}}}$\\
&$ \alpha_{sREL}$\footnotemark&$\manySess{\static{\corrStandard{\anoaREL}}}$\\
&$ \alpha_{sUL}$\footnotemark&$\manySess{\static{\corrStandard{\anoaUL}}}$\\ \hline
Bohli's&$S/SA=R/SA$&$\bohliSA$\\
&$ R/SUP$& $\bohliRSUP$\\
&$ R/WUP$& $\bohliRWUP$\\
& $ R/PS$&$\bohliRPS$\\
&$ R/SUU$& $\bohliRSUU$\\
&$ R/WUU$&$\bohliRWUU$\\
&$ R/AN$&$\bohliRAN$\\
&$ R/WU$&$\bohliRWU$\\
& $ R/WA$&$\bohliRWA$\\
&$ S/SA^\circ$& $\bohliSSAc$\\
&$ S/SUP^\circ$&$\bohliSSUPc$\\
&$ S/WUP^\circ$&$\bohliSWUPc$\\
&$S/PS^\circ$&$\bohliSPSc$\\
&$S/SUU^\circ$&$\bohliSSUUc$\\
& $S/WUU^\circ$&$\bohliSWUUc$\\
&$S/AN^\circ$&$\bohliSANc$\\
&$S/WU^\circ$&$\bohliSWUc$\\
&$S/WA^\circ$&$\bohliSWAc$\\
&$S/X, R/X^\circ$&analogous\\ 
&$X^+$&$\corrStandard{\langle X\rangle}$\\
& $X^*$&$\corrStandard{\langle X^\circ\rangle}$\\ \hline
Hevia's&$UO$&$\noCorr{\heviaUO}$, $k=1$\\
&$SRA$&$\noCorr{\bohliSA}$, $k=1$\\
&$SA^*$&$\noCorr{\bohliRWU}$, $k=1$\\
&$SA$&$\noCorr{\bohliSSAc}$, $k=1$\\
&$UL$&$\noCorr{\heviaUL}$, $k=1$\\
&$SUL$&$\noCorr{\bohliSWUc}$, $k=1$\\
&$RA^*, RUL, RA$&analogous\\ \hline
Gelernter's&$R^{H, \tau}_{SA}$&$\noCorr{\gelernterSA} \iff \noCorr{\bohliSPSc}$, $k=1$\\
&$R^{H, \tau}_{SUL}$&$\noCorr{\gelernterSUL} \iff \noCorr{\bohliSWAc}$, $k=1$\\
&$R_X$& analogous Hevia: $\langle X\rangle$ \\
&$R^H_X$& analogous Hevia: $\corrNoComm{\langle X\rangle}$\\
&$\hat{R}^H_X$&  analogous Hevia $\corrOnlyPartnerSender{\langle X\rangle}$\\
\end{tabular}}
 \caption{Equivalences, $\langle X \rangle$ equivalence of $X$ used}
  \label{mapping}
\end{table}

\addtocounter{footnote}{-1}
\footnotetext{Under the assumption that in all cases $m_0$ is communicated like in $\alpha_{REL}$ of   \cite{backes17anoa} and in $\alpha_{SREL}$ of one older AnoA version \cite{backes14anoa}.}
\stepcounter{footnote}
\footnotetext{Under the assumption that the receiver in stage 2 can be another than in stage 1 like in $\alpha_{UL}$ of   \cite{backes17anoa}.}

\else
\vspace{-1cm}
 \begin{table} [h!]
\center
\resizebox{0.35\textwidth}{!}{%
 \begin{tabular}{ c  c c }

Framework & Notion & Equivalent to \\ \hline
AnoA& $\alpha_{SA}$&$ {\bohliSSAc}$\\
&$\alpha_{RA}$&${\anoaRA}$\\
&$ \alpha_{REL}$&${\anoaREL}$\\
&$ \alpha_{UL}$&${\anoaUL}$\\ \hline
Bohli's&$S/SA=R/SA$&$\bohliSA$\\
&$ R/SUP$& $\bohliRSUP$\\
&$ R/WUP$& $\bohliRWUP$\\
& $ R/PS$&$\bohliRPS$\\
&$ R/SUU$& $\bohliRSUU$\\
&$ R/WUU$&$\bohliRWUU$\\
&$ R/AN$&$\bohliRAN$\\
&$ R/WU$&$\bohliRWU$\\
& $ R/WA$&$\bohliRWA$\\
&$ S/SA^\circ$& $\bohliSSAc$\\
&$ S/SUP^\circ$&$\bohliSSUPc$\\
&$ S/WUP^\circ$&$\bohliSWUPc$\\
&$S/PS^\circ$&$\bohliSPSc$\\
&$S/SUU^\circ$&$\bohliSSUUc$\\
& $S/WUU^\circ$&$\bohliSWUUc$\\
&$S/AN^\circ$&$\bohliSANc$\\
&$S/WU^\circ$&$\bohliSWUc$\\
&$S/WA^\circ$&$\bohliSWAc$\\
&$S/X, R/X^\circ$&analogous\\  \hline
Hevia's&$UO$&${\heviaUO}$\\
&$SRA$&${\bohliSA}$\\
&$SA^*$&${\bohliRWU}$\\
&$SA$&${\bohliSSAc}$\\
&$UL$&${\heviaUL}$\\
&$SUL$&${\bohliSWUc}$\\
&$RA^*, RUL, RA$&analogous\\ \hline
Gelernter's&$R^{H, \tau}_{SA}$&${\gelernterSA} \iff {\bohliSPSc}$\\
&$R^{H, \tau}_{SUL}$&${\gelernterSUL} \iff {\bohliSWAc}$\\
&$R_X$& analogous Hevia: $\langle X\rangle$ \\
\end{tabular}}
 \caption{Equivalences, $\langle X \rangle$ equivalence of $X$ used}
  \label{mapping}
\end{table}
\fi

 \fi

%% file: sections/appendix/symbolList.tex
\label{sec:allSymbols}
\begin{table} [h!]
\center
\resizebox{0.48\textwidth}{!}{%
  \begin{tabular}{ p{3cm} p{7.5cm} }

Symbol &Description\\ \hline
$U/U'$& Who sends/receives is equal for both scenarios.\\
$Q/Q'$&Which sender/receiver sends/receives how often  is equal for both scenarios.\\
$H/H'$& How many senders/receivers send/receive how often is equal for both scenarios.\\
$P/P'$&  Which messages are sent/received from the same sender/receiver  is equal for both scenarios.\\
$|U|/|U'|$& How many senders/receivers communicate is equal for both scenarios.\\
$|M|$& Messages in the two scenarios always have the same length.\\
$\EveryButSender$& Everything but the senders is identical in both scenarios.\\
$\EveryButRec, \EveryButMsg$& analogous\\
$\EveryButSenderMsg$&Everything but the senders and messages is identical in both scenarios.\\
$\EveryButReceiverMsg, \EveryButSenderRec$& analogous\\
\nothing&nothing will be checked; always true\\
$\EqualOrNothing$&If something is sent in both scenarios, the communication is the same.\\
\something & In every communication something must be sent.\\ \hline
$R_{SR}$& Adversary picks two sender-receiver-pairs. One of the senders and one of the receivers is chosen randomly.
For b=0 one of the adversary chosen sender-receiver pairs is drawn.
For b=1 the sender is paired with the receiver of the other pair.\\ 
$R_{SM}, R_{RM}$& analogous\\ \hline
$T_S$&Adversary picks two senders. The other sender might send the second time (stage 2).
For b=0 the same sender sends in both stages, for b=1 each sender sends in one of the stages.\\ 
$T_R$& analogous \\ \hline 
$M_{SR}$&Adversary picks two sender-receiver-pairs. Sender-receiver-pairs might be mixed.
For b=0 both adversary chosen sender-receiver-pairs communicate.
For b=1 both mixed sender-receiver-pairs communicate.\\ 
$M_{SM},M_{RM}$&analogous\\
\end{tabular}}
 \caption{Properties}
  \label{tab:allNotions}
\end{table}

\iflong
\vspace{-1cm}
\fi

\begin{table} [h!]
\center
\resizebox{0.48\textwidth}{!}{%
  \begin{tabular}{ p{5cm} p{5cm} }

Symbol &Description\\ \hline
$\mathcal{A}$&Adversary\\
$Ch$&Challenger\\
$\Pi$&\ac{ACN} protocol model\\
$b\in\{0,1\}$&Challenge bit\\
$g\in\{0,1\}$&Adversary's guess\\
$\underline{r}_0=(r_{0_1}, r_{0_2}, \dots, r_{0_l})$&Batch of communications\\ 
$r_{b_i}\in\{ \Diamond, (u,u',m,aux)\}$&Communication\\ 
$\Diamond$&Nothing is communicated\\
$(u,u',m,aux)$& $m$ is sent from $u $ to $u'$ with auxiliary information $aux$\\
$(\underline{r}_{0_1}, \dots,\underline{r}_{0_k})$& (First) Scenario\\
$\perp$&Abort game\\
\iflong
$CR(\underline{r}_0,\underline{r}_1)$ &Challenge rows of batches $\underline{r}_0,\underline{r}_1$\\
\fi
\iflong
$\Psi$&Challenge Number\\
$\#cr$&Number of challenge rows allowed in challenge\\
$n$&Number of challenges allowed\\
$c$&Number of challenge rows allowed in game\\
$\hat{U}$&Set of corrupted users\\
\fi
$\mathcal{U}$&Set of possible senders\\
$\mathcal{U'}$&Set of possible receivers\\
\end{tabular}}
 \caption{Symbols used in the Game}
  \label{tab:otherSymbols}
\end{table}

\begin{table}[h!]
\center
\resizebox{0.48\textwidth}{!}{%
  \begin{tabular}{ p{2.5cm} p{8cm} }

Symbol &Description\\ \hline
$\bohliRSUP $& \bohliRSUPLong \\
$\bohliRWUP$ & \bohliRWUPLong \\
$\bohliRPS $& \bohliRPSLong\\
$\bohliRSUU $&\bohliRSUULong\\
$\bohliRWUU$ &\bohliRWUULong\\
$\bohliRAN $&\bohliRANLong\\
$\bohliRWU $&\bohliRWULong\\
\mbox{$\bohliRWA $}&\bohliRWALong\\
$\bohliSSAc$& \bohliSSAcLong \\
$\bohliSSUPc $& \bohliSSUPcLong \\
$\bohliSWUPc$ & \bohliSWUPcLong \\
$\bohliSPSc $& \bohliSPScLong \\
$\bohliSSUUc $&\bohliSSUUcLong \\
$\bohliSWUUc$ &\bohliSWUUcLong \\
$\bohliSANc $&\bohliSANcLong \\
$\bohliSWUc $&\bohliSWUcLong \\
$\bohliSWAc $&\bohliSWAcLong \\
$\newAnoaSA $&\newAnoaSALong \\
$\anoaUL$ &\anoaULLong\\
$\newSAUO$ &\newSAUOLong\\
$\newSA$& \newSALong\\
$\loopixSUONe$& \loopixSUONeLong \\
Receiver notions&analogous\\
$\heviaUO$&\heviaUOLong\\
$\bohliSA$& \bohliSALong \\
$\anoaREL$ &\anoaRELLong\\
$\heviaUL$& \heviaULLong \\
$\newCONF$ &\newCONFLong\\
$ \loopixSRTPU$& \loopixSRTPULong \\
\iflong \hline
$X$&notion, standard assumptions\\ 
$X$&adaptive corruption\\ 
$\noCorr{X}$ & No corruption of users is allowed.\\
$\static{X}$ & Only static corruption of users is allowed.\\ 
$\corrNoComm{X}$&Corrupted users are not allowed to be chosen as senders or receivers.\\
$\corrOnlyPartnerSender{X}$&Corrupted users are not  allowed to be senders.\\ 
$\corrOnlyPartnerReceiver{X}$&Corrupted users are not  allowed to be receivers.\\
$\corrStandard{X}$&Corrupted nodes send/receive identical messages in both scenarios.\\
$X$&Communication of corrupted users not restricted.\\
$\manySess{X}$ &Sender and  receiver of challenge rows stay the same for this challenge and stage.\\
$X$&Sessions are not restricted.\\ 
$\challengeRows{X}$& $c$ communications in the two scenarios are allowed to differ.\\ 
\fi
\end{tabular}}
 \caption{Notions \iflong and Restriction Options\fi}
  \label{tab:allSymbols}
\end{table}

\iflong 
 \fi

%% file: sections/proofs/completeSketch.tex
\begin{proofsketch}[Proof Sketch (continued)]

Tables \ref{Construction Idea}  gives the idea of some proofs.  $|U'|$ means  the number of receivers is leaked. The other abbreviations are used analogously. 
The attack is shortened to the format $\langle$(communications of instance 0 scenario 0),(communications of instance 1 of scenario 0)$\rangle$,$\langle$(communications of instance 0 scenario 1),(communications of instance 1 of scenario 1)$\rangle$ (if both instances of the scenario are equal, we shorten to:$\langle$(communications of instance 0 scenario 0)$\rangle$,$\langle$(communications of instance 0 scenario 1)$\rangle$ ) and all not mentioned elements  are equal in both scenarios. 
$m_0,m_1, m_2, m_3$ are messages with $|m_0|<|m_1|$, $|m_2|=|m_3|$ and $m_0 \neq m_1 \neq m_2\neq m_3$; $u_0,u_1,u_2$ senders and $u'_0,u'_1,u'_2$ receivers. 

\begin{table} [thb]
  \center
\resizebox{0.48\textwidth}{!}{
  \begin{tabular}{p{2.6cm} p{2.6cm} p{1cm} p{5.7cm} }

 $X_1$&$X_2$&$I$ & attack\\ \hline
$\bohliRSUP$&$ \newAnoaUL$& $|U'|$& $\langle((u'_0,m_0),switchStage,(u'_0,m_0)),$ $((u'_1,m_0),switchStage,(u'_1,m_0))\rangle,$ $\langle((u'_0,m_0),switchStage,(u'_1,m_0)),$ $((u'_1,m_0),switchStage,(u'_0,m_0))\rangle$\\
\hline
$\bohliRPS$&$\newCONF$& $m$&$((m_2)),$ $((m_3))$\\ 
$\bohliRPS$ &$ \newRAUO$&$|U'|, m$& $\langle((u'_0,m_0), (u'_0,m_2)),$ $((u'_0,m_0),(u'_1,m_3))\rangle,$ $\langle((u'_0,m_0), (u'_0,m_3)),$ $((u'_0,m_0),(u'_1,m_2))\rangle$\\ 
$\bohliRPS$&$\newRA$&$P'$&$\langle((u'_0,m_2), (u'_0,m_0),(u'_1,m_1)),$ $((u'_0,m_2), (u'_1,m_1),(u'_0,m_0))\rangle,$ $\langle ((u'_0,m_2),(u'_0,m_1),(u'_1,m_0)),$ $ ((u'_0,m_2),(u'_1,m_0),(u'_0,m_1))\rangle$\\
\end{tabular}}
 \caption{Some counter example ideas with $X'_1=X_1$}
  \label{Construction Idea}
\end{table}
\vspace{-1em}
\end{proofsketch}

%% file: sections/proofs/otherFrameworksSketch.tex
%

We define new notions  as $\gelernterSA$=\nothing $\land G$ and  $\gelernterSUL$= \nothing $\land Q \land G$ that are equivalent to some of the already introduced notions to make the mapping to the Gelernter's notions obvious.
 They use a new property $G$, in which scenarios are only allowed to differ in the sender names.
 \begin{definition}[Property $G$]
 Let $\mathcal{U}$ be the set of all possible senders,  $L_{b_i}$  the sender-message linking for scenario $b\in \{0,1\}$.
We say that $G$ is met, iff a permutation $perm$ on $\mathcal{U}$ exists such that for all $ (u,M)\in L_{0_k}: (\text{perm}(u),M) \in L_{1_k}$.
 \end{definition}

\begin{theorem}It holds that
\begin{align*}
\iflong(\epsilon, \delta)-\fi\gelernterSA &\iff \iflong(\epsilon, \delta)-\fi\bohliSPSc,\\
\iflong(\epsilon, \delta)-\fi\gelernterSUL &\iff \iflong(\epsilon, \delta)-\fi\bohliSWAc\text{.}
\end{align*}
\end{theorem}

\begin{proofsketch}
Analogous to Theorem \ref{the:impl}:
\iflong
$\gelernterSA \Rightarrow \bohliSPSc$: Every attack on $\bohliSPSc$ is valid against $\gelernterSA$: Since $P$ is fulfilled, for every sender  $u_0$ in the first scenario, there exists a sender $u_0'$ in the second scenario sending the same messages.
Hence, the permutation between senders of the first and second scenario exists.

	$\gelernterSA \Leftarrow \bohliSPSc$: Every attack on $\gelernterSA$ is valid against $\bohliSPSc$: Since there exists a permutation between the senders of the first and second scenario sending the same messages, the partitions of messages sent by the same sender are equal in both scenarios, i.e.
$P$ is fulfilled.

$\gelernterSUL \iff \bohliSWAc:$
	 $Q$ is required in both notions by definition.
Arguing that $P$ resp.
$G$ is fulfilled given the other attack is analogous to $\gelernterSA \iff \bohliSPSc$.
\else
See long version for details.
\fi
\end{proofsketch}

%% file: sections/proofs/Loopix1Sketch.tex
 \begin{table} 
\center
\resizebox{0.35\textwidth}{!}{%
  \begin{tabular}{ c c c }

Notion&Name&\iflong Aspects\else Properties \fi \\ \hline
$\loopixSUO$ &Loopix's Sender Unobservavility& $ \EqualOrNothing$\\ 
$\loopixRUO$ & Loopix's Receiver Unobservability&$\EqualOrNothing$\\
$\loopixSUONe$ &Restricted Sender Unobservability& $\centernot \rightarrow \land \EveryButSender$  \\
$\loopixRUONe$ &Restricted Receiver Unobservability& $\centernot \rightarrow' \land \EveryButRec$ \\
\end{tabular}}
 \caption{Definition of the Loopix notions}
  \label{NotionsLoopix}
\end{table}

We define $\loopixSUO$ and $\loopixRUO$ according to Table \ref{NotionsLoopix}. Therefore, we need the property that if something is sent in both scenarios, it is the same.

 \begin{definition}[$\EqualOrNothing$]
 Let the checked batches be $\underline{r_0},\underline{r_1}$, including communications ${r_0}_j, {r_1}_j, j \in \{1, \dots l\} $. We say $\EqualOrNothing$  is met, iff for all $j \in \{1, \dots l\}$:
\[    \EqualOrNothing: {r_0}_j ={r_1}_j \lor {r_0}_j =\Diamond \lor {r_1}_j = \Diamond \]
 \end{definition} 

\begin{theorem}
\label{loopixUO}
It holds that
$\iflong(\epsilon, \delta)-\fi\heviaUO \iff \iflong(\epsilon, \delta)-\fi \loopixSUO\text{.}$
\end{theorem}

\begin{proofsketch}
Implications are proven analogously to the ones in Theorem \ref{the:impl}.
$\heviaUO \implies \loopixSUO$  by definition.
$\loopixSUO \implies \heviaUO$ because for every challenge row $(r_0,r_1)$ in the attack on $\heviaUO$, we can create two batches $(r_0, \Diamond  )$ and $( \Diamond , r_1)$.
\end{proofsketch}

%% file: sections/proofs/Loopix2Sketch.tex
We define $\loopixSUONe$ and $\loopixRUONe$ according to Table \ref{NotionsLoopix}.
To formulate these notions we need a new property that some sender/receiver is not participating in any communication in the second scenario:

\begin{definition}[Property $\centernot\rightarrow$]
Let $u$ be the sender of the first scenario in the first challenge row of this challenge.
We say that $\centernot \rightarrow$ is fulfilled iff  for all $j:$ $  u_{1_j}\neq u$.
(Property $\centernot \rightarrow'$ is defined analogously for receivers.) 
\end{definition}

\begin{theorem}
\label{loopixUO2}
It holds that $\bohliSSAc \iff \loopixSUONe$.
 
\end{theorem}
\begin{proofsketch}
Analogously to Theorem \ref{the:impl}.
\iflong
\todo[inline]{Here are n's that do not belong here}
$\Rightarrow$: Every attack on $\loopixSUONe$  is by definition a valid attack on $\bohliSSAc$.

$\Leftarrow$: Given an attack $\mathcal{A}$ on $(c,n , 0, 2\delta)-\bohliSSAc$ where both scenarios of a challenge use all users (otherwise it would be a valid attack on $\loopixSUONe$).
Let $(r_{2_1}, \dots r_{2_l})$ be the same batch as the second of  $\mathcal{A}$ except that the senders of the first challenge row are always replaced with an arbitrary other sender that was not used in the first challenge row.
Let 
$P(0|2)$ be the probability that $\mathcal{A}$ outputs 0, when the new batches are run, $P(0|0)$ when the first scenario of $\mathcal{A}$ is run and $P(0|1)$ when the second is run.
In the worst case for the attacker $P(0|2)=\frac{P(0|0)+P(0|1)}{2}$ (otherwise we would replace the scenario $b$ where $|P(0|2)-P(0|b)|$ is minimal with the new one and get better parameters in the following calculation).
Since $\mathcal{A}$ is an attack on $(c,n , 0, 2\delta)-\bohliSSAc$, $P(0|0)>  P(0|1) + n 2\delta$.
Transposing and inserting the worst case for $P(0|2)$ leads to: $ P(0|0) > 2P(0|2)- P(0|0)+ 2n \delta \iff  P(0|0)> P(0|2) + n \delta $.
Hence, using $\mathcal{A}$ with the adapted scenario is an attack on  $(c,n, 0, \delta)- \loopixSUONe$\footnote{An analogous argumentation works for $(c, \epsilon- \ln(0.5), \delta)-\bohliSSAc \Leftarrow (c, \epsilon, \delta)- \loopixSUONe$.}.
\else
See long version for details.
\fi
\end{proofsketch}

\todo[inline]{Replace long version placeholder!}

%% file: sections/appendix/HierarchyAndTables.tex
\section{Hierarchy And Tables}
\label{sec:HierarchyAndTables}

On the next page the hierarchy can be found combined with the symbol tables (Fig. \ref{fig:hierarchyColored} and Tables of \ref{tab:combinationOfTables}). Further, Fig. \ref{fig:hierarchyInefficient} and  Fig. \ref{fig:hierarchyAnnotated}  highlight special parts of the hierarchy and Fig. \ref{fig:hierarchyold} presents the mapping of the notions of the other frameworks to ours.
\begin{figure*}[h!]
\center
\includegraphics[width=0.8\textwidth]{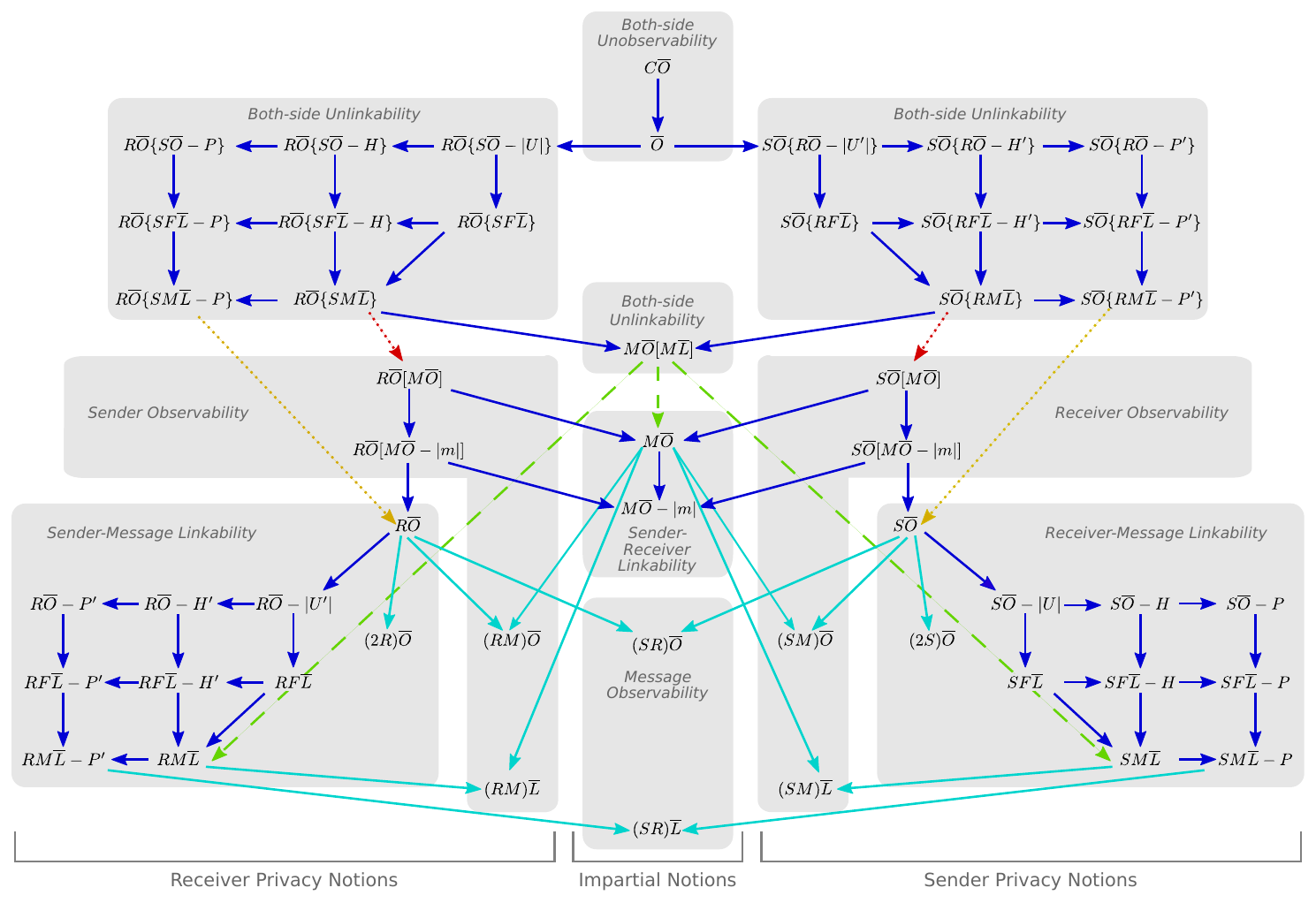}
\caption{Our new hierarchy of privacy notions divided into sender, receiver and impartial notions and clustered by leakage type.}\label{fig:hierarchyColored}
\end{figure*}

\begin{table*}
\centering
\begin{subtable}[t]{0.3\linewidth}
\resizebox{\linewidth}{!}{%
\begin{tabular}{ p{2.5cm} p{5.32cm} }
Usage&Explanation\\
$D \in \{S,R,M\}$& Dimension $\in \{$Sender, Receiver, Message$\}$\\
Dimension $D$\linebreak \emph{not} mentioned& Dimension can leak \\
Dimension $D$\linebreak mentioned &Protection focused on this dimension exists\\ \hline
$D \overline{O}$& not even the participating items regarding D leak,(e.g. $S\overline{O}$: not even $U$ leaks)\\
$DF \overline{L}$& participating items regarding D can leak, but not which exists how often (e.g. $SF\overline{L}$: $U$leaks, but not $Q$)\\
$DM \overline{L}$& participating items regarding D and how often they exist can leak ( e.g. $SM\overline{L}$: $U,Q$ leaks)\\ \hline
$X -Prop, $& like X but additionally Prop can leak\\
$Prop \in \{|U|,H,\allowbreak P,|U'|, H',P', |M| \}$&\\  \hline
$(D_1 D_2)\overline{O}$& uses $R_{D_1 D_2}$; participating items regarding $D_1,D_2$ do not leak, (e.g. $(SR)\overline{O}$: $R_{SR}$)\\ 
$(D_1 D_2)\overline{L}$& uses $M_{D_1 D_2}$; participating items regarding $D_1,D_2$ can leak, (e.g. $(SR)\overline{L}$: $M_{SR}$)\\ 
$(2D)\overline{L}$& uses $T_{D}$; it can leak whether two participating item regarding $D$ are the same,  (e.g. $\anoaUL$: $T_{S}$)\\ \hline
$\overline{O}$&short for $S \overline{O} R\overline{O} M\overline{O} $\\
$\heviaUL$& short for $ M\overline{O}(SM\overline{L}, RM\overline{L})$\\
$S\overline{O}\{X\}$& short for $S\overline{O} M\overline{O} X$\\
$D_1 X_1[ D_2 X_2]$& $D_1$ is dominating dimension, usually $D_1$ has more freedom, i.e. $X_2$ is a weaker restriction than $X_1$ \\ \hline
$C\overline{O}$& nothing can leak (not even the existence of any communication)\\
\end{tabular}}
\subcaption{Naming Scheme}
\label{Tab:NamingScheme}
\end{subtable}
\hfill
\begin{subtable}[t]{0.2295\linewidth}
\resizebox{\linewidth}{!}{%
  \begin{tabular}{ p{2cm} p{3.33cm} }
Notion&Properties \\
$\loopixSRTPU$ &\something $ \land \EveryButSenderRec \land M_{SR}$  \\
$\anoaREL$ &\something$\land \EveryButSenderRec \land R_{SR}$ \\
$\newCONFWOL$ &\something $ \land \EveryButMsg$ \\
$\newCONF$ &\something $ \land \EveryButMsg \land |M| $ \\
$\heviaUL$&  \something$\land Q \land Q'$\\
$\bohliSA$& \something \\
$\heviaUO$& \nothing\\ \hline
$\bohliSSAc$& \something $\land \EveryButSender$\\
$\bohliSSUPc $& \something$\land \EveryButSender \land |U|$\\
$\bohliSWUPc$ &\something$\land \EveryButSender \land H$\\
$\bohliSPSc $&\something $\land \EveryButSender \land P$\\
$\bohliSSUUc$ &\something$\land \EveryButSender \land U$ \\
$\bohliSWUUc$ &\something $\land \EveryButSender \land U \land H$\\
$\bohliSANc $&\something $\land \EveryButSender \land U \land P$\\
$\bohliSWUc$ &\something $\land \EveryButSender \land Q$\\
$\bohliSWAc$ &\something $\land \EveryButSender \land Q \land P$\\
$\anoaUL$ &\something$\land \EveryButSender \land T_S $ \\
$\bohliRSAc$ \ etc.&analogous\\ \hline
\mbox{$\newAnoaSAWOL$}&\something $ \land \EveryButSenderMsg $ \\
\mbox{$\newAnoaSA$}&\something $ \land \EveryButSenderMsg \land |M| $ \\
$\newSAUO$ &\something $  \land \EveryButSenderMsg \land R_{SM}$ \\
$\newSA$&\something $  \land \EveryButSenderMsg \land M_{SM} $ \\
\mbox{$\anoaRA$} \ etc.&analogous \\ \hline
$\sgame{X'}$& Properties of $X'$, remove $\EveryButRec$ \\
\multicolumn{2}{p{4.5cm}}{for $X' \in \{\bohliRSAc,$ $ \bohliRSUPc,$ $\bohliRWUPc,$ $\bohliRPSc,$ $\bohliRSUUc,\bohliRWUUc,$ $\bohliRANc,$ $\bohliRWUc,$ $\bohliRWAc\}$}\\
$\rgame{X'}$ &analogous, remove $\EveryButSender$ \\ 
\end{tabular}}
\subcaption{Definition of the notions\iflong ~for all corruption options as defined in Table~\ref{tab:corruption options}\fi. A description of simple properties was given in Table~\ref{tab:information}. }
\label{NotionsDefinition2}
\end{subtable}
\hfill
\begin{subtable}[t]{0.4495\linewidth}
\resizebox{\linewidth}{!}{%
\begin{tabular}{ p{1.5cm} p{9cm}}
Symbol &Description\\ \hline
$U/U'$& Who sends/receives is equal for both scenarios.\\
$Q/Q'$&Which sender/receiver sends/receives how often  is equal for both scenarios.\\
$H/H'$& How many senders/receivers send/receive how often is equal for both scenarios.\\
$P/P'$&  Which messages are sent/received from the same sender/receiver  is equal for both scenarios.\\
$|U|/|U'|$& How many senders/receivers communicate is equal for both scenarios.\\
$|M|$& Messages in the two scenarios always have the same length.\\
$\EveryButSender$& Everything but the senders is identical in both scenarios.\\
$\EveryButRec, \EveryButMsg$& analogous\\
$\EveryButSenderMsg$&Everything but the senders and messages is identical in both scenarios.\\
$\EveryButReceiverMsg, \EveryButSenderRec$& analogous\\
\nothing&nothing will be checked; always true\\
$\EqualOrNothing$&If something is sent in both scenarios, the communication is the same.\\
\something & In every communication something must be sent.\\ \hline
$R_{SR}$& Adversary picks two sender-receiver-pairs. One of the senders and one of the receivers is chosen randomly.
For b=0 one of the adversary chosen sender-receiver pairs is drawn.
For b=1 the sender is paired with the receiver of the other pair.\\ 
$R_{SM}, R_{RM}$& analogous\\ \hline
$T_S$&Adversary picks two senders. The other sender might send the second time (stage 2).
For b=0 the same sender sends in both stages, for b=1 each sender sends in one of the stages.\\ 
$T_R$& analogous \\ \hline 
$M_{SR}$&Adversary picks two sender-receiver-pairs. Sender-receiver-pairs might be mixed.
For b=0 both adversary chosen sender-receiver-pairs communicate.
For b=1 both mixed sender-receiver-pairs communicate.\\ 
$M_{SM},M_{RM}$&analogous\\
\end{tabular}}
\subcaption{Properties}
\label{tab:allNotions2}
\end{subtable}
\caption{Tables for our naming scheme~(\subref{Tab:NamingScheme}), notion definitions~(\subref{NotionsDefinition2}) and property definitions~(\subref{tab:allNotions2})}
\label{tab:combinationOfTables}
\end{table*}

\begin{figure*}
\center
\includegraphics[width=0.8\textwidth]{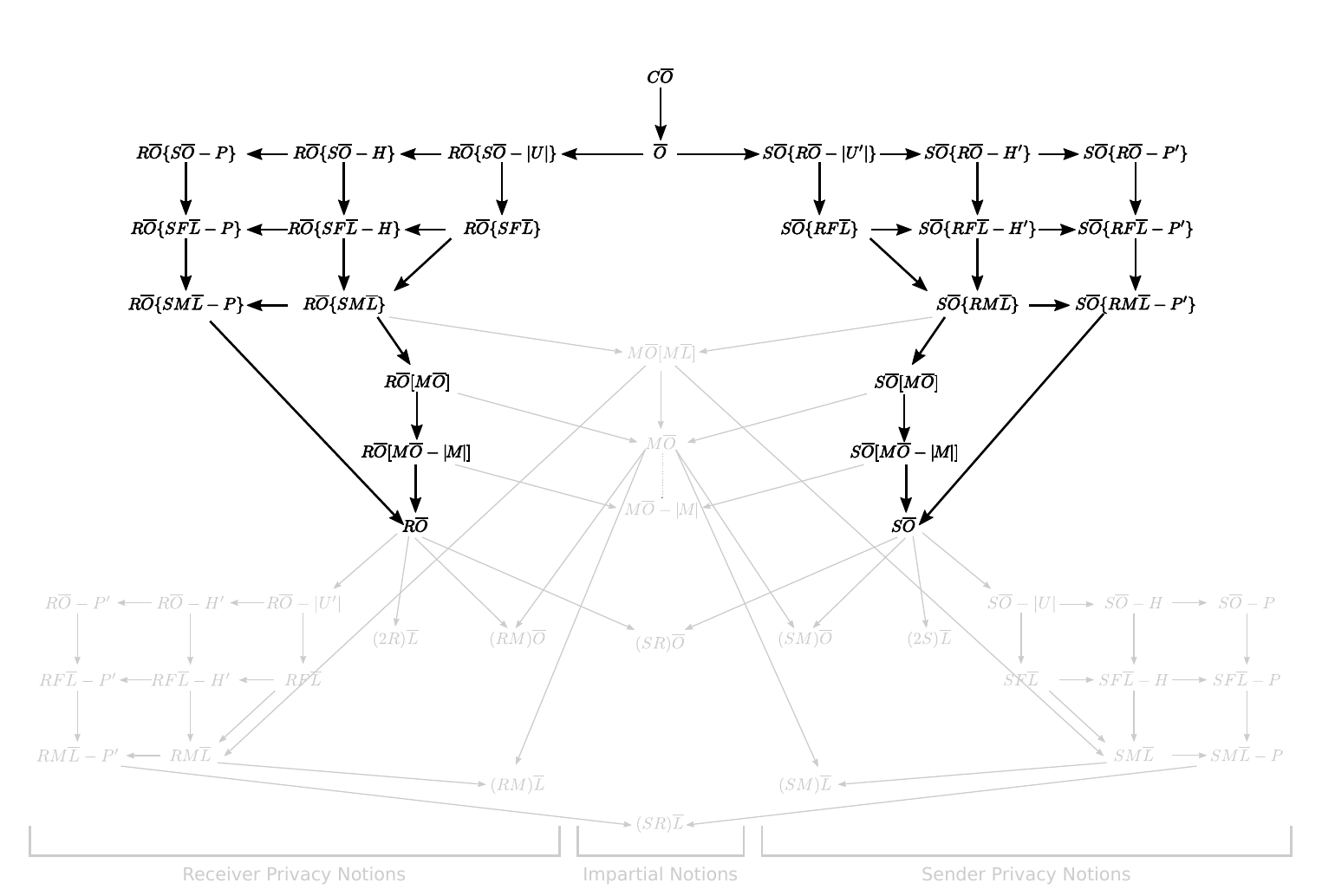}
\caption{Protocols for notions highlighted are inefficient due to a result by Gelernter~\cite{gelernter13limits}.}\label{fig:hierarchyInefficient}
\end{figure*}

\begin{figure*}
\center
\includegraphics[width=0.8\textwidth]{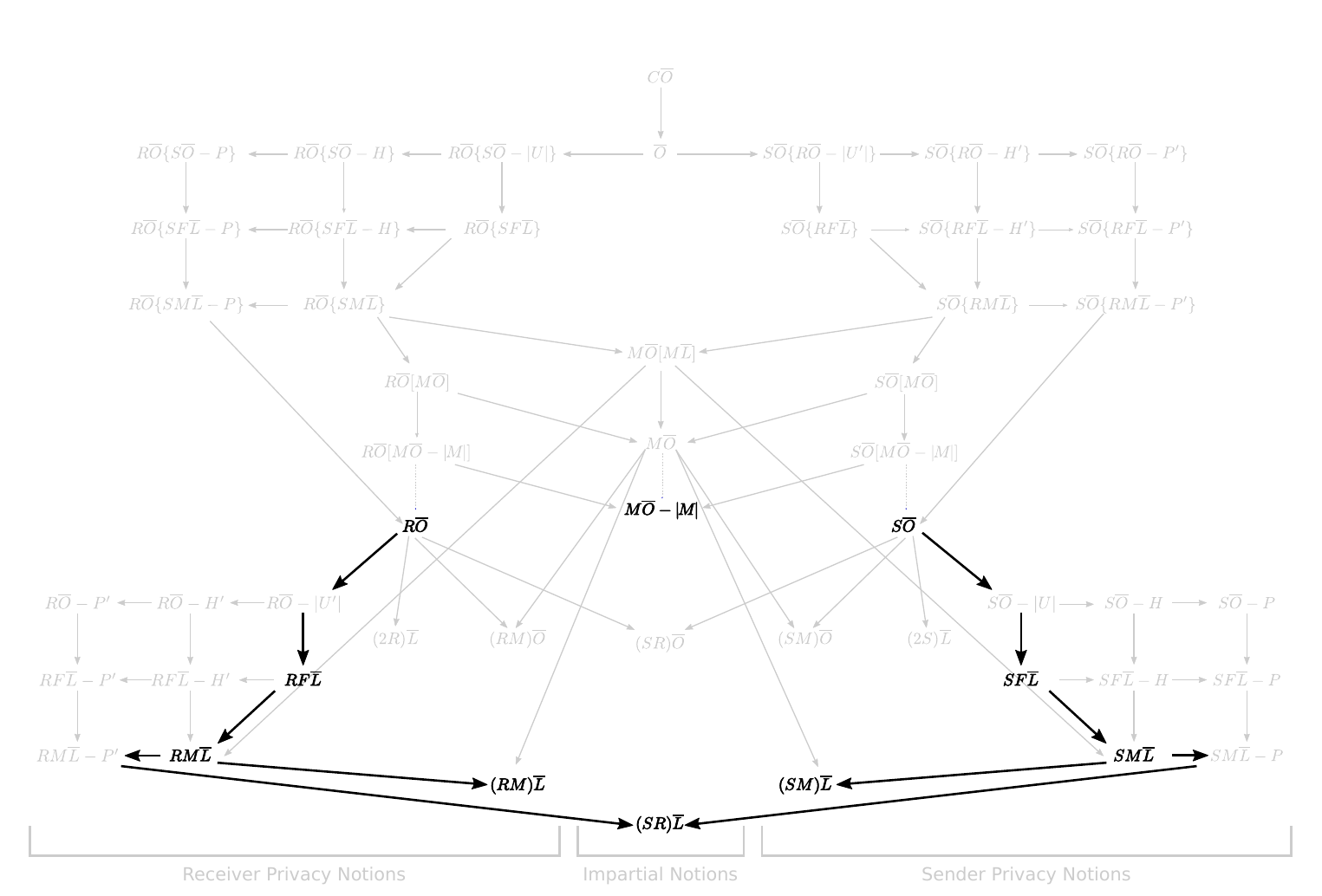}
\caption{Depicted notions are a first guess on which notions might be important based on informal and formal usage in the related work.}\label{fig:hierarchyAnnotated} 
\end{figure*}

%% file: sections/appendix/graphicOtherFrameworks.tex

\iflong
\begin{figure*}[t!] 
  \center
  \includegraphics[width=0.95\textwidth]{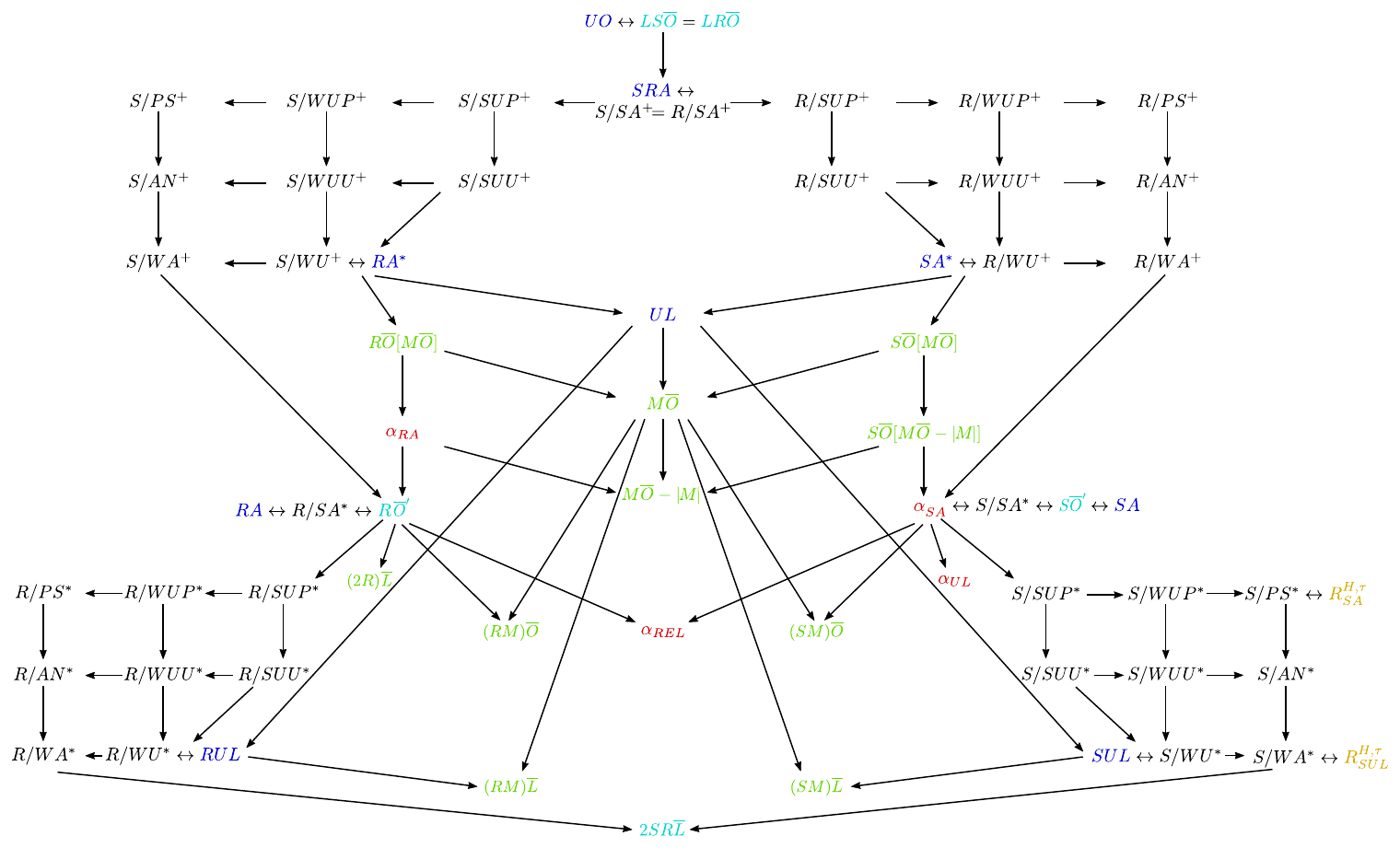}

\caption{\mbox{Our hierarchy with the mapping to other works (Bohli's, \textcolor{red}{AnoA}, \textcolor{blue}{Hevia's}, \textcolor{yellow}{Gelernter's framework}, \textcolor{cyan}{Loopix's \ac{ACN}} and  \textcolor{green}{new notions})}}

 \label{fig:hierarchyold}
\end{figure*}
\fi
%
\label{hierarchyFrameworks}